\documentclass[sigconf,nonacm,pdfa]{acmart}

\newif\ifarxiv
\arxivtrue

\ifarxiv
    \newcommand{\new}[1]{#1}
    \newcommand{\apx}[1]{Appendix~\ref{#1}}
    \newcommand{\Apx}[1]{Appendix~\ref{#1}}
\else
    \newcommand{\new}[1]{#1}
    \newcommand{\apx}[1]{our technical report~\cite{DivaArxiv}}
    \newcommand{\Apx}[1]{Our technical report~\cite{DivaArxiv}}
\fi

\newcommand\vldbdoi{10.14778/3749646.3749664}
\newcommand\vldbpages{3923 - 3936}
\newcommand\vldbvolume{18}
\newcommand\vldbissue{11}
\newcommand\vldbyear{2025}
\newcommand\vldbauthors{\authors}
\newcommand\vldbtitle{\shorttitle} 
\newcommand\vldbavailabilityurl{https://github.com/n3slami/Diva.git}
\newcommand\vldbpagestyle{empty} 

\usepackage{array}
\usepackage{breqn}
\usepackage{booktabs,makecell}
\usepackage{enumitem}
\usepackage[normalem]{ulem}
\usepackage{tikz,pgfplots}
\usetikzlibrary{positioning}
\usetikzlibrary{arrows}
\usetikzlibrary{patterns}
\usetikzlibrary{shapes.geometric}
\usetikzlibrary{decorations.pathreplacing}
\usetikzlibrary{calligraphy}
\usetikzlibrary{hobby}
\usetikzlibrary{calc}

\usepackage{pifont}
\newcommand{\cmark}{\ding{51}}

\usepackage{thm-restate}
\newtheorem{theorem}{Theorem}[section]

\newtheorem{lemma}[theorem]{Lemma}
\newtheorem{claim}[theorem]{Claim}

\usepackage{colorprofiles}
\usepackage[a-2b,mathxmp]{pdfx}[2018/12/22]

\graphicspath{{./figures/}}

\usepackage{xspace}
\newcommand{\ST}{S{\nobreakdash-}Trie\xspace}
\newcommand{\DT}{D{\nobreakdash-}Trie\xspace}
\newcommand{\DTs}{{\DT}s\xspace}

\begin{document}
\title{Diva++: Dynamic Range Filtering over Hard Workloads}

\author{Navid Eslami}
\email{navideslami@cs.toronto.edu}
\affiliation{%
  \institution{University of Toronto}
  \city{Toronto}
  \country{Canada}
}

\author{Ioana O. Bercea}
\email{bercea@kth.se}
\affiliation{%
  \institution{KTH Royal Institute of Technology}
  \city{Stockholm}
  \country{Sweden}
}

\author{Niv Dayan}
\email{nivdayan@cs.toronto.edu}
\affiliation{%
  \institution{University of Toronto}
  \city{Toronto}
  \country{Canada}
}

\begin{abstract}
    Range filters are compact probabilistic data structures that answer
    approximate range emptiness queries. They are used in many domains, e.g.,
    in key-value stores, to quickly rule out the existence of keys in a given
    query range and avoid searching for them in storage. However, all existing
    range filters exhibit at least one of three shortcomings: (1)~they do not
    provide any false positive rate or performance guarantees, (2)~they do not
    support variable-length keys and query ranges, and (3)~they do not allow
    dynamic updates.

    We introduce Diva, the first range filter to address all the above
    challenges simultaneously. Diva learns the dataset's distribution by
    sampling keys and storing them in a cache-efficient trie. It compresses
    keys in-between samples by removing their longest common prefix and
    truncating their suffixes while leaving enough bits in the middle (i.e., an
    infix) to differentiate the keys in sorted order. It stores infixes in
    constant-time dynamic data blocks, which it splits to handle insertions and
    expansions. It processes a range query by traversing the trie and checking
    for the inclusion of infixes in the target query range.

    We mathematically prove that Diva provides the best possible trade-off
    between memory and false positive rate on many common real-world data
    distributions. We extend these benefits to a wider range of real-world
    workloads by introducing Diva++, an enhanced Diva variant. Diva++ saves
    memory by removing redundancies among infixes using order-preserving
    entropy encoding. It then removes any remaining identical infixes and uses
    the freed space to store more bits of the original keys \mbox{within
    compact binary tries}.

    We compare Diva and Diva++ to all prior range filters, and show that they
    achieve a false positive rate on par with the state of the art on
    real-world datasets while supporting dynamicity and variable-length queries
    and keys.
\end{abstract}

\maketitle

\pagestyle{\vldbpagestyle}
\begingroup\small\noindent\raggedright\textbf{PVLDB Reference Format:}\\
\vldbauthors. \vldbtitle. PVLDB, \vldbvolume(\vldbissue): \vldbpages, \vldbyear.\\
\href{https://doi.org/\vldbdoi}{doi:\vldbdoi}
\endgroup
\begingroup
\renewcommand\thefootnote{}\footnote{\noindent
This work is licensed under the Creative Commons BY-NC-ND 4.0 International
License. Visit \url{https://creativecommons.org/licenses/by-nc-nd/4.0/} to view
a copy of this license. For any use beyond those covered by this license,
obtain permission by emailing \href{mailto:info@vldb.org}{info@vldb.org}.
Copyright is held by the owner/author(s). Publication rights licensed to the
VLDB Endowment. \\
\raggedright Proceedings of the VLDB Endowment, Vol. \vldbvolume, No. \vldbissue\ %
ISSN 2150-8097. \\
\href{https://doi.org/\vldbdoi}{doi:\vldbdoi} \\
}\addtocounter{footnote}{-1}\endgroup

\ifdefempty{\vldbavailabilityurl}{}{
\vspace{.3cm}
\begingroup\small\noindent\raggedright\textbf{PVLDB Artifact Availability:}\\
The source code, data, and/or other artifacts have been made available at
\url{\vldbavailabilityurl}.
\endgroup
}

\section{Introduction}~\label{sec:introduction}
\textbf{What is a Filter?}
A filter is a memory-efficient probabilistic data structure that answers
whether a query key exists in a given set. Its compactness allows it to fit in
a higher level of the memory hierarchy than the set it represents, making it
fast to query. A filter never returns a false negative but may return a false
positive with a probability called the \emph{false positive rate~(FPR)} that
depends on its memory footprint. Due to these qualities, filters are used in
many applications to avoid redundant disk reads~\cite{RocksDB} and network
hops~\cite{FiltersNetworking} when a query happens to target a nonexistent key.

\textbf{Range Filters and Applications.}
Traditional filters answer membership queries for a single key. In contrast, a
range filter accepts two keys as the end-points of a range and determines if at
least one key in the set is in-between them. As with a traditional filter, a
range filter does not return false negatives but may return false positives.
Range filters are used in many domains, especially in the context of
LSM-Trees~\cite{LSMTree}, to avoid accessing files that do not contain any keys
within a range predicate, significantly boosting
performance~\cite{AdaptiveRangeFilter,SuRF,Rosetta,bloomRF,Proteus,SNARF,Oasis,Memento}.
They also aid in preventing redundant I/Os to SQL
tables~\cite{AdaptiveRangeFilter} and B-Tree indices~\cite{Memento}, and they
have potential applications in social web analytics~\cite{YCSB} and replication
in distributed key-value stores~\cite{Rose}.

\textbf{Range Filtering Goals.}
The ideal general-purpose range filter must simultaneously support (G1)~the
lowest possible FPR under a stringent memory budget, (G2)~range queries of any
length, (G3)~variable-length keys, (G4)~dynamic modification operations such as
insertions, deletions, expansions, and contractions, and the best possible
(G5)~query and (G6)~construction performance.

\textbf{Design Contentions.}
Every existing range filter only fulfills a subset of these
goals~\cite{AdaptiveRangeFilter,SuRF,Rosetta,REncoder,REncoder_Journal,bloomRF,Proteus,SNARF,Oasis,Grafite,Memento,Aeris}.
In fact, almost none of the existing ones attain~(G3) or
(G4)~\cite{Rosetta,REncoder,REncoder_Journal,bloomRF,Proteus,SNARF,Oasis,Grafite}.
Moreover, Goswami et al. have proven an information-theoretic lower bound on
the memory footprint of range filters~\cite{Goswami}, stating that it is
impossible for any range filter to achieve~(G1) and (G2) at the same time.
However, this lower bound only holds for worst-case
\new{workloads}, meaning that it may still be possible to
achieve~(G1) and (G2) for ``common'' \new{workloads}. As such, we
pose the following research question: \emph{Is it possible to design a range
filter that simultaneously fulfills all six goals for common
\new{workloads}?} This paper presents an affirmative answer.

\textbf{Core Contribution: Diva.}
We introduce Diva, the first range filter to support
\underline{d}ynam\underline{i}c operations, \underline{va}riable-length queries
and keys, and high performance, all at the same time. Diva learns the dataset's
distribution by sampling keys and storing them in a cache-efficient sample trie
(\ST). For all keys in-between two samples of the \ST, Diva removes the longest
common prefix. It also truncates their suffixes while keeping enough bits in
the middle of each key (i.e., an infix) to differentiate them in most cases,
thus achieving~(G1). At the same time, the \ST separates dense and sparse
regions of the key space. This allows for handling short range queries over
densely populated regions and long range queries over sparse regions, thus
meeting~(G2). By discretizing all keys into fixed-length infixes without the
use of hashing, Diva achieves~(G3). This discretization also allows for storing
the infixes of adjacent groups of keys within a constant time data structure
called an Infix Store, fulfilling~(G5). As Diva derives infixes without hashing
and stores them in the original sorted order of the keys, it does a single
efficient sequential pass over the dataset during construction, making it the
fastest range filter to construct. It therefore attains~(G6). Diva handles
dynamic updates by splitting Infix Stores, thereby meeting~(G4).

\begin{table}
    \centering
    \caption{Definitions of terms and symbols.}
    \vspace{0.9mm}
    \bgroup
    \def\arraystretch{1.02}
    \small
    \begin{tabular}{cc}
        \toprule
        \textbf{Symbol} & \textbf{Definition} \\
        \midrule
        $q = [q_l, q_r]$ & An inclusive query range. \\
        $\epsilon$ & Target FPR, i.e., probability of a false positive. \\
        $L$ & Average key length. \\
        $(a)_2$ & Binary value represented by $a$. \\
        \cmidrule(lr){1-2}
        \ST & The trie storing a sample of the dataset. \\
        $T$ & Number of keys between two samples. \\
        $m^i_{\text{shared}}$ & \makecell[c]{Length of the longest common prefix of \\[-2pt] the $i$-th and $(i+1)$-th samples, in bits.} \\
        $m_{\text{infix}}$ & Length of the infixes. \\
        $m^i_{\text{redundant}}$ & \makecell[c]{Number of redundant bits corresponding \\[-2pt] to the $i$-th and $(i+1)$-th samples.} \\
        \DT & A binary trie differentiating keys with identical infixes. \\
        $\alpha$ & Load factor of the Infix Stores. \\
        $x_q$ & Quotient of the infix $x$. \\
        $x_r$ & Remainder of the infix $x$. \\
        \bottomrule
    \end{tabular}
    \egroup
    \label{tab:term_and_symbol_definitions}
\end{table}

\new{
As described above, Diva learns many data distributions common in practice
using its \ST. Nevertheless, we observe that the string data present in certain
applications (e.g., emails or URLs) follow ``jagged'' distributions that cannot
be learned with the \ST alone. Such distributions introduce redundancies among
Diva's infixes, making them less likely to differentiate the keys from a range
query's endpoints and increasing the FPR.
}

\new{
To support such
distributions, we introduce Diva++, an enhanced variant of Diva. Diva++
achieves this by eliminating redundancies among infixes using two complementary
techniques: (1)~It applies order-preserving entropy encoding to the input keys.
Doing so removes bits common among adjacent keys and ``flattens'' the data
distribution, thus allowing more of each key to be included in its infix and
increasing the likelihood of infixes being differentiable. (2)~For infixes that
remain identical, Diva++ avoids storing redundant copies of them and instead
uses their slots to fully differentiate their keys by encoding more of their
truncated, less significant bits. In particular, it achieves this by storing a
binary data trie (\DT) that encodes just enough bits from each key to
\mbox{differentiate it from the rest}.
}

\textbf{Additional Contributions.}
\begin{itemize}
    \item We quantify the necessary conditions on the data distribution for
        Diva \new{and Diva++} to achieve a low FPR and memory
        footprint while supporting variable-length range queries.

    \item \new{For Diva++, we design a novel \underline{bi}nary
        \underline{tr}ie encoding (\emph{BITR}), which jointly encodes both the
        topology and edge labels of a trie in optimal space{\textemdash}less
        than previous approaches.}

    \item We mathematically prove \new{the low FPR and high
        performance properties of our range filter designs.}

    \item 
        \ifarxiv
        \new{By proving several memory footprint lower bounds, we show that it
        is impossible to extend Diva (and more generally, any filter with
        ``common-case'' guarantees) to support stronger FPR guarantees at the
        same memory footprint.}
        \fi

    \item We empirically evaluate Diva \new{and Diva++} against
        other range filters in static and dynamic settings. We also conduct
        end-to-end experiments on top of WiredTiger~\cite{WiredTiger}, a
        popular B-Tree-based key-value store.
\end{itemize}

\section{Problem Analysis}~\label{sec:problem_analysis}
This section shows that no current range filter satisfies all of the six goals
outlined in Section~\ref{sec:introduction}. In-depth summaries of these filters
are presented in \cite{Memento,BeyondBloomTutorial}.

\begin{table}
    \centering
    \caption{Diva is the first range filter to simultaneously support
    variable-length queries and keys, as well as dynamicity.}
    \vspace{-3mm}
    \def\arraystretch{1.0687}
    \setlength{\tabcolsep}{5.5pt}
    \begin{tabular}{cccccc}
        \toprule

        \textbf{\footnotesize Filter} 
        & \makecell[c]{\scriptsize \textbf{Robust} \\[-4pt] \scriptsize \textbf{FPR (G1)}}
        & \makecell[c]{\scriptsize \textbf{Semi-Robust} \\[-4pt] \scriptsize \textbf{FPR (G1)}}
        & \makecell[c]{\scriptsize \textbf{Var.-Len.} \\[-4pt] \scriptsize \textbf{Queries (G2)}} 
        & \makecell[c]{\scriptsize \textbf{Var.-Len.} \\[-4pt] \scriptsize \textbf{Keys (G3)}}
        & \makecell[c]{\scriptsize \textbf{Dynamic} \\[-4pt] \scriptsize \textbf{(G4)}} \\

        \midrule

        \footnotesize \texttt{SuRF} & & & \cmark & \cmark & \\

        \footnotesize \texttt{Rosetta} & \cmark & \cmark & & & \\
        
        \footnotesize \texttt{REncoder} & & & & & \\

        \footnotesize \texttt{bloomRF} & & & & & \\

        \footnotesize \texttt{Proteus} & & & & & \\

        \footnotesize \texttt{SNARF} & & \cmark & \cmark & & \\

        \footnotesize \texttt{Oasis+} & & \cmark & \cmark & & \\

        \footnotesize \texttt{Grafite} & \cmark & \cmark & & & \\

        \footnotesize \texttt{Memento} & \cmark & \cmark & & & \cmark \\

        \footnotesize \texttt{Aeris} & \cmark & \cmark & & & \cmark \\

        \cmidrule(lr){1-6}

        \footnotesize \texttt{Diva} & & \cmark & \cmark & \cmark & \cmark \\

        \footnotesize \texttt{Diva++} & & \cmark & \cmark & \cmark & \cmark \\

        \bottomrule
    \end{tabular}
    \label{tab:method_stats}
\end{table}

\begin{table*}
    \centering
    \caption{We compare the range filters' time and space complexities assuming
    an FPR of~$\epsilon$, $N$ keys, and queries of length~$R$ with endpoint
    keys of length $L$. For SuRF and Proteus, $z$ is the number of internal
    nodes in their trie, while $m$ denotes the length of the fingerprints
    stored at the leaves. For REncoder, $k$ is its number of hash functions,
    which is roughly an~$O(\log \frac{1}{\epsilon})$ value. For SNARF, $B$ is
    the block size. For Memento, Aeris, and Diva, $\alpha$ is the load factor.
    \new{Finally, for Diva++, $t$ represents the total number of
    extra differentiating bits (over the original space allocated for the
    infixes) stored in its D-Tries. As discussed in
    Section~\ref{sec:static_bit-level_full_key_diff},
    we often found this quantity to be~0 in practice.} The operation costs are
    measured as the expected number of cache misses. The construction analysis
    assumes that the filters are constructed on sorted keys. The range query
    columns focus on empty and non-empty range queries, respectively. The
    filters marked with * do not provide mathematical bounds on their memory
    consumption. Thus, following \cite{Memento}, we give a conservative
    estimate of their memory to enable a comparison.}
    \def\arraystretch{1.1}
    \begin{tabular}{ccccccc}
        \toprule

        \small \textbf{Filter} & \small \textbf{Construction} & \small
        \textbf{Insert} & \small \textbf{Delete} & \small \textbf{Range Query
        (-)} & \small \textbf{Range Query (+)} & \small \textbf{BPK} \\

        \midrule

        \small \texttt{SuRF}~\cite{SuRF} & $O(NL)$ & - & - & $O(L)$ & $O(L)$ &
        $10 + \frac{10z}{N} + m + o(1)$ \\

        \small \texttt{Rosetta}~\cite{Rosetta} & $O(N\log {\frac{R}{\epsilon}})$
        & $O(\log {\frac{R}{\epsilon}})$ & - & $O(\log R)$ & $O(\log
        {\frac{R}{\epsilon}})$ & $1.44 \cdot \log_2 {\frac{R}{\epsilon}}$ \\

        \small \texttt{REncoder}~\cite{REncoder,REncoder_Journal} * & $O(Nk)$ &
        $k$ & - & $k$ & $k$ & $O(k + \log \frac{1}{\epsilon})$ \\

        \small \texttt{bloomRF}~\cite{bloomRF} * & $O(N (L - \log N))$ & $O(L -
        \log N)$ & - & $O(L - \log N)$ & $O(L - \log N)$ & $\approx 1.2 \cdot
        \log_2 \frac{R}{\epsilon}$ \\

        \small \texttt{Proteus}~\cite{Proteus} & $O(N (L + \log
        \frac{1}{\epsilon}))$ & - & - & $O(L)$ & $O(L + \log
        \frac{1}{\epsilon})$ & $\frac{10z}{N} + 1.44 \cdot \log_2
        \frac{1}{\epsilon}$ \\

        \small \texttt{SNARF}~\cite{SNARF} & $O(N)$ & $O(\log N+B)$ & $O(\log
        N+B)$ & $O(\log N)$ & $O(\log N)$ & $2.4 + \log_2 \frac{1}{\epsilon}$
        \\

        \small \texttt{Oasis+}~\cite{Oasis} & $O(N)$ & - & - & $O(\log N)$ &
        $O(\log N)$ & $\approx 2.4 + \log_2 \frac{1}{\epsilon}$ \\

        \small \texttt{Grafite}~\cite{Grafite} & $O(N \log N)$ & - & - & $3$ &
        $3$ & $2 + \log_2 \frac{R}{\epsilon} + o(1)$ \\

        \small \texttt{Memento}~\cite{Memento} & $O(N)$ & $O(1)$ & $O(1)$ & $1-4$
        & $1-4$ & $\frac{1}{\alpha} (3.125 + \log_2 \frac{R}{\epsilon})$ \\

        \small \texttt{Aeris}~\cite{Aeris} & $O(N)$ & $O(1)$ & $O(1)$ & $1-4$
        & $1-4$ & $\frac{1}{\alpha} (4.125 + \log_2 \frac{R}{\epsilon})$ \\

        \cmidrule(lr){1-7}

        \small \texttt{Diva} \small{[Static]} & $O(N)$ & - & - & $O(\log
        L)$ & $O(\log L)$ & $3.19 + \log_2 \frac{1}{\epsilon}$ \\

        \small \texttt{Diva} \small{[Dynamic]} & $O(N)$ & $O(\log L)$ &
        $O(\log L)$ & $O(\log L)$ & $O(\log L)$ & $1.19 + \frac{1}{\alpha^2} (2
        + \log_2 \frac{1}{\epsilon})$ \\

        \small \new{\texttt{Diva++} [Static]} & $O(N)$ & - & - & $O(\log
        L + t/N)$ & $O(\log L + t/N)$ & $3.19 + t/N + \log_2 \frac{1}{\epsilon}$ \\

        \small \new{\texttt{Diva++} [Dynamic]} & $O(N)$ & $O(\log L + t/N)$ &
        $O(\log L + t/N)$ & $O(\log L + t/N)$ & $O(\log L + t/N)$ & \small $1.19 + t/N
        + \frac{1}{\alpha^2} (2 + \log_2 \frac{1}{\epsilon})$ \\

        \bottomrule \\
    \end{tabular}
    \label{tab:method_complexities}
\end{table*}

\textbf{Range Filtering Definitions.}
A range filter represents a set $S$ of keys coming from a universe of size~$u$.
Given a range query of the form $q=[q_l,q_r]$, the filter checks if the range
is empty by answering whether $q \cap S = \emptyset$ with an FPR of at
most~$\epsilon$, where $0 < \epsilon < 1$. The top three rows of
Table~\ref{tab:term_and_symbol_definitions} summarize terms describing the
range filtering problem throughout the paper.

\textbf{Memory Lower Bound.}
A range filter is \emph{robust} if it guarantees an FPR of at most~$\epsilon$
for any dataset \new{and workload}. A non-robust range filter often
drops all information about the keys' lower-order bits. This can lead to a high
FPR when range queries predicate over these lower-order bits. It is known that
any robust range filter supporting queries of length up to $R$ must use at
least $\log_2 \frac{R}{\epsilon} - O(1)$ \emph{bits per key
(BPK)}~\cite{Goswami}, meaning that answering longer range queries requires
more memory. Intuitively, this is because a robust range filter supporting
longer queries must carry more information about which areas of the key space
are empty.

\textbf{Robustly Achieving (G1), (G2), and (G3) is Impossible.}
In most applications, filters are allotted a stringent memory budget
of~8-16~BPK to fit in memory (with a higher budget, one may as well store the
full keys). If we rearrange the above lower bound in terms of range query
length $R$ and plug in 16~BPK as the memory budget and $\epsilon=0.01$ as the
target FPR, we find that a robust range filter can only answer range queries of
length at most~512. Such short query lengths limit the applicability of the
filter. 

The lower bound also implies that a robust range filter cannot support
variable-length keys. With variable-length keys, there can be infinitely many
possible keys within a range, meaning~$R=\infty$. Plugging~$R=\infty$ into the
lower bound, we find that infinite memory is needed to support variable-length
keys. In sum, it is impossible to attain all of~(G1), (G2), and (G3) with a
robust range filter.

Indeed, none of the existing robust range
filters~\cite{Rosetta,Grafite,Memento,Aeris} support variable-length
queries~(G2) or keys~(G3), as they assume range query lengths bounded by~$R$.
Rosetta~\cite{Rosetta} stores prefixes of keys based on length in a hierarchy
of Bloom filters of depth~$\left\lceil \log_2 R \right\rceil$, with lower
levels storing longer prefixes. Grafite~\cite{Grafite} encodes~$N$ integer keys
in a bitmap of size~$\frac{NR}{\epsilon}$ by mapping each to a bit using a
locality-preserving hash function and setting it to~1. Memento
filter~\cite{Memento} and Aeris filter~\cite{Aeris} split each key into a
prefix and a~$\left\lceil \log_2 R \right\rceil$-bit suffix, hashing each
prefix to compute a fingerprint key to a compact hash table and storing its
corresponding suffix in it. In general, issuing longer range queries forces
each of these filters to issue more Bloom filter probes (in Rosetta), check
more bits (in Grafite), or check more fingerprints (in Memento filter and Aeris
filter). If the range query length exceeds the intended value of~$R$ used to
construct these filters, the FPR grows beyond~$\epsilon$ and approaches~$1$.

\textbf{Semi-Robust FPR Guarantee.}
This paper observes that, despite the aforementioned impossibility, one can
achieve~(G1), (G2), and (G3) by providing a semi-robust FPR guarantee:
\emph{Any query\footnote{\new{To reflect how filters are used in
real-world systems, this query is chosen before the dataset is revealed to the
user. Mathematically, this is equivalent to the query being statistically
independent of the dataset.}} must be answered with an FPR of at
most~$\epsilon$ when the dataset comes from a ``well-behaved'' distribution.}
Intuitively, a well-behaved distribution is one for which the cumulative
distribution function is smooth. This property is commonly satisfied in
practice since input datasets typically follow well-known, smooth distributions
(e.g., \new{uniform, normal, multi-modal, Zipfian, power law}). We
formally define well-behaved data distributions and show that Diva provides the
above FPR guarantee for them in
Section~\ref{sec:theoretical_analysis_and_results}.

\new{
Interestingly, SuRF~\cite{SuRF} and the learning-augmented
filters~\cite{LearnedIndexes,DaisyBloomFilters} SNARF~\cite{SNARF} and
Oasis+~\cite{Oasis} offer a semi-robust FPR guarantee for variable-length
queries~(G2), though their publications do not formalize
it.\footnote{\new{One can adapt the proof of Diva's semi-robustness
in Section~\ref{sec:theoretical_analysis_and_results}
to these filters.}} SuRF achieves this by storing the dataset's keys within a
compact trie, truncating each key as much as possible while ensuring it is
associated with a unique leaf node. Range queries search the trie for a key
falling within the range. Using this trie allows SuRF to also support
variable-length keys and meet~(G3).
}

\new{
However, SuRF's memory footprint is suboptimal, violating~(G1). This is because
SuRF's trie stores a full prefix of each key explicitly. Traversing this trie
to answer a range query is also slow~(G5), since it incurs one cache miss per
level. Moreover, the trie's tightly packed, pointer-less encoding precludes
efficient dynamic updates~(G4). 
}

\new{
The learning-augmented filters SNARF~\cite{SNARF} and Oasis+~\cite{Oasis} fit a
linear spline model to the keys' CDF, use the model to map each key to a bit in
a large bitmap, and set that bit to~1. These filters handle variable-length
queries by exploiting the mapping's monotonicity: the query maps to a
contiguous bit range, which is scanned for a~1. Unlike SuRF, these filters do
not force full differentiation of the keys, as multiple keys can be mapped to
the same bit within the bitmap. Thus, by compactly encoding the bitmap, these
filters bound their memory foorprint while providing a semi-robust FPR
guarantee, thereby meeting~(G1). 
}

Nevertheless, SNARF and Oasis+ do not meet~(G3), (G4), and (G5). They assume
fixed-length keys when learning the dataset's distribution, not attaining~(G3).
Although SNARF provides insertion and deletion APIs, they are prohibitively
slow. This is because SNARF (and Oasis+) encode their bitmaps using the
Elias-Fano scheme~\cite{EliasFanoElias,EliasFanoFano} and partition them into
blocks of~$B$ entries to save space and speed up decoding. SNARF's insertion
and deletion APIs are slow since they rewrite an entire block and change its
size, potentially increasing it to~$O(N)$ over time. Oasis+ disallows dynamic
operations since it prunes away empty regions of the key space to improve its
FPR. Thus, both do not achieve~(G4). They are also slow to query due to using
binary search, thereby not meeting~(G5).

\textbf{Challenges of Attaining Dynamicity (G4).}
Range filters exhibit an intrinsic contention between support for fast dynamic
operations and memory efficiency. Many existing range filters optimize for
memory by operating as Bloom filters in their core, hashing keys into a bitmap
and setting bits from 0s to
1s~\cite{Rosetta,REncoder,REncoder_Journal,bloomRF,Proteus}. As with standard
Bloom filters, such filters cannot support deletes (by resetting a bit back
to~0) or expansions (by remapping the~1s to a larger bitmap) without
introducing false negatives~\cite{Memento}.

Other range filters utilize compact encoding schemes (e.g., succinct
tries~\cite{SuRF} or Elias-Fano~\cite{SNARF,Oasis,Grafite}). Such formats are
difficult to update, as they tightly pack their data to avoid storing pointers
and offsets. This necessitates changing their entire representation in memory
to make room for new insertions.

As of today, Memento filter~\cite{Memento} and Aeris filter~\cite{Aeris} are
the only filters that support insertions and deletions in expected constant
time. However, since they doubles in size to expand, they wastes as much as
50\% of their capacity.

Diva overcomes this contention between memory efficiency and dynamicity by
establishing a one-to-one mapping between keys and the metadata it stores to
facilitate deletions, similarly to Memento filter and other expandable
filters~\cite{TaffyFilters,InfiniFilter,AlephFilter,Memento,Aeris}.
Furthermore, It slightly overprovisions memory to absorb insertions and avoids
wasting memory by expanding in small increments.

\textbf{Contention between Supporting Variable-Length Queries (G2) and Query Speed (G5).}
It is an open question whether achieving fast operations while supporting
variable-length range queries is possible. Grafite~\cite{Grafite}, Memento
filter~\cite{Memento}, and Aeris filter~\cite{Aeris} are the only filters that
provide constant time queries. They do so by localizing partitions of size~$R$
(i.e., the maximum range query length) of the key space to the same memory
region. It is unclear how to achieve a similar localization with
variable-length range queries. As such, these three filters do not
fulfill~(G2). All other filters that attain~(G1) and (G2), i.e.,
SNARF~\cite{SNARF} and Oasis+~\cite{Oasis}, use predecessor search to handle
queries, which requires super-constant
time~\cite{PredecessorLowerBound1,PredecessorLowerBound2}. In practice, these
filters use binary search to compute predecessors. Diva alleviates this
contention by employing a $y$-Fast trie~\cite{xyTries}, leading to faster
predecessor searches.

\textbf{Construction Speed (G6).}
Existing range filters extensively use
hashing~\cite{Rosetta,REncoder,REncoder_Journal,bloomRF,Proteus,Grafite,Memento,Aeris},
floating point operations~\cite{SNARF,Oasis}, or trie
traversal~\cite{SuRF,Proteus} during construction. Such operations lead to high
CPU and cache miss costs. In contrast, Diva makes limited use of hashing,
avoids floating point operations, and has good cache locality, attaining the
best construction speed and~(G6).

\textbf{Other Range Filters.}
Proteus~\cite{Proteus} employs a SuRF~\cite{SuRF} instance with a tunable
height and a Bloom filter storing key prefixes of a tunable length. It co-tunes
these structures based on a sample of the query workload. While it performs
well when the workload remains the same, it does not bound the FPR in the face
of workload shifts. As it also inherits the limitations of SuRF and Bloom
filters discussed earlier, it does not address any of the goals.

REncoder~\cite{REncoder,REncoder_Journal} and bloomRF~\cite{bloomRF} are
variants of Rosetta~\cite{Rosetta} that improve its speed by co-locating bits
representing similar prefixes in a single bitmap. However, in doing so, they
assign the same FPR to all levels in the hierarchy, breaking the robust FPR
guarantee. Since they inherit Rosetta's other limitations, they do not attain
the other goals either.

\textbf{Summary.}
Table~\ref{tab:method_stats} summarizes the goals each range filter attains,
and Table~\ref{tab:method_complexities} outlines their operation costs and
memory footprint. They demonstrate that no range filter achieves all of~(G1),
(G2), (G3), and (G4). They also highlight the question of how fast a range
filter can be while meeting these goals, i.e., (G5) and (G6).

\begin{figure}
    \centering
    \begin{tikzpicture}
        \def\circlediam{0.6mm}
        \def\divwidth{1}
        \def\storesep{0.2}
        \def\storeh{0.2}

        \fill (0.0*\divwidth,0.0) circle (\circlediam);
        \draw plot[hobby] coordinates { (0.0*\divwidth,0.0) (-0.45*\divwidth,-0.3) (-2.6*\divwidth,-1) (-3.0*\divwidth,-1.5) };
        \fill (-2.25*\divwidth,-0.8) circle (\circlediam);
        \draw plot[hobby] coordinates { (-2.25*\divwidth,-0.8) (-2.05*\divwidth,-1.2) (-2.0*\divwidth,-1.5) };
        \fill (-0.9*\divwidth,-0.47) circle (\circlediam);
        \draw plot[hobby] coordinates { (-0.9*\divwidth,-0.47) (-0.95*\divwidth,-0.8) (-1.0*\divwidth,-1.5) };
        \draw plot[hobby] coordinates { (0.0*\divwidth,0.0) (0.45*\divwidth,-0.3) (2.6*\divwidth,-1) (3.0*\divwidth,-1.5) };
        \draw plot[hobby] coordinates { (0.0*\divwidth,0.0) (0.0*\divwidth,-1.5) };
        \fill (1.5*\divwidth,-0.6) circle (\circlediam);
        \draw plot[hobby] coordinates { (1.5*\divwidth,-0.6) (1.25*\divwidth,-0.9) (1.0*\divwidth,-1.5) };
        \draw plot[hobby] coordinates { (1.2*\divwidth,-1.0) (1.75*\divwidth,-1.2) (2.0*\divwidth,-1.5) };
        \foreach \x in {-3, ..., 2} {
            \draw[fill=gray!20] (\x * \divwidth + 0.5 * \storesep, -1.5 - 0.5 * \storeh) rectangle (\x * \divwidth + \divwidth - 0.5 * \storesep, -1.5 + 0.5 * \storeh);
        }
        \draw plot[hobby] coordinates { (1.2*\divwidth,-1.0) (1.3*\divwidth,-1.2) (1.5*\divwidth,-1.5) };
        \node[inner sep=0pt] at (1.5*\divwidth,-1.65-0.5*\storeh) {\small $k$};
        \node[inner sep=1pt] (ub) at (4.0*\divwidth,-1.95) {\smaller Successor};
        \draw[-stealth] (2.0*\divwidth,-1.5-0.5*\storeh) |- (ub.west);
        \node[inner sep=1pt,below right=1pt and -34.25pt of ub] (lb) {\smaller Predecessor};
        \draw[-stealth] (1.0*\divwidth,-1.5-0.5*\storeh) |- (lb.west);
        \fill (1.2*\divwidth,-1.0) circle (\circlediam);

        \node[inner sep=1pt] (sample_trie_label) at (1.75*\divwidth,-0.13) {\small \ST};
        \draw[-stealth] plot[hobby] coordinates { (0.7*\divwidth,-0.35) (0.85*\divwidth,-0.25) (sample_trie_label.west) };
        \draw[decorate,decoration={brace,raise=1pt,amplitude=2pt}] (3.05*\divwidth,-1.5+0.75*\storeh) -- (3.05*\divwidth,-1.5-0.75*\storeh) 
            node[pos=0.5,right=3pt,inner sep=1pt] (infix_store_label) {\small Infix Stores};

        \node[inner sep=0pt] at (-3.0*\divwidth,-1.75-0.5*\storeh) {\small $0$};
        \node[inner sep=0pt,align=center] at (-2.0*\divwidth,-2.02-0.5*\storeh) {\small $T$ \\[-4pt] \small $=$ \\[-4pt] \small 1024};
        \node[inner sep=0pt,align=center] at (-1.0*\divwidth,-2.02-0.5*\storeh) {\small $2T$ \\[-4pt] \small $=$ \\[-4pt] \small 2048};
        \node[inner sep=0pt] at (0.0*\divwidth,-1.75-0.5*\storeh) {\small $\dots$};
    \end{tikzpicture}
    \caption{Diva learns the dataset's distribution by storing every $T$-th key
    as a sample in the S-Trie. It manages all keys between two adjacent samples in
    an Infix Store.}
    \label{fig:diva_overview}
\end{figure}
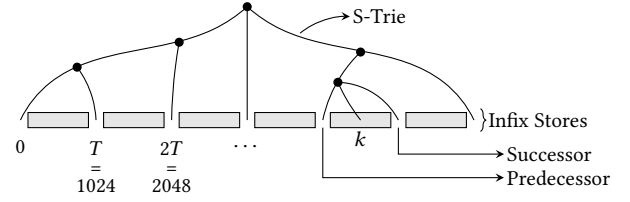

\section{Diva}~\label{sec:diva_static} 
We present Diva: the first range filter to simultaneously support (G1)~low FPR
and memory consumption, (G2)~arbitrarily sized range queries,
(G3)~variable-length keys, (G4)~dynamic updates, deletes, and resizability, and
high (G5)~query and (G6)~construction performance. Diva picks every $T$-th key
in the ordered key set as samples and stores them in a cache-efficient trie,
dubbed the \emph{\ST}, to approximate the data distribution. It inserts all
keys between two samples into a compact data structure called an \emph{Infix
Store}, as shown in
Figure~\ref{fig:diva_overview}.
For all keys in each Infix Store, Diva removes the longest common prefix. It
also homogenizes the keys' lengths by truncating a suffix from each while
leaving enough bits in the middle (i.e., an infix) to differentiate the keys
and guarantee a given FPR. \new{This is a \emph{semi-robust} FPR
guarantee, meaning that it holds when the dataset follows a ``well-behaved''
(i.e., smooth) distribution, as proven in
Section~\ref{sec:theoretical_analysis_and_results}.}
Figure~\ref{fig:diva_intro}
illustrates the truncation and mapping of keys into an Infix Store. To answer a
range query, Diva checks for the absence of overlapping keys and infixes across
the \ST and the relevant Infix Stores, as discussed in
Section~\ref{sec:diva_queries}. We describe
a static version of Diva in
Sections~\ref{sec:diva_static_sampling_keys_and_deriving_infixes}
to \ref{sec:diva_static_construction} and
generalize it to dynamic sets in Section~\ref{sec:diva_dynamic}.
\new{Additionally, Section~\ref{sec:diva++}
presents Diva++, a variant of Diva enhanced to handle datasets from ``jagged''
distributions as well.} The bottom part of
Table~\ref{tab:term_and_symbol_definitions}
summarizes the terms used throughout the paper.

\subsection{Sampling Keys and Deriving Infixes}\label{sec:diva_static_sampling_keys_and_deriving_infixes}
The static variant of Diva requires the input to arrive in sorted order. This
requirement is commonly satisfied in practice. For instance, many database
tables and indices are structured as sorted column
files~\cite{BigQuery,Cassandra,Druid,MariaDB,Snowflake},
B-Trees~\cite{BTree,B+Tree}, or LSM-Trees~\cite{RocksDB,SplinterDB,Chucky}. If
the input is unsorted, users must sort it before construction.

\begin{figure}
    \centering
    \begin{tikzpicture}
        \def\commonprefixwidth{1.0}
        \def\infixiw{0.6}
        \def\infixew{0.9}
        \def\keywidth{0.25}
        \def\keysep{0.45}
        \def\diffbitwidth{0.3}
        \def\slotw{0.5}
        \def\storeh{0.30}
        \def\storeoffsetx{2.6}
        \def\storeoffsety{0.1}
        \def\slotcnt{5}

        \node[inner sep=0pt,align=center] (sample_1) at (-\commonprefixwidth-0.8,0.125) {\small Predecessor};
        \draw[fill=gray!20] (-\commonprefixwidth,0.0) rectangle (0.0,\keywidth);
        \node[inner sep=0pt] (diff_1) at (0.5*\diffbitwidth,0.5*\keywidth) {0};
        \draw (\diffbitwidth,0.0) rectangle (\infixiw+\infixew,\keywidth);
        \draw[pattern=north east lines,pattern color=gray!50] (\infixiw+\infixew,0.0) rectangle (\infixiw+\infixew+0.5,\keywidth);

        \node[inner sep=0pt,align=center] (sample_2) at (-\commonprefixwidth-0.8,-5*\keysep+0.5*\keywidth) {\small Successor};
        \draw[fill=gray!20] (-\commonprefixwidth,-5*\keysep) rectangle (0.0,-5*\keysep+\keywidth);
        \node[inner sep=0pt] (diff_2) at (0.5*\diffbitwidth,-5*\keysep+0.5*\keywidth) {1};
        \draw (\diffbitwidth,-5*\keysep) rectangle (\infixiw+\infixew,-5*\keysep+\keywidth);
        \draw[pattern=north east lines,pattern color=gray!50] (\infixiw+\infixew,-5*\keysep) rectangle (\infixiw+\infixew+0.8,-5*\keysep+\keywidth);

        \node[inner sep=0pt] (k1) at (-\commonprefixwidth-0.2,-1*\keysep+0.5*\keywidth) {\small $a$};
        \draw[fill=gray!20] (-\commonprefixwidth,-1*\keysep) rectangle (0.0,-1*\keysep+\keywidth);
        \draw (0.0,-1*\keysep) rectangle (\infixiw,-1*\keysep+\keywidth);
        \draw (\infixiw,-1*\keysep) rectangle (\infixiw+\infixew,-1*\keysep+\keywidth);
        \draw[pattern=north east lines,pattern color=gray!50] (\infixiw+\infixew,-1*\keysep) rectangle (\infixiw+\infixew+0.2,-1*\keysep+\keywidth);
        \node[inner sep=0] (a_q) at (0.5*\infixiw,-1*\keysep+0.5*\keywidth) {\tiny $a_q$};
        \node[inner sep=0,right=0.5*\infixiw+0.5*\infixew of a_q.center,anchor=center] (a_r) {\tiny $a_r$};

        \node[inner sep=0pt] (k2) at (-\commonprefixwidth-0.2,-2*\keysep+0.5*\keywidth) {\small $b$};
        \draw[fill=gray!20] (-\commonprefixwidth,-2*\keysep) rectangle (0.0,-2*\keysep+\keywidth);
        \draw (0.0,-2*\keysep) rectangle (\infixiw,-2*\keysep+\keywidth);
        \draw (\infixiw,-2*\keysep) rectangle (\infixiw+\infixew,-2*\keysep+\keywidth);
        \draw[pattern=north east lines,pattern color=gray!50] (\infixiw+\infixew,-2*\keysep) rectangle (\infixiw+\infixew+0.5,-2*\keysep+\keywidth);
        \node[inner sep=0] (b_q) at (0.3,-2*\keysep+0.5*\keywidth) {\tiny $b_q$};
        \node[inner sep=0,right=0.5*\infixiw+0.5*\infixew of b_q.center,anchor=center] (b_r) {\tiny $b_r$};

        \node[inner sep=0pt] (k3) at (-\commonprefixwidth-0.2,-4*\keysep+0.5*\keywidth) {\small $z$};
        \draw[fill=gray!20] (-\commonprefixwidth,-4*\keysep) rectangle (0.0,-4*\keysep+\keywidth);
        \draw (0.0,-4*\keysep) rectangle (\infixiw,-4*\keysep+\keywidth);
        \draw (\infixiw,-4*\keysep) rectangle (\infixiw+\infixew,-4*\keysep+\keywidth);
        \draw[pattern=north east lines,pattern color=gray!50] (\infixiw+\infixew,-4*\keysep) rectangle (\infixiw+\infixew+0.35,-4*\keysep+\keywidth);
        \node[inner sep=0] (z_q) at (0.3,-4*\keysep+0.5*\keywidth) {\tiny $z_q$};
        \node[inner sep=0,right=0.5*\infixiw+0.5*\infixew of z_q.center,anchor=center] (z_r) {\tiny $z_r$};

        \node[inner sep=0] at (-1.2,-3*\keysep+0.8*\keywidth) {\footnotesize $\vdots$};
        \draw[decorate,decoration={brace,mirror,raise=1pt,amplitude=2pt}] ($(k1.north west)+(-2pt,0pt)$) -- ($(k3.south west)+(-2pt,0pt)$)
            node[pos=0.5,left=4pt,inner sep=1pt,align=center] (shared_label) {\small Keys};

        \draw[decorate,decoration={brace,raise=1pt,amplitude=2pt}] (-\commonprefixwidth,\keywidth) -- (0.0,\keywidth) 
            node[pos=0.5,above=4pt,inner sep=1pt,align=center] (shared_label) {\small Longest Common \\[-4pt] \small Prefix};
        \draw[decorate,decoration={brace,raise=1pt,mirror,amplitude=2pt}] (0.0,-5*\keysep) -- (\infixiw+\infixew,-5*\keysep) 
            node[pos=0.5,below=4pt,inner sep=1pt] (reduced_string_label) {\small Differentiator};
        \draw[dotted] (0.0,-5*\keysep) -- (0.0,0.0);
        \draw[dotted] (\infixiw+\infixew,-5*\keysep) -- (\infixiw+\infixew,0.0);
        \draw[decorate,decoration={brace,raise=1pt,amplitude=2pt}] (\infixiw+\infixew,\keywidth) -- (\infixiw+\infixew+0.8,\keywidth) 
            node[pos=0.5,above=4pt,inner sep=1pt,align=center] (reduced_string_label) {\small Truncated \\[-4pt] \small Suffix};
        \draw[dotted] (\infixiw+\infixew+0.8,-5*\keysep) -- (\infixiw+\infixew+0.8,\keywidth);

        \draw (\storeoffsetx,\storeoffsety) rectangle (\storeoffsetx+\slotcnt*\slotw,\storeoffsety+\storeh);
        \foreach \x in {1, ..., \slotcnt} {
            \draw (\storeoffsetx+\x*\slotw,\storeoffsety) -- (\storeoffsetx+\x*\slotw,\storeoffsety+\storeh);
        }
        \node[inner sep=0] (stored_a_r) at (\storeoffsetx+0*\slotw+0.5*\slotw,\storeoffsety+0.5*\storeh) {\small $a_r$};
        \node[inner sep=0] (stored_b_r) at (\storeoffsetx+1*\slotw+0.5*\slotw,\storeoffsety+0.5*\storeh) {\small $b_r$};
        \node[inner sep=0] (dots) at (\storeoffsetx+2*\slotw+0.5*\slotw,\storeoffsety+0.5*\storeh) {\small $\dots$};
        \node[inner sep=0] (stored_z_r) at (\storeoffsetx+3*\slotw+0.5*\slotw,\storeoffsety+0.5*\storeh) {\small $z_r$};
        \node[inner sep=0,align=center] (infix_store_label) at (\storeoffsetx+5.8*\slotw,\storeoffsety+0.5*\storeh) {\small Infix \\[-4pt] \small Store};

        \draw ($(a_q.north) + (0.0,0.05)$) -- ($(a_q.north) + (0.0,0.15)$);
        \draw[-stealth] ($(a_q.north) + (0.0,0.15)$) -| ($(stored_a_r.south) + (0.0,-0.06)$);
        \draw ($(b_q.north) + (0.0,0.01)$) -- ($(b_q.north) + (0.0,0.11)$);
        \draw[-stealth] ($(b_q.north) + (0.0,0.11)$) -| ($(stored_b_r.south) + (0.0,-0.03)$);
        \draw ($(z_q.north) + (0.0,0.04)$) -- ($(z_q.north) + (0.0,0.14)$);
        \draw[-stealth] ($(z_q.north) + (0.0,0.14)$) -| ($(stored_z_r.south) + (0.0,-0.06)$);

        \node[inner sep=1pt] (query) at (\storeoffsetx+3.5*\slotw,2*\storeh) {\small Query};
        \draw[-stealth] (query) to[bend right=20] ($(stored_b_r.north)+(0.0,0.05)$);
        \draw[-stealth] (query) to[bend right=20] ($(dots.north)+(0.0,0.15)$);

    \end{tikzpicture}
    \caption{For all keys in-between a pair of adjacent samples in the trie,
    Diva removes the longest common prefix and truncates the longest possible
    suffix while still ensuring that there are enough bits left to distinguish
    between the keys.}
    \label{fig:diva_intro}
\end{figure}
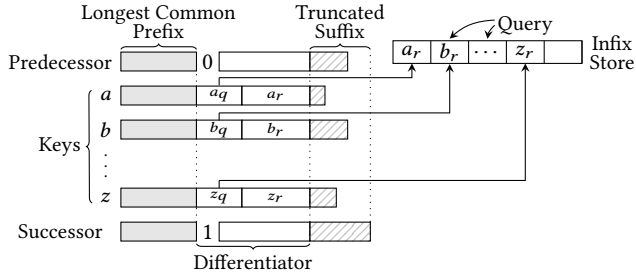

Diva samples every $T$-th key from the ordered key set and inserts them into
its \ST. This trie approximates the dataset's distribution. The reason is that
adjacent samples that are lexicographically close to each other correspond to
denser regions of the key space, while faraway samples correspond to sparser
regions. Thus, the \ST ``learns'' the data distribution. The smaller $T$ is,
the more accurate the approximate distribution becomes, yet the larger the
\ST's memory footprint grows. We quantify the relationship between $T$ and the
approximation accuracy in
Section~\ref{sec:theoretical_analysis_and_results}.
We have empirically found that $T=1024$ provides good all-round accuracy and
performance with little memory overhead ($\approx 1\%$ of the total memory
footprint).

\textbf{Removing the Longest Common Prefix.}
Consider a pair of consecutive samples. We refer to them as the
\emph{predecessor} and \emph{successor} of the keys between them.
Figure~\ref{fig:diva_overview}
shows an example of predecessor and successor for a key $k$. For all keys
in-between such a pair of samples, Diva removes their longest common prefix
since it can be inferred from the \ST and is thus redundant. 

Keys in the denser regions of the key space have longer common prefixes.
Therefore, removing the longest common prefixes saves more memory in these
regions. As we will see shortly, Diva exploits this and encodes the keys in
such regions at a higher resolution by also storing their lower-order bits,
losing less information. This enables accurately answering fine-grained queries
over dense regions while supporting coarse-grained queries over sparse regions.

Given the $i$-th and $(i+1)$-th samples as the predecessor and successor, we
denote the length of their longest common prefix, in bits, by
$m^i_{\text{shared}}$.
Figure~\ref{fig:shared_redundant_infix_example}{\nobreakdash-}A shows an
example where the longest common prefix of the predecessor and successor is
$m^i_{\text{shared}}=4$ bits long. If the samples have different lengths, Diva
treats the shorter one as having trailing zeros to match the length of the
longer one. 

\textbf{Deriving Infixes.}
To further curb memory footprint, Diva also truncates suffixes for all keys
between two samples. What remains of each key is called an \emph{infix} since
it is a sequence of adjacent bits in the middle of the original key. All
infixes in the filter have the same length, denoted by $m_{\text{infix}}$. Due
to the removed common prefix, infixes represent less significant bits of the
keys in dense regions of the key space and more significant bits in sparser
ones. Truncating keys into fixed-length infixes discretizes the range between
two samples while removing the variation in length of the keys. As we will see,
Diva supports range queries over these infixes by correspondingly discretizing
the query boundaries and checking for the inclusion of infixes between them.
The fixed length of the infixes allows for quickly checking for inclusion using
integer arithmetic.
\new{
Figure~\ref{fig:shared_redundant_infix_example}{\nobreakdash-}A
shows an example with keys of varying lengths and $m_{\text{infix}} = 7$. As
$m^i_{\text{shared}}=4$, Diva derives~$k$'s infix as the~7 bits following its
first~4 bits, i.e., the string~$(0101001)_2$.
}

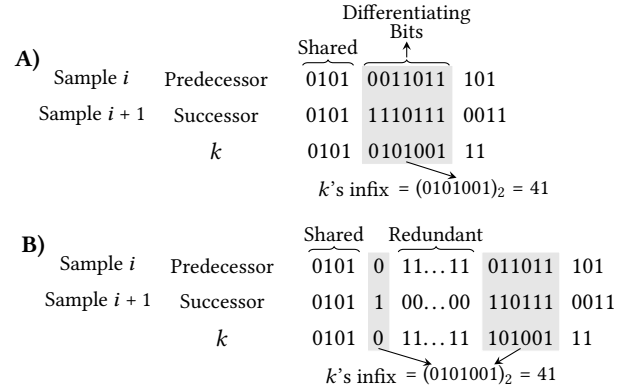
\begin{figure}
    \centering

    \pgfdeclarelayer{background layer}
    \pgfsetlayers{background layer,main}
    \begin{tikzpicture}
        \node[inner sep=0pt] (fig_label) at (-2.5,0.25) {\textbf{A)}};
        \node[inner sep=4.85pt,align=center] (s2) at (0.0,0.0) {\small Predecessor};
        \node[inner sep=2pt,left=6pt of s2,align=center] (sample_i) {\small Sample $i$};
        \begin{scope}[name prefix=s2_]
            \node[inner sep=1pt,right=8pt of s2] (shared) {0101};
            \node[inner sep=1pt,right=4pt of shared] (diff_bit) {0};
            \node[inner sep=1pt,right=-2pt of diff_bit] (infix) {011011};
            \node[inner sep=1pt,right=4pt of infix] (rest) {$101$};
        \end{scope}

        \node[inner sep=8.35pt,below=-5.5pt of s2.south,align=center,anchor=north] (s1) {\small Successor};
        \node[inner sep=2pt,left=0pt of s1,align=center] (sample_i1) {\small Sample $i+1$};
        \begin{scope}[name prefix=s1_]
            \node[inner sep=1pt,right=8pt of s1] (shared) {0101};
            \node[inner sep=1pt,right=4pt of shared] (diff_bit) {1};
            \node[inner sep=1pt,right=-2pt of diff_bit] (infix) {110111};
            \node[inner sep=1pt,right=4pt of infix] (rest) {$0011$};
        \end{scope}

        \draw[decorate,decoration={brace,raise=1pt,amplitude=2pt}] (s2_shared.north west) -- (s2_shared.north east) 
            node[pos=0.5,above=4pt,inner sep=1pt] (n_shared_label) {\small Shared};

        \draw[decorate,decoration={brace,raise=1pt,amplitude=2pt}] (s2_diff_bit.north west) -- (s2_infix.north east) 
            node[pos=0.5,above=10pt,inner sep=1pt,align=center] (n_diff_bits_label) {\small Differentiating \\[-4pt] \small Bits};
        \draw[stealth-] (n_diff_bits_label.south) -- ($(n_diff_bits_label.south)+(0,-6pt)$);

        \node[inner sep=23.0pt,below=-7pt of s2] (k) {$k$};
        \begin{scope}[name prefix=k_]
            \node[inner sep=1pt,right=7.5pt of k] (shared) {0101};
            \node[inner sep=1pt,right=3.75pt of shared] (diff_bit) {0};
            \node[inner sep=1pt,right=-2pt of diff_bit] (infix) {101001};
            \node[inner sep=1pt,right=4pt of infix] (rest) {$11$};
        \end{scope}

        \begin{pgfonlayer}{background layer}
            \draw[white,fill=gray!20] ($(s2_diff_bit.north west) + (-1pt,1pt)$) rectangle ($(k_infix.south east) + (1pt,-1pt)$);
        \end{pgfonlayer}
        \node[inner sep=1pt,below right=6pt and -15pt of k_shared] (res_infix) {\small $k$'s infix};
        \node[inner sep=1pt,right=1pt of res_infix] {\small $= (0101001)_2 = 41$};
        \draw[-stealth] (k_infix.235) -- ($(res_infix.north east) + (0.85,0.0)$);

        \node[inner sep=1pt,right=3.5 of k_shared] { };
        \node at (4.75,0.0) {};
    \end{tikzpicture}

    \pgfdeclarelayer{background layer}
    \pgfsetlayers{background layer,main}
    \begin{tikzpicture}
        \node[inner sep=0pt] (fig_label) at (-2.5,0.25) {\textbf{B)}};
        \node[inner sep=4.85pt,align=center] (s2) at (0.0,0.0) {\small Predecessor};
        \node[inner sep=2pt,left=6pt of s2,align=center] (sample_i) {\small Sample $i$};
        \begin{scope}[name prefix=s2_]
            \node[inner sep=1pt,right=8pt of s2] (shared) {0101};
            \node[inner sep=1pt,right=4pt of shared] (diff_bit) {0};
            \node[inner sep=1pt,right=4pt of diff_bit] (redundant) {11\dots11};
            \node[inner sep=1pt,right=4pt of redundant] (infix) {011011};
            \node[inner sep=1pt,right=4pt of infix] (rest) {$101$};
        \end{scope}

        \node[inner sep=8.35pt,below=-5.5pt of s2.south,align=center,anchor=north] (s1) {\small Successor};
        \node[inner sep=2pt,left=0pt of s1,align=center] (sample_i1) {\small Sample $i+1$};
        \begin{scope}[name prefix=s1_]
            \node[inner sep=1pt,right=8pt of s1] (shared) {0101};
            \node[inner sep=1pt,right=4pt of shared] (diff_bit) {1};
            \node[inner sep=1pt,right=4pt of diff_bit] (redundant) {00\dots00};
            \node[inner sep=1pt,right=4pt of redundant] (infix) {110111};
            \node[inner sep=1pt,right=4pt of infix] (rest) {$0011$};
        \end{scope}

        \draw[decorate,decoration={brace,raise=1pt,amplitude=2pt}] (s2_shared.north west) -- (s2_shared.north east) 
            node[pos=0.5,above=4pt,inner sep=1pt] (n_shared_label) {\small Shared};
        \draw[decorate,decoration={brace,raise=1pt,amplitude=2pt}] (s2_redundant.north west) -- (s2_redundant.north east) 
            node[pos=0.5,above=4pt,inner sep=1pt] (n_shared_label) {\small Redundant};

        \node[inner sep=23.0pt,below=-7pt of s2] (k) {$k$};
        \begin{scope}[name prefix=k_]
            \node[inner sep=1pt,right=7pt of k] (shared) {0101};
            \node[inner sep=1pt,right=3.75pt of shared] (diff_bit) {0};
            \node[inner sep=1pt,right=4pt of diff_bit] (redundant) {11\dots11};
            \node[inner sep=1pt,right=4pt of redundant] (infix) {101001};
            \node[inner sep=1pt,right=4pt of infix] (rest) {$11$};
        \end{scope}

        \begin{pgfonlayer}{background layer}
            \draw[white,fill=gray!20] ($(s2_diff_bit.north west) + (-1pt,1pt)$) rectangle ($(k_diff_bit.south east) + (1pt,-1pt)$);
            \draw[white,fill=gray!20] ($(s2_infix.north west) + (-1pt,1pt)$) rectangle ($(k_infix.south east) + (1pt,-1pt)$);
        \end{pgfonlayer}
        \node[inner sep=1pt,below right=6pt and -14pt of k_shared] (res_infix) {\small $k$'s infix};
        \node[inner sep=1pt,right=1pt of res_infix] {\small $=(0101001)_2=41$};

        \draw[-stealth] (k_diff_bit.south) -- ($(res_infix.north east) + (0.425,0.0)$);
        \draw[-stealth] (k_infix.south) -- ($(res_infix.north east) + (1.25,0.0)$);
    \end{tikzpicture}
    \caption{Diva derives key $k$'s infix by removing the $m^i_{\text{shared}}$
    common prefix bits and the $m^i_{\text{redundant}}$ bits based on its
    predecessor and successor, and also truncating its suffix.}
    \label{fig:shared_redundant_infix_example}
\end{figure}

\new{
\textbf{FPR.}
We prove in Section~\ref{sec:theoretical_analysis_and_results}
that Diva achieves an FPR of~$\epsilon$ by retaining infixes of
length~$m_{\text{infix}} = \left\lceil \log_2 \frac{2T}{\epsilon} \right\rceil$
bits from each key. This guarantee is semi-robust; it holds when the dataset
follows a well-behaved (i.e., smooth) distribution. In
Section~\ref{sec:diva_static_infix_stores}, we
show how Diva encodes each infix in only~$3+\left\lceil \log_2
\frac{1}{\epsilon} \right\rceil$ bits. This allows it to provide a similar FPR
vs. memory trade-off to that of traditional point filters such as
Bloom~\cite{Bloom} and Cuckoo~\cite{CuckooFilter} filters. For example, by
spending one byte per infix, Diva achieves an FPR of~$\approx 3\%$.
}

\textbf{Redundant Bits.}
In well-behaved data distributions, infixes derived with the above method often
contain redundant bits that do not contribute to filtering. Intuitively, if the
considered predecessor and successor are still close to each other after
removing their longest common prefix, the discretized key space between them
becomes smaller than desired. As a result, infixes of keys in-between can
become less distinguishable, reducing filtering accuracy. This phenomenon
occurs when many consecutive bits after the first differentiating bit in the
predecessor and successor are all 1s and 0s, respectively. In particular, in
data distributions such as \new{normal} and Zipfian, there is at
least one such bit on average due to the randomness in the samples after their
longest common prefix. The positions within each key that correspond to these
bits contain redundant bits, as these bits do not carry useful information for
comparing the keys between the~$i$-th and $(i+1)$-th samples beyond what the
first differentiating bit provides. If not removed, these redundant bits
increase the FPR by~$2\times$ compared to the case where infixes are
well-differentiated. The predecessor and successor in
Figure~\ref{fig:shared_redundant_infix_example}{\nobreakdash-}B
represent an example.

Diva prevents the aforementioned issue by identifying the longest matching
sequence of 1s from the predecessor and 0s from the successor after their first
differentiating bit. We denote the length of this sequence for the $i$-th and
$(i+1)$-th samples by $m^i_{\text{redundant}}$. Diva removes the bits
corresponding to this sequence from each key between the samples. Diva derives
the infix of each key as the first differentiating bit concatenated with the
first~$m_{\text{infix}}-1$ bits after the redundant bits. Intuitively, each of
these bits doubles the size of the discretized space between two samples,
allowing to better differentiate the infixes in-between. The following lemma
\new{(proven in \apx{proof:total_infixes})} formalizes this intuition:
\begin{lemma}[Redundant Bits]
    After removing the $m^i_{\text{redundant}}$ bits, there are at least
    $\frac{T}{\epsilon}$ and at most $\frac{2T}{\epsilon}$ different and valid
    values that infixes of keys between two samples can take. 
    \label{lemma:total_infixes} 
\end{lemma}

Diva is able to reconstruct the redundant bits removed from an infix. It does
so by recomputing~$m^i_{\text{redundant}}$ from its predecessor and successor
in the \ST. It adds $m^i_{\text{redundant}}$ bits between the infix's first and
second bits, each equal to the negation of its first bit. This results in the
original sequence of bits in the key since the redundant bits of an infix
always equal the negation of its first bit.

Figure~\ref{fig:shared_redundant_infix_example}{\nobreakdash-}B shows an
example derivation of a key $k$'s infix with $m_{\text{infix}} = 7$. Here, the
longest common prefix of the predecessor and successor is
$m^i_{\text{shared}}=4$ bits long. After removing the $m^i_{\text{redundant}}$
bits, Diva appends the next 6 bits in $k$ (i.e., 101001) to the first
differentiating bit (i.e., 0), yielding the infix $(0101001)_2=41$.

\textbf{Infix Uniformity.}
The samples in the \ST behave similarly to the boundaries of an equi-depth
histogram. When the dataset's distribution is well-behaved, the higher-order
bits of the samples capture the overall shape of the distribution and the
lower-order bits of the keys in-between behave like noise. Thus, keys falling
into the same bin and are almost uniformly distributed, implying that their
infixes are also uniformly distributed. We formally prove this property in
Section~\ref{sec:theoretical_analysis_and_results}.
This uniformity allows infixes to be treated as hash digests. We will use this
property to design an efficient data structure for storing the infixes
in-between a pair of samples.

\subsection{Infix Stores}\label{sec:diva_static_infix_stores}
An Infix Store is a random-access array of slots that leverages the uniformity
of the infixes in-between two samples to efficiently and succinctly store them.
We describe a static Infix Store variant and extend it to support dynamic
operations in Section~\ref{sec:diva_dynamic} by
utilizing techniques inspired by Cleary hash tables~\cite{Cleary} and the
Rank-and-Select Quotient Filter~\cite{GQF}.

\textbf{Quotienting.}
An Infix Store applies Knuth's quotienting technique~\cite{Quotienting} to its
infixes, splitting them into \emph{quotients} and \emph{remainders}.
Specifically, it takes the~$\left\lceil \log_2 \frac{2}{\epsilon} \right\rceil$
least-significant bits of an infix $x$ as its remainder $x_r$ while taking the
rest of its bits as its quotient~$x_q$.
Figure~\ref{fig:infix_store}{\nobreakdash-}A
illustrates this split for an infix~$x=(11011)_2$.

\textbf{Storing Quotients.}
Splitting each infix allows an Infix Store to succinctly encode its quotient
within a dense bitmap called the \texttt{occupieds}
bitmap.\footnote{\new{The dense encoding of the \texttt{occupieds}
bitmap allows Diva to store a large portion of each key implicitly. This is in
contrast with SuRF~\cite{SuRF}, which attains a semi-robust FPR guarantee by
explicitly storing a prefix of each key (as described in
Section~\ref{sec:problem_analysis}).}}
The~$i$-th bit in this bitmap is set to~1 if an infix with a quotient equal
to~$i$ exists and is set to~0 otherwise. For example, in
Figure~\ref{fig:infix_store}{\nobreakdash-}B,
infix~$x$ causes the bit at offset~3 (or $(011)_2$ in binary) of the
\texttt{occupieds} bitmap to be set to~1. This bitmap needs as many bits as
there are possible quotients, which varies in the range from~$\frac{T}{2}$ to
$T$ based on the predecessor and successor of the Infix Store. One can show
this by applying Lemma~\ref{lemma:total_infixes}
and removing the $\left\lceil \log_2 \frac{2}{\epsilon} \right\rceil$
least-significant bits belonging to remainders from the possible infixes to
derive the range of all possible quotients.

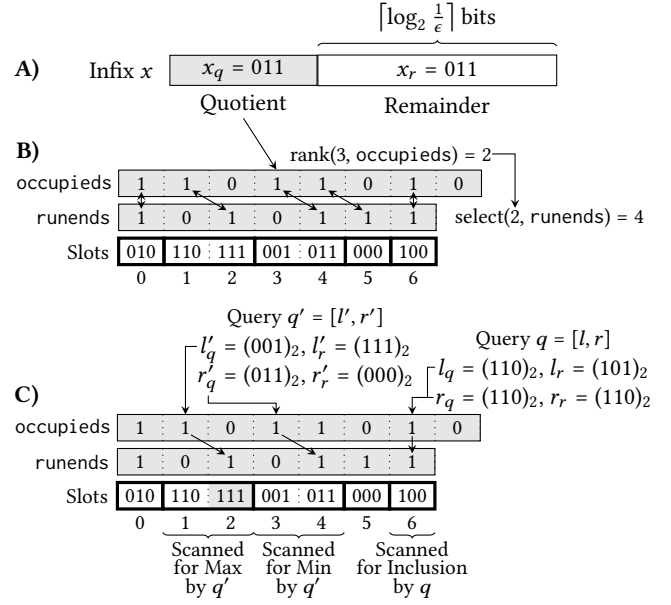
\begin{figure}
    \centering
    
    \pgfdeclarelayer{background layer}
    \pgfsetlayers{background layer,main}
    \begin{tikzpicture}
        \def\implicitw{20pt}
        \def\explicitw{55pt}
        \def\slotw{0.6}
        \def\sloth{0.35}
        \def\storeoffx{-1.25}

        \node[inner sep=0pt] (x) at (-1.2,0.35) {Infix $x$};
        \node[inner sep=2pt,right=16pt of x] (x_i) {$x_q=011$};
        \node[inner sep=2pt,right=1.33 of x_i] (x_e) {$x_r=011$};
        \begin{pgfonlayer}{background layer}
            \draw[fill=gray!20] ($(x_i.north west) + (-0.5*\implicitw,-0.75pt)$) rectangle ($(x_i.south east) + (0.5*\implicitw,0.75pt)$);
            \draw ($(x_e.north west) + (-0.5*\explicitw,0.1pt)$) rectangle ($(x_e.south east) + (0.5*\explicitw,-0.1pt)$);
        \end{pgfonlayer}
        \draw[decorate,decoration={brace,raise=4pt,amplitude=2pt}] ($(x_e.north west) + (-0.5*\explicitw,-0.5pt)$) -- ($(x_e.north east) + (0.5*\explicitw,-0.5pt)$)
            node[pos=0.5,above=8pt,inner sep=1pt] (remainder_size_label) {$\left\lceil \log_2 \frac{1}{\epsilon} \right\rceil$ bits};
        \node[below=0.05cm of x_i,inner sep=1pt] (quotient_label) {Quotient};
        \node[below=0.1cm of x_e,inner sep=1pt] (remainder_label) {Remainder};

        \node[inner sep=0pt] (fig_label) at (\storeoffx-2*\slotw,1.0*\sloth) {\textbf{A)}};
        \node[inner sep=0pt] at (\storeoffx+11.75*\slotw,0.4*\sloth) { };

        \draw[-stealth] (quotient_label.south) -- (\storeoffx+3*\slotw+0.5*\slotw,-2.75*\sloth);
    \end{tikzpicture}
    \vspace*{-23pt}

    \pgfdeclarelayer{background layer}
    \pgfsetlayers{background layer,main}
    \begin{tikzpicture}
        \def\slotw{0.6}
        \def\sloth{0.35}
        \def\occupiedcnt{8}
        \def\slotcnt{7}
        \def\storeoffx{-1.25}
        \def\storegap{0.1}

        \node[inner sep=0pt] (occupieds_label) at (\storeoffx-0.72,-0.5*\sloth-\storegap) {\small \texttt{occupieds}};
        \draw[fill=gray!20] (\storeoffx,-\storegap) rectangle (\storeoffx+\occupiedcnt*\slotw,-\storegap-\sloth);
        \foreach \x in {1, ..., \occupiedcnt} {
            \draw[dotted] (\storeoffx+\slotw*\x,-\storegap-\sloth) -- (\storeoffx+\slotw*\x,-\storegap);
        }
        \foreach \i in {1, ..., \occupiedcnt} {
            \def\x{\storeoffx+\i*\slotw-0.5*\slotw}
            \def\y{-0.5*\sloth-\storegap}
            \ifnum \i = 1 {
                \node[inner sep=0.5pt] (occ_0) at (\x,\y) {\small 1};
            }
            \else \ifnum \i = 2 {
                \node[inner sep=0pt] (occ_1) at (\x,\y) {\small 1};
            }
            \else \ifnum \i = 4 {
                \node[inner sep=0pt] (occ_2) at (\x,\y) {\small 1};
            }
            \else \ifnum \i = 5 {
                \node[inner sep=0pt] (occ_3) at (\x,\y) {\small 1};
            }
            \else \ifnum \i = 7 {
                \node[inner sep=0.5pt] (occ_4) at (\x,\y) {\small 1};
            }
            \else {
                \node at (\x,\y) {\small 0};
            }
            \fi \fi \fi \fi \fi
        }

        \node[inner sep=0pt] (runends_label) at (\storeoffx-0.6,-1.5*\sloth-2*\storegap) {\small \texttt{runends}};
        \draw[fill=gray!20] (\storeoffx,-\sloth-2*\storegap) rectangle (\storeoffx+\slotcnt*\slotw,-2*\sloth-2*\storegap);
        \foreach \x in {1, ..., \slotcnt} {
            \draw[dotted] (\storeoffx+\slotw*\x,-2*\sloth-2*\storegap) -- (\storeoffx+\slotw*\x,-\sloth-2*\storegap);
        }
        \foreach \i in {1, ..., \slotcnt} {
            \def\x{\storeoffx+\i*\slotw-0.5*\slotw}
            \def\y{-1.5*\sloth-2*\storegap}
            \ifnum \i = 1 {
                \node[inner sep=0.5pt] (run_0) at (\x,\y) {\small 1};
            }
            \else \ifnum \i = 3 {
                \node[inner sep=0pt] (run_1) at (\x,\y) {\small 1};
            }
            \else \ifnum \i = 5 {
                \node[inner sep=0pt] (run_2) at (\x,\y) {\small 1};
            }
            \else \ifnum \i = 6 {
                \node[inner sep=0pt] (run_3) at (\x,\y) {\small 1};
            }
            \else \ifnum \i = 7 {
                \node[inner sep=0.5pt] (run_4) at (\x,\y) {\small 1};
            }
            \else {
                \node at (\x,\y) {\small 0};
            }
            \fi \fi \fi \fi \fi
        }

        \node[inner sep=0pt] (slots_label) at (\storeoffx-0.4,-2.5*\sloth-3*\storegap) {\small Slots};
        \draw (\storeoffx,-2*\sloth-3*\storegap) rectangle (\storeoffx+\slotcnt*\slotw,-3*\sloth-3*\storegap);
        \foreach \x in {1, ..., \slotcnt} {
            \draw[dotted] (\storeoffx+\slotw*\x,-3*\sloth-3*\storegap) -- (\storeoffx+\slotw*\x,-2*\sloth-3*\storegap);
        }
        \node at (\storeoffx+0*\slotw+0.5*\slotw,-2.5*\sloth-3*\storegap) {\small 010};
        \draw[very thick] (\storeoffx+0*\slotw,-3*\sloth-3*\storegap) rectangle (\storeoffx+1*\slotw,-2*\sloth-3*\storegap);

        \node at (\storeoffx+1*\slotw+0.5*\slotw,-2.5*\sloth-3*\storegap) {\small 110};
        \node at (\storeoffx+2*\slotw+0.5*\slotw,-2.5*\sloth-3*\storegap) {\small 111};
        \draw[very thick] (\storeoffx+1*\slotw,-3*\sloth-3*\storegap) rectangle (\storeoffx+3*\slotw,-2*\sloth-3*\storegap);

        \node at (\storeoffx+3*\slotw+0.5*\slotw,-2.5*\sloth-3*\storegap) {\small 001};
        \node at (\storeoffx+4*\slotw+0.5*\slotw,-2.5*\sloth-3*\storegap) {\small 011};
        \draw[very thick] (\storeoffx+3*\slotw,-3*\sloth-3*\storegap) rectangle (\storeoffx+5*\slotw,-2*\sloth-3*\storegap);

        \node at (\storeoffx+5*\slotw+0.5*\slotw,-2.5*\sloth-3*\storegap) {\small 000};
        \draw[very thick] (\storeoffx+5*\slotw,-3*\sloth-3*\storegap) rectangle (\storeoffx+6*\slotw,-2*\sloth-3*\storegap);

        \node at (\storeoffx+6*\slotw+0.5*\slotw,-2.5*\sloth-3*\storegap) {\small 100};
        \draw[very thick] (\storeoffx+6*\slotw,-3*\sloth-3*\storegap) rectangle (\storeoffx+7*\slotw,-2*\sloth-3*\storegap);

        \foreach \i in {0, ..., \slotcnt} {
            \ifnum \i < \slotcnt {
                \node[inner sep=0pt] at (\storeoffx+\i*\slotw+0.5*\slotw,-3.5*\sloth-3*\storegap) {\small \i};
            }
            \fi
        }

        \draw[stealth-stealth] (occ_0.south) -- (run_0.north);
        \draw[stealth-stealth] (occ_1.south east) -- (run_1.north west);
        \draw[stealth-stealth] (occ_2.south east) -- (run_2.north west);
        \draw[stealth-stealth] (occ_3.south east) -- (run_3.north west);
        \draw[stealth-stealth] (occ_4.south) -- (run_4.north);

        \node[inner sep=0pt] (rank_label) at (\storeoffx+6.0*\slotw,0.25*\sloth) {\small $\text{rank}(3,$~\texttt{occupieds}$)=2$};
        \node[inner sep=0pt] (select_label) at (\storeoffx+\slotcnt*\slotw+2.5*\slotw,-2.1*\sloth) {\small $\text{select}(2,$~\texttt{runends}$)=4$};
        \draw[-stealth] ($(rank_label.east)+(0pt,0.75pt)$) -| ($(select_label.163)+(-0.5pt,0pt)$);

        \node[inner sep=0pt] (fig_label) at (\storeoffx-2*\slotw,0.4*\sloth) {\textbf{B)}};
    \end{tikzpicture}
    \vspace*{8pt}

    \pgfdeclarelayer{background layer}
    \pgfsetlayers{background layer,main}
    \begin{tikzpicture}
        \def\slotw{0.6}
        \def\sloth{0.35}
        \def\occupiedcnt{8}
        \def\slotcnt{7}
        \def\storeoffx{-1.25}
        \def\storegap{0.1}

        \node[inner sep=0pt] (occupieds_label) at (\storeoffx-0.72,-0.5*\sloth-\storegap) {\small \texttt{occupieds}};
        \draw[fill=gray!20] (\storeoffx,-\storegap) rectangle (\storeoffx+\occupiedcnt*\slotw,-\storegap-\sloth);
        \foreach \x in {1, ..., \occupiedcnt} {
            \draw[dotted] (\storeoffx+\slotw*\x,-\storegap-\sloth) -- (\storeoffx+\slotw*\x,-\storegap);
        }
        \foreach \i in {1, ..., \occupiedcnt} {
            \def\x{\storeoffx+\i*\slotw-0.5*\slotw}
            \def\y{-0.5*\sloth-\storegap}
            \ifnum \i = 1 {
                \node[inner sep=0.5pt] (occ_0) at (\x,\y) {\small 1};
            }
            \else \ifnum \i = 2 {
                \node[inner sep=0pt] (occ_1) at (\x,\y) {\small 1};
            }
            \else \ifnum \i = 4 {
                \node[inner sep=0pt] (occ_2) at (\x,\y) {\small 1};
            }
            \else \ifnum \i = 5 {
                \node[inner sep=0pt] (occ_3) at (\x,\y) {\small 1};
            }
            \else \ifnum \i = 7 {
                \node[inner sep=0.5pt] (occ_4) at (\x,\y) {\small 1};
            }
            \else {
                \node at (\x,\y) {\small 0};
            }
            \fi \fi \fi \fi \fi
        }

        \node[inner sep=0pt] (runends_label) at (\storeoffx-0.6,-1.5*\sloth-2*\storegap) {\small \texttt{runends}};
        \draw[fill=gray!20] (\storeoffx,-\sloth-2*\storegap) rectangle (\storeoffx+\slotcnt*\slotw,-2*\sloth-2*\storegap);
        \foreach \x in {1, ..., \slotcnt} {
            \draw[dotted] (\storeoffx+\slotw*\x,-2*\sloth-2*\storegap) -- (\storeoffx+\slotw*\x,-\sloth-2*\storegap);
        }
        \foreach \i in {1, ..., \slotcnt} {
            \def\x{\storeoffx+\i*\slotw-0.5*\slotw}
            \def\y{-1.5*\sloth-2*\storegap}
            \ifnum \i = 1 {
                \node[inner sep=0.5pt] (run_0) at (\x,\y) {\small 1};
            }
            \else \ifnum \i = 3 {
                \node[inner sep=0pt] (run_1) at (\x,\y) {\small 1};
            }
            \else \ifnum \i = 5 {
                \node[inner sep=0pt] (run_2) at (\x,\y) {\small 1};
            }
            \else \ifnum \i = 6 {
                \node[inner sep=0pt] (run_3) at (\x,\y) {\small 1};
            }
            \else \ifnum \i = 7 {
                \node[inner sep=0.5pt] (run_4) at (\x,\y) {\small 1};
            }
            \else {
                \node at (\x,\y) {\small 0};
            }
            \fi \fi \fi \fi \fi
        }

        \draw[-stealth] (occ_1.south east) -- (run_1.north west);
        \draw[-stealth] (occ_2.south east) -- (run_2.north west);
        \draw[-stealth] (occ_4.south) -- (run_4.north);

        \node[inner sep=0pt] (slots_label) at (\storeoffx-0.4,-2.5*\sloth-3*\storegap) {\small Slots};
        \draw (\storeoffx,-2*\sloth-3*\storegap) rectangle (\storeoffx+\slotcnt*\slotw,-3*\sloth-3*\storegap);
        \foreach \x in {1, ..., \slotcnt} {
            \draw[dotted] (\storeoffx+\slotw*\x,-3*\sloth-3*\storegap) -- (\storeoffx+\slotw*\x,-2*\sloth-3*\storegap);
        }
        \node at (\storeoffx+0*\slotw+0.5*\slotw,-2.5*\sloth-3*\storegap) {\small 010};
        \draw[very thick] (\storeoffx+0*\slotw,-3*\sloth-3*\storegap) rectangle (\storeoffx+1*\slotw,-2*\sloth-3*\storegap);

        \node at (\storeoffx+1*\slotw+0.5*\slotw,-2.5*\sloth-3*\storegap) {\small 110};
        \node at (\storeoffx+2*\slotw+0.5*\slotw,-2.5*\sloth-3*\storegap) {\small 111};
        \draw[very thick] (\storeoffx+1*\slotw,-3*\sloth-3*\storegap) rectangle (\storeoffx+3*\slotw,-2*\sloth-3*\storegap);
        \begin{pgfonlayer}{background layer}
            \draw[white,fill=gray!20] (\storeoffx+2*\slotw,-3*\sloth-3*\storegap) rectangle (\storeoffx+3*\slotw,-2*\sloth-3*\storegap);
        \end{pgfonlayer}

        \node at (\storeoffx+3*\slotw+0.5*\slotw,-2.5*\sloth-3*\storegap) {\small 001};
        \node at (\storeoffx+4*\slotw+0.5*\slotw,-2.5*\sloth-3*\storegap) {\small 011};
        \draw[very thick] (\storeoffx+3*\slotw,-3*\sloth-3*\storegap) rectangle (\storeoffx+5*\slotw,-2*\sloth-3*\storegap);

        \node at (\storeoffx+5*\slotw+0.5*\slotw,-2.5*\sloth-3*\storegap) {\small 000};
        \draw[very thick] (\storeoffx+5*\slotw,-3*\sloth-3*\storegap) rectangle (\storeoffx+6*\slotw,-2*\sloth-3*\storegap);

        \node at (\storeoffx+6*\slotw+0.5*\slotw,-2.5*\sloth-3*\storegap) {\small 100};
        \draw[very thick] (\storeoffx+6*\slotw,-3*\sloth-3*\storegap) rectangle (\storeoffx+7*\slotw,-2*\sloth-3*\storegap);

        \foreach \i in {0, ..., \slotcnt} {
            \ifnum \i < \slotcnt {
                \node[inner sep=0pt] at (\storeoffx+\i*\slotw+0.5*\slotw,-3.5*\sloth-3*\storegap) {\small \i};
            }
            \fi
        }

        \node[inner sep=0pt,align=center] (query_1) at (\storeoffx+9.4*\slotw,\sloth+1.25*\storegap) {\small Query $q=[l, r]$ \\ $l_q=(110)_2$, $l_r=(101)_2$ \\ $r_q=(110)_2$, $r_r=(110)_2$};
        \draw ($(query_1.west)+(0.0,0.05)$) -| ($(query_1.west)+(-0.1,-0.35)$);
        \draw[-stealth] ($(query_1.west)+(-0.02,-0.35)$) -| (\storeoffx+6*\slotw+0.5*\slotw,-\storegap);
        \draw[decorate,decoration={brace,mirror,raise=1pt,amplitude=2pt}] (\storeoffx+6*\slotw,-3.5*\sloth-4*\storegap) -- (\storeoffx+7*\slotw,-3.5*\sloth-4*\storegap) 
            node[pos=0.5,below=3pt,inner sep=1pt,align=center] (target_label) {\small Scanned \\[-4pt] \small for Inclusion \\[-4pt] \small by $q$};

        \node[inner sep=0pt,align=center] (query_2) at (\storeoffx+4.15*\slotw,2.0*\sloth+0.5*\storegap) {\small Query $q'=[l', r']$ \\ $l'_q=(001)_2$, $l'_r=(111)_2$ \\ $r'_q=(011)_2$, $r'_r=(000)_2$};
        \draw[-stealth] ($(query_2.west)+(0.0,0.05)$) -| (\storeoffx+1*\slotw+0.5*\slotw,-\storegap);
        \draw[-stealth] ($(query_2.west)+(0.15,-0.6)$) |- (\storeoffx+3.5*\slotw,0.5*\storegap) -- (\storeoffx+3.5*\slotw,-\storegap);
        \draw[decorate,decoration={brace,mirror,raise=1pt,amplitude=2pt}] (\storeoffx+1*\slotw,-3.5*\sloth-4*\storegap) -- (\storeoffx+3*\slotw,-3.5*\sloth-4*\storegap) 
            node[pos=0.5,below=3pt,inner sep=1pt,align=center] (target_label) {\small Scanned \\[-4pt] \small for Max \\[-4pt] \small by $q'$};
        \draw[decorate,decoration={brace,mirror,raise=1pt,amplitude=2pt}] (\storeoffx+3*\slotw,-3.5*\sloth-4*\storegap) -- (\storeoffx+5*\slotw,-3.5*\sloth-4*\storegap) 
            node[pos=0.5,below=3pt,inner sep=1pt,align=center] (target_label) {\small Scanned \\[-4pt] \small for Min \\[-4pt] \small by $q'$};

        \node[inner sep=0pt] (fig_label) at (\storeoffx-2*\slotw,0.4*\sloth) {\textbf{C)}};
    \end{tikzpicture}
    \caption{Diva splits an infix into a remainder and a quotient (Part A). It
    encodes the quotient in the \texttt{occupieds} bitmap while storing the
    remainder in contiguous runs in the array of slots, delineated using thick
    lines (Part B). Diva handles a query by locating the corresponding runs
    using rank and select operations, followed by scanning them (Part C). The
    remainders that satisfy the queries are highlighted.}
    \label{fig:infix_store}
\end{figure}

\textbf{Storing Remainders.}
An Infix Store places infix remainders in its array of slots. As there are
exactly $T-1$ keys in-between two samples, an Infix Store allocates $T-1$ slots
in this array, each $\left\lceil \log_2 \frac{2}{\epsilon} \right\rceil$ bits
wide, to accommodate their remainders. It stores remainders of infixes sharing
the same quotient in a set of contiguous slots called a \emph{run}. Runs are
stored in increasing order of their quotients. Runs are 1-2 slots long in
expectation since the infixes are uniformly distributed, as described at the
end of
Section~\ref{sec:diva_static_sampling_keys_and_deriving_infixes},
and have between $\frac{T}{2}$ and $T$ distinct quotients.

An Infix Store employs a $(T-1)$-bit bitmap with one bit per slot, called the
\texttt{runends} bitmap, to delimit runs in an Infix Store. The $i$-th bit of
this bitmap is 1 if the $i$-th slot in the Infix Store is the last slot of a
run and is 0 otherwise.
Figure~\ref{fig:infix_store}{\nobreakdash-}B
shows an example of runs, depicted as boxes, and the \texttt{runends} bitmap
delimiting them.

In total, the \texttt{occupieds} and \texttt{runends} bitmaps consume~$2$ bits
per slot in the array of slots, or equivalently,~$2$ bits per infix.

\textbf{Matching Invariant.}
Each run in an Infix Store is associated with exactly one 1 in the
\texttt{occupieds} bitmap and one 1 in the \texttt{runends} bitmap. Since runs
are in increasing quotient order, a \emph{matching invariant} holds: the run
associated with the $i$-th~1 in the \texttt{occupieds} bitmap ends at the slot
corresponding to the $i$-th~1 in the \texttt{runends} bitmap.
Figure~\ref{fig:infix_store}{\nobreakdash-}B
illustrates this invariant with bidirectional arrows.

\textbf{Locating Runs.}
Diva locates an infix's run by first checking the bit corresponding to its
quotient in the \texttt{occupieds} bitmap. If it is 0, the infix could not have
been inserted and thus does not have a run. Otherwise, Diva locates the end of
the infix's run by leveraging the one-to-one correspondences in the matching
invariant, allowing for processing queries via a right-to-left scan of an
average of 1-2 slots. Following a similar process, Diva is able to reconstruct
all the infixes by concatenating each quotient in the \texttt{occupieds} bitmap
with the remainders in its run.

Diva speeds up locating an infix's run by employing rank and select
primitives~\cite{GQF}. Formally, let $\text{rank}(i,B)$ be the number of 1s
before the $i$-th bit in a bitmap~$B$ and $\text{select}(i,B)$ be the position
of the $i$-th 1 bit in~$B$. Diva finds the last slot in an infix~$x$'s run by
\mbox{evaluating}
$\text{select}(\text{rank}(x_q,$~\texttt{occupieds}$),$~\texttt{runends}$)$. In
other words, it determines the run number it must jump to by computing rank
over the \texttt{occupieds} bitmap and uses the result to locate the
corresponding \texttt{runends} bits by computing select.
Figure~\ref{fig:infix_store}{\nobreakdash-}B
shows an example of this derivation with $x_q=(011)_2$. Diva employs
specialized hardware instructions~\cite{GQF} to quickly apply rank and
select. 

\subsection{Query Processing}~\label{sec:diva_queries}
Diva processes a range query by searching for the query's endpoints in its \ST.
A range query intersecting at least one sample in the \ST immediately returns a
positive since that sample represents a key within the range. In contrast, a
range query that falls between two adjacent samples searches the corresponding
Infix Store and returns a positive if at least one overlapping key exists.

\textbf{Infix Store Range Queries.}
An Infix Store processes a range query over infixes $q=[l,r]$ by finding the
quotients and remainders of the endpoints and considering two
cases:

\textit{If the range spans multiple quotients} ($l_q < r_q$),
Diva checks if any quotient strictly between the endpoint quotients exists
using the occupieds bitmap. If there is such a quotient, Diva answers with a
positive, as all of the infixes in its run are strictly in the query range.
Otherwise, it checks if any remainder in the left endpoint's run is larger than
its remainder $l_r$, or if any remainder in the right endpoint's run is smaller
than its remainder $r_r$. If either condition holds, Diva reports a positive
and a negative otherwise. Query $q$ in
Figure~\ref{fig:infix_store}{\nobreakdash-}C
shows an example. Here, Diva first checks the range of bits
$[l_q+1,r_q-1]=[(010)_2,(010)_2]$ in the \mbox{\texttt{occupieds}} bitmap for
ones. As there is no bit set to 1 in that range, Diva scans the remainders
in~$l_q$ and $r_q$'s runs and compares them to $l_r$ and $r_r$. As the former
run contains a remainder equal to $l_r=(111)_2$, Diva returns a positive.

\textit{If the range spans a single quotient} ($l_q=r_q$),
Diva checks if a run corresponding to that quotient exists. If not, Diva
returns a negative. Otherwise, it scans this run for a remainder between the
remainders of the endpoints $l_r$ and $r_r$. If found, it returns a positive
and a negative otherwise. This scan is efficient, as the expected size of a run
is 1-2 slots. Query~$q'$ in
Figure~\ref{fig:infix_store}{\nobreakdash-}C
is an example. Here, $l'_q=r'_q=(110)_2$ and the bit at position 6 of the
\mbox{\texttt{occupieds}} bitmap is one. Thus, Diva searches the relevant run
for a remainder between $l'_r=(101)_2$ and $r'_r=(110)_2$. Since
$\text{rank}(6,$~\mbox{\texttt{occupieds}}$)=3$ and
$\text{select}(3,$~\mbox{\texttt{runends}}$)=6$, Diva checks the slots from
Slot~$6=(110)_2$ backward until it exhausts the run. It returns a negative, as
it does not find a remainder in the desired range.

False positives occur when query and key infixes collide, which is most likely
when queries follow the dataset's distribution. Diva's \ST learns this
distribution, enabling a low FPR in this setting and potentially a better FPR
in others.

\textbf{Point Queries.}
Point queries are equivalent to specialized range queries with equal endpoints.
Thus, they are handled by searching the \ST and the corresponding Infix Store
using the second case of the range query algorithm above.

\subsection{Trie Choice}
Diva can use any data structure supporting predecessor and successor searches
as its \ST. This structure must also associate a pointer to an Infix Store with
each sample. The samples and the extra pointers result in a negligible overhead
of~$O(L)+64$ bits per sample, where $L$ is the average key length. This
translates to a total overhead of $\frac{O(L)+64}{T} \approx 0.06+O(L/T)$ BPK.
Using a traditional trie~\cite{ART,MassTree,HOT} for this purpose yields poor
performance for predecessor and successor queries, as it incurs $O(L)$ cache
misses for~tree~traversal.

\textbf{$y$-Fast Tries.} 
Instead of a traditional trie, one can employ a $y${\nobreakdash-}Fast
trie~\cite{xyTries} to achieve an~$O(\log_2 L)$ number of cache misses for
predecessor and successor queries. A $y${\nobreakdash-}Fast trie enables this
by partitioning the ordered set of keys into groups of~$\approx L$ consecutive
keys and storing each group in a balanced binary search tree. It efficiently
accesses these groups by storing their boundary keys in a binary trie and doing
predecessor/successor searches over it.

The nodes in this binary trie are represented as a hash table storing all
prefixes of the boundary keys. Each inner node without a left (resp. right)
child stores a pointer to the predecessor (resp. successor) of the minimum
(resp. maximum) key in its subtree, allowing for fast predecessor and successor
searches. This binary trie finds the predecessor/successor of a key by
determining its longest prefix with a node in the trie via binary search. If
the found node is a leaf node, it is returned as the answer. Otherwise, the
predecessor/successor pointers are followed to answer the query. 

A $y${\nobreakdash-}Fast trie processes a predecessor/successor query by doing a
predecessor/successor search over its binary trie to find the corresponding
binary search tree and searching in it.

Searching the binary trie and the associated binary search tree each incurs
$O(\log_2 L)$ cache misses, resulting in $O(\log_2 L)$ cache misses for a
predecessor/successor query. This cost applies out-of-the-box to Diva's queries
since searching an Infix Store incurs only a constant number of cache misses
due to its small size, leaving the search cost of the trie as the main
bottleneck.

A~$y${\nobreakdash-}Fast trie supports amortized~$O(\log_2 L)$ insertions and
deletions by merging and splitting the binary search trees to maintain a size
of~$\approx L$ for each and updating its binary trie.

\textbf{Wormhole.} 
As our~$y${\nobreakdash-}Fast trie implementation, we use
Wormhole~\cite{Wormhole}, a~$y${\nobreakdash-}Fast trie with a high fanout. We
extended Wormhole with bidirectional iterators to quickly traverse the trie.

\subsection{Construction}\label{sec:diva_static_construction}
The static variant of Diva is constructed using an efficient bulk-loading
procedure where a sorted set of keys is scanned once. Here, Diva inserts every
$T$-th key into the \ST. Since it encounters the keys in-between in increasing
order, it also sees their infixes in increasing order. During this scan, Diva
sets the bits in the \texttt{occupieds} bitmap corresponding to the quotients
it sees to 1 while sequentially creates the runs. Whenever it encounters a new
quotient or fills the Infix Store, it marks the end of the current run in the
\texttt{runends} bitmap by setting the bit corresponding to the last filled
slot to 1.

This process incurs $O(N)$ cache misses for scanning the sorted key set and
inserting infixes into Infix Stores, and $O\left(N \cdot \frac{\log_2
L}{T}\right)$ cache misses for inserting every $T$-th key as a sample into the
\ST, totaling to $O\left(N \cdot \left(1+\frac{\log_2 L}{T}\right)\right)$
cache misses. Since $\frac{\log_2 L}{T} \ll 1$, we can rewrite the cost as
$O(N)$. The leftmost column in Table~\ref{tab:method_complexities}
compares the construction times of the filters and the penultimate row
expresses this cost for Diva.

Practically, Diva leverages the hardware prefetcher to the fullest during
bulk-loading, as all its memory accesses are sequential. Its CPU overhead is
also minimal, as it makes no use of hashing, save for the hashing of the \ST,
allowing Diva to enjoy a much faster construction time than its competitors. 

\subsection{Dynamicity}~\label{sec:diva_dynamic}
In addition to the static variant of Diva described in previous sections, we
provide an equally performant design that supports dynamicity in exchange
for~$\approx 1$ BPK more memory. This variant of Diva employs four core
techniques, some of which draw on prior works on dynamic
filters~\cite{InfiniFilter,AlephFilter}. To better illustrate these techniques,
we have reorganized the presentation of an Infix Store's metadata bitmaps and
its array of slots in Figure~\ref{fig:dynamic_infix_store}.

\begin{figure}
    \centering
    \pgfdeclarelayer{background layer}
    \pgfsetlayers{background layer,main}
    \begin{tikzpicture}
        \def\slotw{0.7}
        \def\sloth{0.35}
        \def\occupiedcnt{8}
        \def\slotcnt{8}
        \def\storeoffx{-0.8}
        \def\storegap{0.1}
        \def\circlediam{0.6mm}

        \node[inner sep=0pt] (occupieds_label) at (\storeoffx-0.72,-0.5*\sloth-\storegap) {\small \texttt{occupieds}};
        \draw[fill=gray!20] (\storeoffx,-\storegap) rectangle (\storeoffx+\occupiedcnt*\slotw,-\storegap-\sloth);
        \foreach \x in {1, ..., \occupiedcnt} {
            \draw[dotted] (\storeoffx+\slotw*\x,-\storegap-\sloth) -- (\storeoffx+\slotw*\x,-\storegap);
        }
        \foreach \i in {1, ..., \occupiedcnt} {
            \def\x{\storeoffx+\i*\slotw-0.5*\slotw}
            \def\y{-0.5*\sloth-\storegap}
            \ifnum \i = 1 {
                \node[inner sep=0pt] (occ_0) at (\x,\y) {\small 1};
            }
            \else \ifnum \i = 3 {
                \node[inner sep=2.25pt] (occ_1) at (\x,\y) {\small 1};
            }
            \else \ifnum \i = 4 {
                \node[inner sep=0pt] (occ_2) at (\x,\y) {\small 0};
            }
            \else \ifnum \i = 5 {
                \node[inner sep=2.25pt] (occ_3) at (\x,\y) {\small 1};
            }
            \else \ifnum \i = 7 {
                \node[inner sep=0pt] (occ_4) at (\x,\y) {\small 1};
            }
            \else \ifnum \i = 8 {
                \node[inner sep=0pt] (occ_5) at (\x,\y) {\small 1};
            }
            \else {
                \node at (\x,\y) {\small 0};
            }
            \fi \fi \fi \fi \fi \fi
        }

        \node[inner sep=0pt] (slots_label) at (\storeoffx-0.45,-1.5*\sloth-2*\storegap) {\small Slots};
        \draw (\storeoffx,-1*\sloth-2*\storegap) rectangle (\storeoffx+\slotcnt*\slotw,-2*\sloth-2*\storegap);
        \foreach \x in {1, ..., \slotcnt} {
            \draw[dotted] (\storeoffx+\slotw*\x,-2*\sloth-2*\storegap) -- (\storeoffx+\slotw*\x,-1*\sloth-2*\storegap);
        }
        \node at (\storeoffx+0*\slotw+0.5*\slotw,-1.5*\sloth-2.15*\storegap) {\small 0\underline{100}};
        \node at (\storeoffx+1*\slotw+0.5*\slotw,-1.5*\sloth-2.15*\storegap) {\small 01\underline{10}};
        \draw[very thick] (\storeoffx+0*\slotw,-2*\sloth-2*\storegap) rectangle (\storeoffx+2*\slotw,-1*\sloth-2*\storegap);

        \node at (\storeoffx+2*\slotw+0.5*\slotw,-1.5*\sloth-2.15*\storegap) {\small 01\underline{10}};
        \draw[very thick] (\storeoffx+2*\slotw,-2*\sloth-2*\storegap) rectangle (\storeoffx+3*\slotw,-1*\sloth-2*\storegap);

        \node at (\storeoffx+4*\slotw+0.5*\slotw,-1.5*\sloth-2.15*\storegap) {\small 00\underline{10}};
        \node at (\storeoffx+5*\slotw+0.5*\slotw,-1.5*\sloth-2.15*\storegap) {\small 111\underline{1}};
        \draw[very thick] (\storeoffx+4*\slotw,-2*\sloth-2*\storegap) rectangle (\storeoffx+6*\slotw,-1*\sloth-2*\storegap);

        \node at (\storeoffx+6*\slotw+0.5*\slotw,-1.5*\sloth-2.15*\storegap) {\small 010\underline{1}};
        \draw[very thick] (\storeoffx+6*\slotw,-2*\sloth-2*\storegap) rectangle (\storeoffx+7*\slotw,-1*\sloth-2*\storegap);

        \node at (\storeoffx+7*\slotw+0.5*\slotw,-1.5*\sloth-2.15*\storegap) {\small \underline{1000}};
        \draw[very thick] (\storeoffx+7*\slotw,-2*\sloth-2*\storegap) rectangle (\storeoffx+8*\slotw,-1*\sloth-2*\storegap);

        \node[inner sep=0pt] (runends_label) at (\storeoffx-0.6,-2.5*\sloth-3*\storegap) {\small \texttt{runends}};
        \draw[fill=gray!20] (\storeoffx,-2*\sloth-3*\storegap) rectangle (\storeoffx+\slotcnt*\slotw,-3*\sloth-3*\storegap);
        \foreach \x in {1, ..., \slotcnt} {
            \draw[dotted] (\storeoffx+\slotw*\x,-3*\sloth-3*\storegap) -- (\storeoffx+\slotw*\x,-2*\sloth-3*\storegap);
        }
        \foreach \i in {1, ..., \slotcnt} {
            \def\x{\storeoffx+\i*\slotw-0.5*\slotw}
            \def\y{-2.5*\sloth-3*\storegap}
            \ifnum \i = 2 {
                \node[inner sep=0pt] (run_0) at (\x,\y) {\small 1};
            }
            \else \ifnum \i = 3 {
                \node[inner sep=0pt] (run_1) at (\x,\y) {\small 1};
            }
            \else \ifnum \i = 6 {
                \node[inner sep=0pt] (run_2) at (\x,\y) {\small 1};
            }
            \else \ifnum \i = 7 {
                \node[inner sep=0pt] (run_3) at (\x,\y) {\small 1};
            }
            \else \ifnum \i = 8 {
                \node[inner sep=0pt] (run_3) at (\x,\y) {\small 1};
            }
            \else {
                \node at (\x,\y) {\small 0};
            }
            \fi \fi \fi \fi \fi
        }

        \node[inner sep=0pt] at (\storeoffx+0.5*\slotw,-3.5*\sloth-3*\storegap) {\small 000};
        \node[inner sep=0pt] at (\storeoffx+1.5*\slotw,-3.5*\sloth-3*\storegap) {\small 001};
        \node[inner sep=0pt] at (\storeoffx+2.5*\slotw,-3.5*\sloth-3*\storegap) {\small 010};
        \node[inner sep=0pt] at (\storeoffx+3.5*\slotw,-3.5*\sloth-3*\storegap) {\small 011};
        \node[inner sep=0pt] at (\storeoffx+4.5*\slotw,-3.5*\sloth-3*\storegap) {\small 100};
        \node[inner sep=0pt] at (\storeoffx+5.5*\slotw,-3.5*\sloth-3*\storegap) {\small 101};
        \node[inner sep=0pt] at (\storeoffx+6.5*\slotw,-3.5*\sloth-3*\storegap) {\small 110};
        \node[inner sep=0pt] at (\storeoffx+7.5*\slotw,-3.5*\sloth-3*\storegap) {\small 111};

        \node[inner sep=2pt,align=center] (si) at (\storeoffx-2.24*\slotw,-2*\sloth-2*\storegap) {\small \textbf{0} \\[-4pt] \small \textbf{1} \\[-4pt] \small \textbf{0} \\[-4pt] \small 0 \\[-4pt] \small 0 \\[-4pt] \small 0 \\[-4pt] \small $\vdots$};
        \node[inner sep=2pt,rotate=90,above left=-2.25mm and 0mm of si] (si_label) {\small Predecessor};
        \node[inner sep=2pt,align=center] (si1) at (\storeoffx+\slotcnt*\slotw+2.75*\storegap,-2*\sloth-2*\storegap) {\small \textbf{0} \\[-4pt] \small \textbf{1} \\[-4pt] \small \textbf{0} \\[-4pt] \small 1 \\[-4pt] \small 1 \\[-4pt] \small 1 \\[-4pt] \small $\vdots$};
        \node[inner sep=2pt,rotate=90,below right=-3.5mm and -0.5mm of si1] (si1_label) {\small Successor};

        \node[inner sep=0pt,anchor=west] (fig_label) at (\storeoffx-2.6*\slotw,\sloth-0.5*\storegap) {\small \textbf{A) Unstretched Infix Store}};
        \node[inner sep=0pt,anchor=west] (centering_dummy) at (\storeoffx+\slotcnt*\slotw+1.0*\slotw,\sloth+\storegap) { };
    \end{tikzpicture}
    \vspace{4pt}

    \pgfdeclarelayer{background layer}
    \pgfsetlayers{background layer,main}
    \begin{tikzpicture}
        \def\slotw{0.7}
        \def\sloth{0.35}
        \def\occupiedcnt{8}
        \def\slotcnt{10}
        \def\storeoffx{-1.5}
        \def\storegap{0.1}
        \def\circlediam{0.6mm}

        \draw[fill=gray!20] (\storeoffx,-\storegap) rectangle (\storeoffx+\occupiedcnt*\slotw,-\storegap-\sloth);
        \foreach \x in {1, ..., \occupiedcnt} {
            \draw[dotted] (\storeoffx+\slotw*\x,-\storegap-\sloth) -- (\storeoffx+\slotw*\x,-\storegap);
        }
        \foreach \i in {1, ..., \occupiedcnt} {
            \def\x{\storeoffx+\i*\slotw-0.5*\slotw}
            \def\y{-0.5*\sloth-\storegap}
            \ifnum \i = 1 {
                \node[inner sep=0pt] (occ_0) at (\x,\y) {\small 1};
            }
            \else \ifnum \i = 3 {
                \node[inner sep=0pt] (occ_1) at (\x,\y) {\small 1};
            }
            \else \ifnum \i = 5 {
                \node[inner sep=0pt] (occ_2) at (\x,\y) {\small 1};
            }
            \else \ifnum \i = 7 {
                \node[inner sep=0pt] (occ_3) at (\x,\y) {\small 1};
            }
            \else \ifnum \i = 8 {
                \node[inner sep=0pt] (occ_4) at (\x,\y) {\small 1};
            }
            \else {
                \node at (\x,\y) {\small 0};
            }
            \fi \fi \fi \fi \fi
        }

        \draw (\storeoffx,-1*\sloth-3*\storegap) rectangle (\storeoffx+\slotcnt*\slotw,-2*\sloth-3*\storegap);
        \foreach \x in {1, ..., \slotcnt} {
            \draw[dotted] (\storeoffx+\slotw*\x,-2*\sloth-3*\storegap) -- (\storeoffx+\slotw*\x,-1*\sloth-3*\storegap);
        }
        \node at (\storeoffx+0*\slotw+0.5*\slotw,-1.5*\sloth-3.15*\storegap) {\small 0\underline{100}};
        \node at (\storeoffx+1*\slotw+0.5*\slotw,-1.5*\sloth-3.15*\storegap) {\small 01\underline{10}};
        \draw[very thick] (\storeoffx+0*\slotw,-2*\sloth-3*\storegap) rectangle (\storeoffx+2*\slotw,-1*\sloth-3*\storegap);

        \node at (\storeoffx+2*\slotw+0.5*\slotw,-1.5*\sloth-3.15*\storegap) {\small 01\underline{10}};
        \draw[very thick] (\storeoffx+2*\slotw,-2*\sloth-3*\storegap) rectangle (\storeoffx+3*\slotw,-1*\sloth-3*\storegap);

        \node at (\storeoffx+5*\slotw+0.5*\slotw,-1.5*\sloth-3.15*\storegap) {\small 00\underline{10}};
        \node at (\storeoffx+6*\slotw+0.5*\slotw,-1.5*\sloth-3.15*\storegap) {\small 111\underline{1}};
        \draw[very thick] (\storeoffx+5*\slotw,-2*\sloth-3*\storegap) rectangle (\storeoffx+7*\slotw,-1*\sloth-3*\storegap);

        \node at (\storeoffx+7*\slotw+0.5*\slotw,-1.5*\sloth-3.15*\storegap) {\small 010\underline{1}};
        \draw[very thick] (\storeoffx+7*\slotw,-2*\sloth-3*\storegap) rectangle (\storeoffx+8*\slotw,-1*\sloth-3*\storegap);

        \node at (\storeoffx+9*\slotw+0.5*\slotw,-1.5*\sloth-3.15*\storegap) {\small \underline{1000}};
        \draw[very thick] (\storeoffx+9*\slotw,-2*\sloth-3*\storegap) rectangle (\storeoffx+10*\slotw,-1*\sloth-3*\storegap);

        \draw[fill=gray!20] (\storeoffx,-2*\sloth-4*\storegap) rectangle (\storeoffx+\slotcnt*\slotw,-3*\sloth-4*\storegap);
        \foreach \x in {1, ..., \slotcnt} {
            \draw[dotted] (\storeoffx+\slotw*\x,-3*\sloth-4*\storegap) -- (\storeoffx+\slotw*\x,-2*\sloth-4*\storegap);
        }
        \foreach \i in {1, ..., \slotcnt} {
            \def\x{\storeoffx+\i*\slotw-0.5*\slotw}
            \def\y{-2.5*\sloth-4*\storegap}
            \ifnum \i = 2 {
                \node[inner sep=0pt] (run_0) at (\x,\y) {\small 1};
            }
            \else \ifnum \i = 3 {
                \node[inner sep=0pt] (run_2) at (\x,\y) {\small 1};
            }
            \else \ifnum \i = 7 {
                \node[inner sep=0pt] (run_6) at (\x,\y) {\small 1};
            }
            \else \ifnum \i = 8 {
                \node[inner sep=0pt] (run_7) at (\x,\y) {\small 1};
            }
            \else \ifnum \i = 10 {
                \node[inner sep=0pt] (run_9) at (\x,\y) {\small 1};
            }
            \else {
                \node at (\x,\y) {\small 0};
            }
            \fi \fi \fi \fi \fi
        }

        \node[inner sep=0pt] at (\storeoffx+0.5*\slotw,-3.5*\sloth-4*\storegap) {\small 0000};
        \node[inner sep=0pt] at (\storeoffx+1.5*\slotw,-3.5*\sloth-4*\storegap) {\small 0001};
        \node[inner sep=0pt] at (\storeoffx+2.5*\slotw,-3.5*\sloth-4*\storegap) {\small 0010};
        \node[inner sep=0pt] at (\storeoffx+3.5*\slotw,-3.5*\sloth-4*\storegap) {\small 0011};
        \node[inner sep=0pt] at (\storeoffx+4.5*\slotw,-3.5*\sloth-4*\storegap) {\small 0100};
        \node[inner sep=0pt] at (\storeoffx+5.5*\slotw,-3.5*\sloth-4*\storegap) {\small 0101};
        \node[inner sep=0pt] at (\storeoffx+6.5*\slotw,-3.5*\sloth-4*\storegap) {\small 0110};
        \node[inner sep=0pt] at (\storeoffx+7.5*\slotw,-3.5*\sloth-4*\storegap) {\small 0111};
        \node[inner sep=0pt] at (\storeoffx+8.5*\slotw,-3.5*\sloth-4*\storegap) {\small 1000};
        \node[inner sep=0pt] at (\storeoffx+9.5*\slotw,-3.5*\sloth-4*\storegap) {\small 1001};

        \draw[-stealth] (\storeoffx+0.5*\slotw,-1*\sloth-1*\storegap) -- (\storeoffx+0.5*\slotw,-1*\sloth-3*\storegap);
        \draw[-stealth] (\storeoffx+1.5*\slotw,-1*\sloth-1*\storegap) -- (\storeoffx+1.5*\slotw,-1*\sloth-3*\storegap);
        \draw[-stealth] (\storeoffx+2.5*\slotw,-1*\sloth-1*\storegap) -- (\storeoffx+2.5*\slotw,-1*\sloth-3*\storegap);
        \draw[-stealth] (\storeoffx+3.5*\slotw,-1*\sloth-1*\storegap) -- (\storeoffx+3.5*\slotw,-1*\sloth-3*\storegap);
        \draw[-stealth] (\storeoffx+4.5*\slotw,-1*\sloth-1*\storegap) -- (\storeoffx+5.5*\slotw,-1*\sloth-3*\storegap);
        \draw[-stealth] (\storeoffx+5.5*\slotw,-1*\sloth-1*\storegap) -- (\storeoffx+6.5*\slotw,-1*\sloth-3*\storegap);
        \draw[-stealth] (\storeoffx+6.5*\slotw,-1*\sloth-1*\storegap) -- (\storeoffx+7.5*\slotw,-1*\sloth-3*\storegap);
        \draw[-stealth] (\storeoffx+7.5*\slotw,-1*\sloth-1*\storegap) -- (\storeoffx+9.5*\slotw,-1*\sloth-3*\storegap);

        \node[inner sep=2pt,align=center] (si) at (\storeoffx-2.25*\storegap,-2*\sloth-2*\storegap) {\small \textbf{0} \\[-4pt] \small \textbf{1} \\[-4pt] \small \textbf{0} \\[-4pt] \small 0 \\[-4pt] \small 0 \\[-4pt] \small 0 \\[-4pt] \small $\vdots$};
        \node[inner sep=2pt,align=center] (si1) at (\storeoffx+\slotcnt*\slotw+2.25*\storegap,-2*\sloth-2*\storegap) {\small \textbf{0} \\[-4pt] \small \textbf{1} \\[-4pt] \small \textbf{0} \\[-4pt] \small 1 \\[-4pt] \small 1 \\[-4pt] \small 1 \\[-4pt] \small $\vdots$};
        \node[inner sep=1pt] (k) at (\storeoffx+7.5*\slotw,\storegap) {\small $k=0101000\dots$};
        \draw[-stealth,thick] (k.west) -| node[pos=0.25,above=-2pt] (split_label) {\small Split} (\storeoffx+4.5*\slotw,-\storegap);

        \node[inner sep=0pt,anchor=west] (fig_label) at (\storeoffx-0.6*\slotw,\sloth-0*\storegap) {\small \textbf{B) Stretched Infix Store}};
        \node[inner sep=0pt,anchor=west] (centering_dummy) at (\storeoffx+\slotcnt*\slotw+0.5*\slotw,\sloth+\storegap) { };
    \end{tikzpicture}
    \vspace{4pt}

    \pgfdeclarelayer{background layer}
    \pgfsetlayers{background layer,main}
    \begin{tikzpicture}
        \def\slotw{0.7}
        \def\sloth{0.35}
        \def\occupiedcnt{5}
        \def\slotcnt{6}
        \def\storeoffx{-1.5}
        \def\storegap{0.1}
        \def\circlediam{0.6mm}

        \draw[fill=gray!20] (\storeoffx,-\storegap) rectangle (\storeoffx+\occupiedcnt*\slotw,-\storegap-\sloth);
        \foreach \x in {1, ..., \occupiedcnt} {
            \draw[dotted] (\storeoffx+\slotw*\x,-\storegap-\sloth) -- (\storeoffx+\slotw*\x,-\storegap);
        }
        \foreach \i in {1, ..., \occupiedcnt} {
            \def\x{\storeoffx+\i*\slotw-0.5*\slotw}
            \def\y{-0.5*\sloth-\storegap}
            \ifnum \i = 1 {
                \node[inner sep=0pt] (occ_0) at (\x,\y) {\small 1};
            }
            \else \ifnum \i = 3 {
                \node[inner sep=0pt] (occ_2) at (\x,\y) {\small 1};
            }
            \else {
                \node at (\x,\y) {\small 0};
            }
            \fi \fi
        }

        \draw (\storeoffx,-1*\sloth-3*\storegap) rectangle (\storeoffx+\slotcnt*\slotw,-2*\sloth-3*\storegap);
        \foreach \x in {1, ..., \slotcnt} {
            \draw[dotted] (\storeoffx+\slotw*\x,-2*\sloth-3*\storegap) -- (\storeoffx+\slotw*\x,-1*\sloth-3*\storegap);
        }
        \node at (\storeoffx+0*\slotw+0.5*\slotw,-1.5*\sloth-3.15*\storegap) {\small 0\underline{100}};
        \node at (\storeoffx+1*\slotw+0.5*\slotw,-1.5*\sloth-3.15*\storegap) {\small 01\underline{10}};
        \draw[very thick] (\storeoffx+0*\slotw,-2*\sloth-3*\storegap) rectangle (\storeoffx+2*\slotw,-1*\sloth-3*\storegap);

        \node at (\storeoffx+2*\slotw+0.5*\slotw,-1.5*\sloth-3.15*\storegap) {\small 01\underline{10}};
        \draw[very thick] (\storeoffx+2*\slotw,-2*\sloth-3*\storegap) rectangle (\storeoffx+3*\slotw,-1*\sloth-3*\storegap);

        \draw[fill=gray!20] (\storeoffx,-2*\sloth-4*\storegap) rectangle (\storeoffx+\slotcnt*\slotw,-3*\sloth-4*\storegap);
        \foreach \x in {1, ..., \slotcnt} {
            \draw[dotted] (\storeoffx+\slotw*\x,-3*\sloth-4*\storegap) -- (\storeoffx+\slotw*\x,-2*\sloth-4*\storegap);
        }
        \foreach \i in {1, ..., \slotcnt} {
            \def\x{\storeoffx+\i*\slotw-0.5*\slotw}
            \def\y{-2.5*\sloth-4*\storegap}
            \ifnum \i = 2 {
                \node[inner sep=0pt] (run_1) at (\x,\y) {\small 1};
            }
            \else \ifnum \i = 3 {
                \node[inner sep=0pt] (run_2) at (\x,\y) {\small 1};
            }
            \else {
                \node at (\x,\y) {\small 0};
            }
            \fi \fi
        }

        \node[inner sep=0pt] at (\storeoffx+0.5*\slotw,-3.5*\sloth-4*\storegap) {\small 000};
        \node[inner sep=0pt] at (\storeoffx+1.5*\slotw,-3.5*\sloth-4*\storegap) {\small 001};
        \node[inner sep=0pt] at (\storeoffx+2.5*\slotw,-3.5*\sloth-4*\storegap) {\small 010};
        \node[inner sep=0pt] at (\storeoffx+3.5*\slotw,-3.5*\sloth-4*\storegap) {\small 011};
        \node[inner sep=0pt] at (\storeoffx+4.5*\slotw,-3.5*\sloth-4*\storegap) {\small 100};
        \node[inner sep=0pt] at (\storeoffx+5.5*\slotw,-3.5*\sloth-4*\storegap) {\small 101};

        \draw[-stealth] (\storeoffx+0.5*\slotw,-1*\sloth-1*\storegap) -- (\storeoffx+0.5*\slotw,-1*\sloth-3*\storegap);
        \draw[-stealth] (\storeoffx+1.5*\slotw,-1*\sloth-1*\storegap) -- (\storeoffx+1.5*\slotw,-1*\sloth-3*\storegap);
        \draw[-stealth] (\storeoffx+2.5*\slotw,-1*\sloth-1*\storegap) -- (\storeoffx+2.5*\slotw,-1*\sloth-3*\storegap);
        \draw[-stealth] (\storeoffx+3.5*\slotw,-1*\sloth-1*\storegap) -- (\storeoffx+3.5*\slotw,-1*\sloth-3*\storegap);
        \draw[-stealth] (\storeoffx+4.5*\slotw,-1*\sloth-1*\storegap) -- (\storeoffx+5.5*\slotw,-1*\sloth-3*\storegap);

        \node[inner sep=2pt,align=center] (si) at (\storeoffx-2.75*\storegap,-2*\sloth-2*\storegap) {\small \textbf{0} \\[-4pt] \small \textbf{1} \\[-4pt] \small \textbf{0} \\[-4pt] \small 0 \\[-4pt] \small 0 \\[-4pt] \small 0 \\[-4pt] \small $\vdots$};
        \node[inner sep=2pt,align=center] (si1) at (\storeoffx+\slotcnt*\slotw+2.75*\storegap,-2*\sloth-2*\storegap) {\small \textbf{0} \\[-4pt] \small \textbf{1} \\[-4pt] \small \textbf{0} \\[-4pt] \small 1 \\[-4pt] \small 0 \\[-4pt] \small 0 \\[-4pt] \small $\vdots$};
        \node[inner sep=2pt,align=center,right=-0.1mm of si1] (k_label) {\small $k$};

        \node[inner sep=0pt,anchor=west] (fig_label) at (\storeoffx-2.75*\slotw,\sloth-1*\storegap) {\small \textbf{C) Left Infix Store After Split}};
        \node[inner sep=0pt,anchor=west] (centering_dummy) at (\storeoffx+\slotcnt*\slotw+2.5*\slotw,\sloth+\storegap) { };
    \end{tikzpicture}
    \vspace{2pt}

    \pgfdeclarelayer{background layer}
    \pgfsetlayers{background layer,main}
    \begin{tikzpicture}
        \def\slotw{0.7}
        \def\sloth{0.35}
        \def\occupiedcnt{8}
        \def\slotcnt{6}
        \def\storeoffx{-1.5}
        \def\storegap{0.1}
        \def\circlediam{0.6mm}

        \draw[fill=gray!20] (\storeoffx,-\storegap) rectangle (\storeoffx+\occupiedcnt*\slotw,-\storegap-\sloth);
        \foreach \x in {1, ..., \occupiedcnt} {
            \draw[dotted] (\storeoffx+\slotw*\x,-\storegap-\sloth) -- (\storeoffx+\slotw*\x,-\storegap);
        }
        \foreach \i in {1, ..., \occupiedcnt} {
            \def\x{\storeoffx+\i*\slotw-0.5*\slotw}
            \def\y{-0.5*\sloth-\storegap}
            \ifnum \i = 1 {
                \node[inner sep=0pt] (occ_0) at (\x,\y) {\small 1};
            }
            \else \ifnum \i = 2 {
                \node[inner sep=0pt] (occ_1) at (\x,\y) {\small 1};
            }
            \else \ifnum \i = 5 {
                \node[inner sep=0pt] (occ_2) at (\x,\y) {\small 1};
            }
            \else \ifnum \i = 7 {
                \node[inner sep=0pt] (occ_3) at (\x,\y) {\small 1};
            }
            \else \ifnum \i = 8 {
                \node[inner sep=0pt] (occ_4) at (\x,\y) {\small 1};
            }
            \else {
                \node at (\x,\y) {\small 0};
            }
            \fi \fi \fi \fi \fi
        }

        \draw (\storeoffx,-1*\sloth-3*\storegap) rectangle (\storeoffx+\slotcnt*\slotw,-2*\sloth-3*\storegap);
        \foreach \x in {1, ..., \slotcnt} {
            \draw[dotted] (\storeoffx+\slotw*\x,-2*\sloth-3*\storegap) -- (\storeoffx+\slotw*\x,-1*\sloth-3*\storegap);
        }
        \node at (\storeoffx+0*\slotw+0.5*\slotw,-1.5*\sloth-3.15*\storegap) {\small 0\underline{100}};
        \draw[very thick] (\storeoffx+0*\slotw,-2*\sloth-3*\storegap) rectangle (\storeoffx+1*\slotw,-1*\sloth-3*\storegap);

        \node at (\storeoffx+1*\slotw+0.5*\slotw,-1.5*\sloth-3.15*\storegap) {\small 11\underline{10}};
        \draw[very thick] (\storeoffx+1*\slotw,-2*\sloth-3*\storegap) rectangle (\storeoffx+2*\slotw,-1*\sloth-3*\storegap);

        \node at (\storeoffx+2*\slotw+0.5*\slotw,-1.5*\sloth-3.15*\storegap) {\small 10\underline{10}};
        \draw[very thick] (\storeoffx+2*\slotw,-2*\sloth-3*\storegap) rectangle (\storeoffx+3*\slotw,-1*\sloth-3*\storegap);

        \node at (\storeoffx+4*\slotw+0.5*\slotw,-1.5*\sloth-3.15*\storegap) {\small \underline{1000}};
        \draw[very thick] (\storeoffx+4*\slotw,-2*\sloth-3*\storegap) rectangle (\storeoffx+5*\slotw,-1*\sloth-3*\storegap);
        \begin{pgfonlayer}{background layer}
            \draw[white,fill=gray!20] (\storeoffx+4*\slotw,-2*\sloth-3*\storegap) rectangle (\storeoffx+5*\slotw,-1*\sloth-3*\storegap);
        \end{pgfonlayer}

        \node at (\storeoffx+5*\slotw+0.5*\slotw,-1.5*\sloth-3.15*\storegap) {\small \underline{1000}};
        \draw[very thick] (\storeoffx+5*\slotw,-2*\sloth-3*\storegap) rectangle (\storeoffx+6*\slotw,-1*\sloth-3*\storegap);
        \begin{pgfonlayer}{background layer}
            \draw[white,fill=gray!20] (\storeoffx+5*\slotw,-2*\sloth-3*\storegap) rectangle (\storeoffx+6*\slotw,-1*\sloth-3*\storegap);
        \end{pgfonlayer}

        \draw[fill=gray!20] (\storeoffx,-2*\sloth-4*\storegap) rectangle (\storeoffx+\slotcnt*\slotw,-3*\sloth-4*\storegap);
        \foreach \x in {1, ..., \slotcnt} {
            \draw[dotted] (\storeoffx+\slotw*\x,-3*\sloth-4*\storegap) -- (\storeoffx+\slotw*\x,-2*\sloth-4*\storegap);
        }
        \foreach \i in {1, ..., \slotcnt} {
            \def\x{\storeoffx+\i*\slotw-0.5*\slotw}
            \def\y{-2.5*\sloth-4*\storegap}
            \ifnum \i = 1 {
                \node[inner sep=0pt] (run_0) at (\x,\y) {\small 1};
            }
            \else \ifnum \i = 2 {
                \node[inner sep=0pt] (run_1) at (\x,\y) {\small 1};
            }
            \else \ifnum \i = 3 {
                \node[inner sep=0pt] (run_2) at (\x,\y) {\small 1};
            }
            \else \ifnum \i = 5 {
                \node[inner sep=0pt] (run_4) at (\x,\y) {\small 1};
            }
            \else \ifnum \i = 6 {
                \node[inner sep=0pt] (run_5) at (\x,\y) {\small 1};
            }
            \else {
                \node at (\x,\y) {\small 0};
            }
            \fi \fi \fi \fi \fi
        }

        \node[inner sep=0pt] (duplicated_label) at (\storeoffx+7.1*\slotw,-1.55*\sloth-3*\storegap) {\small Duplicated};

        \node[inner sep=0pt] at (\storeoffx+0.5*\slotw,-3.5*\sloth-4*\storegap) {\small 000};
        \node[inner sep=0pt] at (\storeoffx+1.5*\slotw,-3.5*\sloth-4*\storegap) {\small 001};
        \node[inner sep=0pt] at (\storeoffx+2.5*\slotw,-3.5*\sloth-4*\storegap) {\small 010};
        \node[inner sep=0pt] at (\storeoffx+3.5*\slotw,-3.5*\sloth-4*\storegap) {\small 011};
        \node[inner sep=0pt] at (\storeoffx+4.5*\slotw,-3.5*\sloth-4*\storegap) {\small 100};
        \node[inner sep=0pt] at (\storeoffx+5.5*\slotw,-3.5*\sloth-4*\storegap) {\small 101};
        \node[inner sep=0pt] at (\storeoffx+6.5*\slotw,-3.5*\sloth-4*\storegap) {\small 110};
        \node[inner sep=0pt] at (\storeoffx+7.5*\slotw,-3.5*\sloth-4*\storegap) {\small 111};

        \draw[-stealth] (\storeoffx+0.5*\slotw,-1*\sloth-1*\storegap) -- (\storeoffx+0.5*\slotw,-1*\sloth-3*\storegap);
        \draw[-stealth] (\storeoffx+1.5*\slotw,-1*\sloth-1*\storegap) -- (\storeoffx+0.5*\slotw,-1*\sloth-3*\storegap);
        \draw[-stealth] (\storeoffx+2.5*\slotw,-1*\sloth-1*\storegap) -- (\storeoffx+1.5*\slotw,-1*\sloth-3*\storegap);
        \draw[-stealth] (\storeoffx+3.5*\slotw,-1*\sloth-1*\storegap) -- (\storeoffx+2.5*\slotw,-1*\sloth-3*\storegap);
        \draw[-stealth] (\storeoffx+4.5*\slotw,-1*\sloth-1*\storegap) -- (\storeoffx+2.5*\slotw,-1*\sloth-3*\storegap);
        \draw[-stealth] (\storeoffx+5.5*\slotw,-1*\sloth-1*\storegap) -- (\storeoffx+3.5*\slotw,-1*\sloth-3*\storegap);
        \draw[-stealth] (\storeoffx+6.5*\slotw,-1*\sloth-1*\storegap) -- (\storeoffx+4.5*\slotw,-1*\sloth-3*\storegap);
        \draw[-stealth] (\storeoffx+7.5*\slotw,-1*\sloth-1*\storegap) -- (\storeoffx+5.5*\slotw,-1*\sloth-3*\storegap);

        \node[inner sep=2pt,align=center] (si) at (\storeoffx-3.25*\storegap,-2*\sloth-2*\storegap) {\small \textbf{0} \\[-4pt] \small \textbf{1} \\[-4pt] \small \textbf{0} \\[-4pt] \small \textbf{1} \\[-4pt] \small 0 \\[-4pt] \small 0 \\[-4pt] \small $\vdots$};
        \node[inner sep=2pt,align=center,left=-0.1mm of si] (k_label) {\small $k$};
        \node[inner sep=2pt,align=center] (si1) at (\storeoffx+\occupiedcnt*\slotw+3.5*\storegap,-2*\sloth-2*\storegap) {\small \textbf{0} \\[-4pt] \small \textbf{1} \\[-4pt] \small \textbf{0} \\[-4pt] \small \textbf{1} \\[-4pt] \small 1 \\[-4pt] \small 1 \\[-4pt] \small $\vdots$};

        \node[inner sep=0pt,anchor=west] (fig_label) at (\storeoffx-1.67*\slotw,\sloth-1*\storegap) {\small \textbf{D) Right Infix Store After Split}};
        \node[inner sep=0pt,anchor=west] (centering_dummy) at (\storeoffx+\slotcnt*\slotw+3.5*\slotw,\sloth+\storegap) { };
    \end{tikzpicture}
    \vspace{-2mm}
    \caption{When an Infix Store runs out of overprovisioned space~(Part~A),
    Diva stretches it to make room for new insertions~(Part~B). The insertion
    of a new sample~$k$ into an Infix Store~(Part~B) splits it into a
    left~(Part~C) and a right~(Part~D) Infix Store. As a result of this split,
    old infixes may lose bits, elongating the unary paddings (underlined) and
    duplicating infixes with empty remainders. The predecessor and successor's
    longest common prefix is highlighted in bold.}
    \label{fig:dynamic_infix_store}
    \vspace{-5mm}
\end{figure}
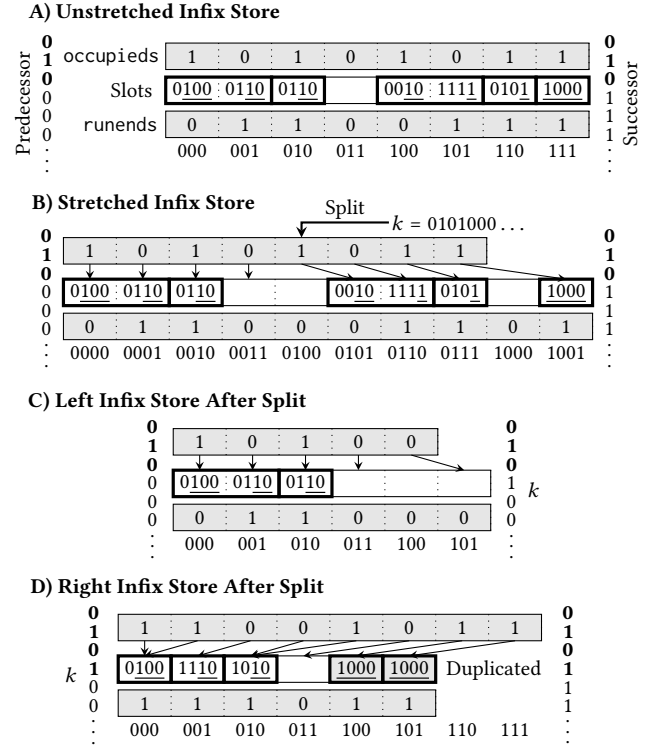

\textbf{(1)~Overprovisioning Infix Stores.}
Diva overprovisions slots within each Infix Store to absorb insertions, as
depicted in
Figure~\ref{fig:dynamic_infix_store}{\nobreakdash-}A.
It inserts a new infix by mapping it to the slot at the offset equal to its
quotient. Doing so uniformly scatters the runs across the Infix Store while
leaving space in-between them. Diva stores each infix's run in its mapped slot
if there is space. If not, it makes room for the new infix by shifting other
runs to the right until hitting an empty slot (and shifting to the left if
there is no room to the right), similarly to Robin-Hood
Hashing~\cite{RobinHoodHashing,GQF}. The amount of shifting is moderate, since
empty slots are spread evenly across the Infix Store.

\textbf{(2)~Stretching Infix Stores.} 
Once an Infix Store is~95\% full, Diva expands it by a small factor ($\approx
1.05\times$). It does this using a technique we developed for Zeno
Filter~\cite{Zeno} called ``\emph{stretching},'' whereby we stretch the
distance between adjacent runs using a linear function to create more empty
slots between them and accommodate more insertions.\footnote{In an Infix Store
with~$n$ slots, Diva maps an infix~$x$ with a quotient of~$x_q$ to the slot at
offset~$\lfloor x_q \cdot n / (r-l) \rfloor$, where $l$ and $r$ denote the
value of the bits in the predecessor and successor corresponding to the
quotient~$x_q$.} For example, stretching the Infix Store in
Figure~\ref{fig:dynamic_infix_store}{\nobreakdash-}A
results in
Figure~\ref{fig:dynamic_infix_store}{\nobreakdash-}B,
which has two more empty slots. Stretching ensures a space utilization of at
least~90\% overall, as opposed to the~50\% space utilization of other dynamic
filters or hash tables that expand by a factor
of~2~\cite{BeyondBloomTutorial,GQF,InfiniFilter,AlephFilter,Memento,Aeris}.

\textbf{(3)~Sampling Keys and Splitting Infix Stores.}
After~$\approx T$ insertions into an Infix Store (and multiple stretching
operations), the Infix Store should ``\emph{split}'' in two, with a new sample
in-between added to the \ST. Splitting keeps runs short and queries efficient.
It also refines the \ST's approximation of the data's distribution,
particularly for the parts of the key space that become denser.

An invariant of our design is to only store full, untruncated keys within the
\ST to allow for refining the key space in-between samples indefinitely. This
invariant requires that only new insertions into the filter can be picked as
samples, since all previously inserted keys have been truncated. Yet, this
potentially leads to a vulnerability: deterministically inserting
every~$T${\nobreakdash-}th key as a sample into the \ST would allow an
adversary to create pathologically uneven splits and degrade performance. We
address this risk by randomly picking each new insertion as a sample with a
probability of~$\frac{1}{T}$. This strategy ensures that, in expectation,
every~$T${\nobreakdash-}th key inserted into an Infix Store splits it and is
taken as a sample. As such, the \ST is still kept small ($\approx \frac{1}{T}$
of the data's size). Moreover, since the samples collectively constitute a
random~$\frac{1}{T}$ fraction of the dataset, they lead to approximately even
splits across the board, even though some splits may occasionally be uneven.

Figure~\ref{fig:dynamic_infix_store}{\nobreakdash-}B
illustrates an example of this procedure, which yields the left and right Infix
Stores depicted in
Figures~\ref{fig:dynamic_infix_store}{\nobreakdash-}C
and
\ref{fig:dynamic_infix_store}{\nobreakdash-}D.
The bit strings on the left and right of each Infix Store represent its
predecessor and successor, respectively.

\textbf{(4)~Variable-Length Infixes.}
Splitting an Infix Store causes the predecessor and successor of one of the
newly created Infix Stores to have a longer common prefix. For example, the
longest common prefix in
Figure~\ref{fig:dynamic_infix_store}{\nobreakdash-}D
(shown in bold) is one bit longer than that of
Figures~\ref{fig:dynamic_infix_store}{\nobreakdash-}A
and
\ref{fig:dynamic_infix_store}{\nobreakdash-}B.

This longer common prefix takes over the higher-order bit(s) of the portion of
each key previously taken as an infix, thereby shortening the original infix.
To ensure each of these shorter infixes has an unambiguous quotient of the same
length as before, Diva repurposes the higher-order bits of the infix's old
remainder to become the lower-order bits of its quotient. Doing so splits each
old run containing remainders with different higher-order bits into shorter
runs.
Figures~\ref{fig:dynamic_infix_store}{\nobreakdash-}B
and
\ref{fig:dynamic_infix_store}{\nobreakdash-}D
present an example whereby the run previously in Slots~0101 and 0110 is split
into two runs stored in Slots~000 and 001.

Despite old remainders becoming shorter, Diva still stores the infixes of new
insertions with full-length remainders. This means that we truncate newer keys
to a lesser extent as their region within the key space becomes denser.
However, doing so necessitates storing variable-length remainders within an
Infix Store. Following our work on InfiniFilter~\cite{InfiniFilter}, Diva does
this by padding each remainder (on the right) to its original slot length with
unary codes consisting of~0s and a delimiting~1. For example, the two-bit
remainder in Slot~100 of
Figure~\ref{fig:dynamic_infix_store}-A
is stored with the two-bit padding~\underline{10}, whereas the full-length
remainder in Slot~101 is stored with a padding of~\underline{1}. Due to the
delimiting~1 bit, each slot must be one bit wider, adding a memory overhead
of~1 BPK. 

Each split reduces the length of the remainders within one Infix Store by one
bit on average. Since each Infix Store is split after roughly~$T$ insertions
(doubling its expected size), one can expect half of its remainders to be of
full length, a quarter (from one split ago) to be shorter by one bit, an eighth
(from two splits ago) to be shorter by two bits, etc. Taking a weighted average
of the remainders' FPRs based on their lengths and proportion, we obtain an FPR
that increases logarithmically with data size, in line with our past work on
expandable filers. One can maintain a constant FPR by elongating new
remainders, similarly to InfiniFilter's Widening Regime or Aleph Filter's
Predictive
Regime~\cite{InfiniFilter,AlephFilter,pagh2013approximatesetknowingsize}.

After many splits, some infixes' remainders may run out of bits, preventing
Diva from computing an unambiguous quotient for them. Diva duplicates such
infixes to account for all possible missing quotient bits and prevent false
negatives, similarly to our work on Aleph Filter~\cite{AlephFilter}. For
example, it duplicates the rightmost remainder in
Figure~\ref{fig:dynamic_infix_store}{\nobreakdash-}B
to create the highlighted remainders in
Figure~\ref{fig:dynamic_infix_store}{\nobreakdash-}D
in two different runs. Duplicating infixes complicates deletions, as there may
be multiple duplicates to clean up. Our past work offers several methods for
tackling this complexity~\cite{AlephFilter}.

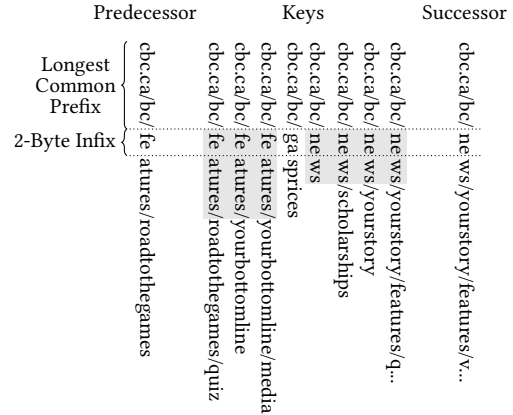
\begin{figure}
    \centering
    \pgfdeclarelayer{background layer}
    \pgfsetlayers{background layer,main}
    \vspace{-6mm}
    \begin{tikzpicture}
        \def\samplesep{20pt}
        \def\keysep{4pt}
        \def\keylabelsep{4pt}

        \node[inner sep=2pt,align=center,rotate=-90] (predecessor) at (0,0) {\small cbc.ca/bc/ fe\hspace{1.5pt} atures/roadtothegames}; 
        \node[inner sep=2pt,above=\keylabelsep of predecessor.west] (predecessor_label) {\small Predecessor};

        \begin{scope}[name prefix=key_]
            \node[inner sep=2pt,align=center,rotate=-90,above right=0pt and \samplesep of predecessor.west] (1) {\small cbc.ca/bc/ fe\hspace{1.5pt} atures/roadtothegames/quiz};
            \node[inner sep=2pt,align=center,rotate=-90,above right=0pt and \keysep of 1.west] (2) {\small cbc.ca/bc/ fe\hspace{1.5pt} atures/yourbottomline};
            \node[inner sep=2pt,align=center,rotate=-90,above right=0pt and \keysep of 2.west] (3) {\small cbc.ca/bc/ fe\hspace{1.5pt} atures/yourbottomline/media};
            \node[inner sep=2pt,align=center,rotate=-90,above right=0pt and \keysep of 3.west] (4) {\small cbc.ca/bc/ ga sprices};
            \node[inner sep=2pt,align=center,rotate=-90,above right=0pt and \keysep of 4.west] (5) {\small cbc.ca/bc/ ne ws};
            \node[inner sep=2pt,align=center,rotate=-90,above right=0pt and \keysep of 5.west] (6) {\small cbc.ca/bc/ ne ws/scholarships};
            \node[inner sep=2pt,align=center,rotate=-90,above right=0pt and \keysep of 6.west] (7) {\small cbc.ca/bc/ ne ws/yourstory};
            \node[inner sep=2pt,align=center,rotate=-90,above right=0pt and \keysep of 7.west] (8) {\small cbc.ca/bc/ ne ws/yourstory/features/q...}; 
        \end{scope}
        \node[inner sep=2pt] (keys_label) at ($(key_4.west)+(\keysep,\keylabelsep+4pt)$) {\small Keys};

        \node[inner sep=2pt,align=center,rotate=-90,above right=0pt and \samplesep of key_8.west] (successor) {\small cbc.ca/bc/ ne ws/yourstory/features/v...}; 
        \node[inner sep=2pt,above=\keylabelsep of successor.west] (successor_label) {\small Successor};

        \begin{pgfonlayer}{background layer}
            \def\sharedprefixdiffy{-35pt}
            \def\truncatedsuffixdiffy{-45pt}
            \def\truncatedsuffixAdiffy{-68.5pt}
            \def\truncatedsuffixBdiffy{-55.5pt}
            \draw[decorate,decoration={brace,mirror,raise=1pt,amplitude=2pt}] ($(predecessor.south west)+(0,-2pt)$) -- ($(predecessor.south west)+(0,\sharedprefixdiffy)$)
                node[pos=0.5,left=4pt,inner sep=1pt,align=center] (longest_common_prefix_label) {\small Longest \\[-3pt] \small Common \\[-3pt] \small Prefix};
            \draw[decorate,decoration={brace,mirror,raise=1pt,amplitude=2pt}] ($(predecessor.south west)+(0,\sharedprefixdiffy)$) -- ($(predecessor.south west)+(0,\truncatedsuffixdiffy)$)
                node[pos=0.5,left=4pt,inner sep=1pt,align=right] (infix_label) {\small 2-Byte Infix};
            \fill[fill=gray!20] ($(key_1.south west)+(2pt,\sharedprefixdiffy)$) rectangle ($(key_3.north west)+(-2pt,\truncatedsuffixAdiffy)$);
            \fill[fill=gray!20] ($(key_5.south west)+(0.5pt,\sharedprefixdiffy)$) rectangle ($(key_8.north west)+(-2pt,\truncatedsuffixBdiffy)$);
            \draw[densely dotted] ($(predecessor.south west)+(0,\sharedprefixdiffy)$) -- ($(successor.north west)+(0,\sharedprefixdiffy)$);
            \draw[densely dotted] ($(predecessor.south west)+(0,\truncatedsuffixdiffy)$) -- ($(successor.north west)+(0,\truncatedsuffixdiffy)$);
        \end{pgfonlayer}

        \node[inner sep=0pt,anchor=west] (centering_dummy) at (5.25,3.5) { };
    \end{tikzpicture}
    \vspace{-2mm}
    \caption{The length of the longest common prefix of adjacent keys within a
    dataset coming from a jagged distribution is jumpy, leading to the creation
    of identical infixes.}
    \label{fig:jagged_distribution}
\end{figure}

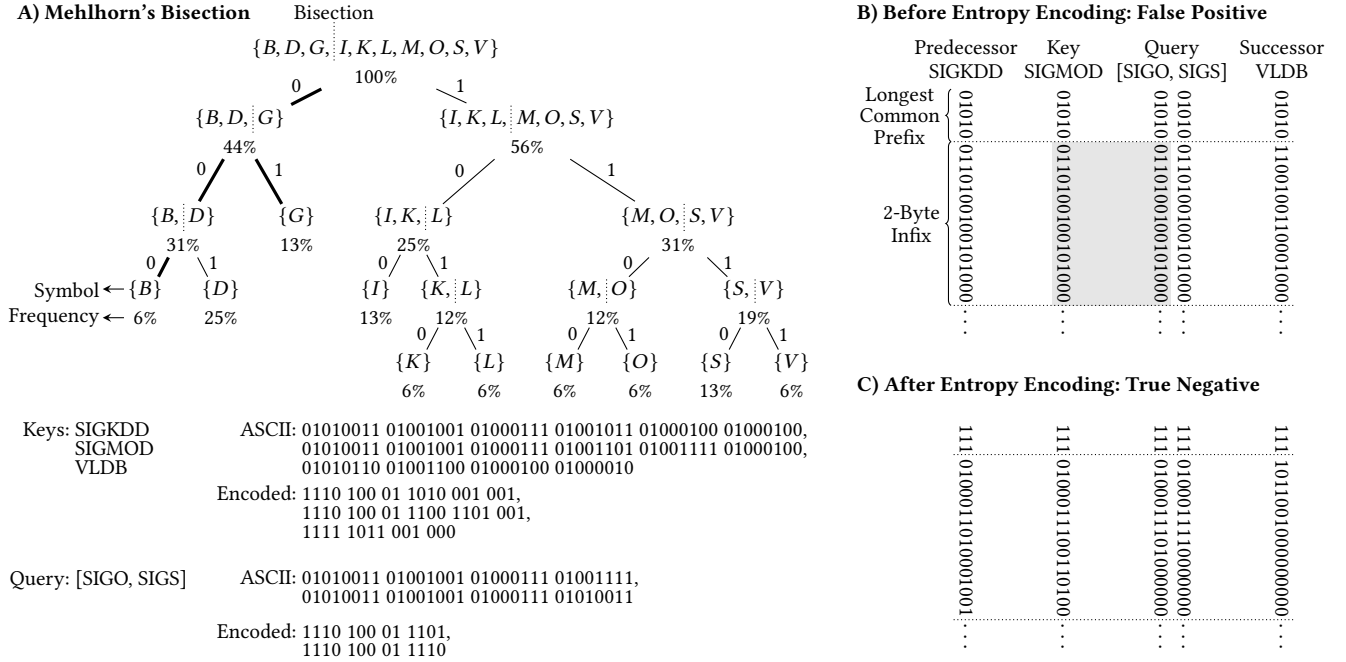
\begin{figure*}
    \noindent
    \begin{minipage}[t]{0.6\textwidth}
        \centering
        \pgfdeclarelayer{background layer}
        \pgfsetlayers{background layer,main}
        \begin{tikzpicture}[level distance=1.3cm,
                            level 1/.style={sibling distance=4cm},
                            level 2/.style={sibling distance=1.5cm},
                            level 3/.style={sibling distance=1.0cm},
                            level 4/.style={sibling distance=1.0cm},
                            edge label/.style={inner sep=2pt, font=\small},
                            normal edge/.style={thin},
                            thick edge/.style={very thick}]
            \def\corpusx{-1.30}
            \def\corpusy{-4.75}
            \def\queryx{-1.1674}
            \def\queryy{-6.5}
            \def\figlabelx{-2.25}
            \def\figlabely{1.0}

            \node[inner sep=2pt,anchor=south,align=center] (bdgiklmosv) at (2.5,0) {\small $\{B,D,G, \; I,K,L,M,O,S,V\}$ \\ \small $100\%$}
                child[sibling distance = 3.6cm] {
                    node[inner sep=2pt,anchor=south,align=center] (BDG) {\small $\{B, D, \; G\}$ \\ \small $44\%$}
                    child {
                        node[inner sep=2pt,align=center] (BD) {\small $\{B, \; D\}$ \\ \small $31\%$}
                        child {
                            node[inner sep=1pt,anchor=south,align=center] (B) {\small $\{B\}$ \\ \small $6\%$}
                            edge from parent[thick edge] node[pos=0.56,above left=-3pt and 0pt,edge label] {0}
                        }
                        child {
                            node[inner sep=1pt,anchor=south,align=center] (D) {\small $\{D\}$ \\ \small $25\%$}
                            edge from parent[normal edge] node[pos=0.5,above right=-3pt and 0pt,edge label] {1}
                        }
                        edge from parent[thick edge] node[pos=0.5,above left=-1pt,edge label] {0}
                    }
                    child {
                        node[inner sep=1pt,align=center] (G) {\small $\{G\}$ \\ \small $13\%$}
                        edge from parent[thick edge] node[pos=0.5,above right=-1pt,edge label] {1}
                    }
                    edge from parent[thick edge] node[pos=0.5,above left=-1pt,edge label] {0}
                }
                child[edge from parent path={(\tikzparentnode.south) -- (\tikzchildnode.north)}] {
                    node[inner sep=2pt,anchor=south,align=center] (IKLMOSV) {\small $\{I,K,L, \; M,O,S,V\}$ \\ \small $56\%$}
                    child[sibling distance = 3cm] {
                        node[inner sep=1pt,align=center] (IKL) {\small $\{I,K, \; L\}$ \\ \small $25\%$}
                        child {
                            node[inner sep=1pt,anchor=south,align=center] (I) {\small $\{I\}$ \\ \small $13\%$}
                            edge from parent[edge from parent path={
                                (\tikzparentnode) --node[pos=0.56,above left=-3pt and 0pt,edge label] {0} (\tikzchildnode)
                            }] 
                        }
                        child {
                            node[inner sep=1pt,anchor=south,align=center] (KL) {\small $\{K, \; L\}$ \\ \small $12\%$}
                            child {
                                node[inner sep=1pt,anchor=south,align=center] (K) {\small $\{K\}$ \\ \small $6\%$}
                                edge from parent[edge from parent path={
                                    (\tikzparentnode) --node[pos=0.5,above left=-3pt and 0pt,edge label] {0} (\tikzchildnode)
                                }] 
                            }
                            child {
                                node[inner sep=1pt,anchor=south,align=center] (L) {\small $\{L\}$ \\ \small $6\%$}
                                edge from parent[edge from parent path={
                                    (\tikzparentnode) --node[pos=0.5,above right=-3pt and 0pt,edge label] {1} (\tikzchildnode)
                                }] 
                            }
                            edge from parent[edge from parent path={
                                (\tikzparentnode) --node[pos=0.56,above right=-3pt and 0pt,edge label] {1} (\tikzchildnode)
                            }] 
                        }
                        edge from parent[normal edge, edge from parent path={
                            (\tikzparentnode) --node[pos=0.5,above left=-1pt,edge label] {0} (\tikzchildnode)
                        }] 
                    }
                    child[sibling distance = 4cm] {
                        node[inner sep=2pt,align=center] (MOSV) {\small $\{M,O, \; S,V\}$ \\ \small $31\%$}
                        child[sibling distance = 2cm] {
                            node[inner sep=1pt,anchor=south,align=center] (MO) {\small $\{M, \; O\}$ \\ \small $12\%$}
                            child {
                                node[inner sep=1pt,anchor=south,align=center] (M) {\small $\{M\}$ \\ \small $6\%$}
                                edge from parent[edge from parent path={
                                    (\tikzparentnode) --node[pos=0.5,above left=-3pt and 0pt,edge label] {0} (\tikzchildnode)
                                }] 
                            }
                            child {
                                node[inner sep=1pt,anchor=south,align=center] (O) {\small $\{O\}$ \\ \small $6\%$}
                                edge from parent[edge from parent path={
                                    (\tikzparentnode) --node[pos=0.5,above right=-3pt and 0pt,edge label] {1} (\tikzchildnode)
                                }] 
                            }
                            edge from parent[edge from parent path={
                                (\tikzparentnode) --node[pos=0.5,above left=-3pt and 0pt,edge label] {0} (\tikzchildnode)
                            }] 
                        }
                        child[sibling distance = 2cm] {
                            node[inner sep=1pt,anchor=south,align=center] (SV) {\small $\{S, \; V\}$ \\ \small $19\%$}
                            child {
                                node[inner sep=1pt,anchor=south,align=center] (S) {\small $\{S\}$ \\ \small $13\%$}
                                edge from parent[edge from parent path={
                                    (\tikzparentnode) --node[pos=0.5,above left=-3pt and 0pt,edge label] {0} (\tikzchildnode)
                                }] 
                            }
                            child {
                                node[inner sep=1pt,anchor=south,align=center] (V) {\small $\{V\}$ \\ \small $6\%$}
                                edge from parent[edge from parent path={
                                    (\tikzparentnode) --node[pos=0.5,above right=-3pt and 0pt,edge label] {1} (\tikzchildnode)
                                }] 
                            }
                            edge from parent[edge from parent path={
                                (\tikzparentnode) --node[pos=0.5,above right=-3pt and 0pt,edge label] {1} (\tikzchildnode)
                            }] 
                        }
                        edge from parent[normal edge, edge from parent path={
                            (\tikzparentnode) --node[pos=0.5,above right=-1pt,edge label] {1} (\tikzchildnode)
                        }] 
                    }
                    edge from parent[normal edge, edge from parent path={
                        (\tikzparentnode) --node[pos=0.5,above right=-1pt,edge label] {1} (\tikzchildnode)
                    }] 
                };
            \node[inner sep=1pt,above left=-10.5pt and 8pt of B] (symbol_label) {\small Symbol};
            \node[inner sep=1pt,below left=-8.5pt and 8pt of B] (freq_label) {\small Frequency};
            \draw[stealth-] ([yshift=1pt]symbol_label.east) -- ([xshift=8pt,yshift=1pt]symbol_label.east);
            \draw[stealth-] (freq_label.east) -- ([xshift=8pt]freq_label.east);

            \node[inner sep=1pt,above right=4pt and -81.01pt of bdgiklmosv] (bisection_label) {\small Bisection};
            \draw[densely dotted] (bisection_label.south) -- ([xshift=0,yshift=-15pt]bisection_label.south);
            \draw[densely dotted] ([xshift=5pt,yshift=-1pt]BDG.north) -- ([xshift=5pt,yshift=10pt]BDG.south);
            \draw[densely dotted] ([xshift=0.25pt,yshift=-1pt]BD.north) -- ([xshift=0.25pt,yshift=9pt]BD.south);
            \draw[densely dotted] ([xshift=-6.0pt,yshift=-1pt]IKLMOSV.north) -- ([xshift=-6.0pt,yshift=10pt]IKLMOSV.south);
            \draw[densely dotted] ([xshift=4.5pt,yshift=-1pt]IKL.north) -- ([xshift=4.5pt,yshift=9pt]IKL.south);
            \draw[densely dotted] ([xshift=1.5pt,yshift=-1pt]KL.north) -- ([xshift=1.5pt,yshift=9pt]KL.south);
            \draw[densely dotted] ([xshift=2.25pt,yshift=-1pt]MOSV.north) -- ([xshift=2.25pt,yshift=9pt]MOSV.south);
            \draw[densely dotted] ([xshift=1.5pt,yshift=-1pt]MO.north) -- ([xshift=1.5pt,yshift=9pt]MO.south);
            \draw[densely dotted] ([xshift=0.5pt,yshift=-1pt]SV.north) -- ([xshift=0.5pt,yshift=9pt]SV.south);

            \node[inner sep=0pt,align=left] (corpus_text) at (\corpusx,\corpusy) {\small Keys: SIGKDD \\[-4pt] \small \hspace*{0.617cm} SIGMOD \\[-4pt] \small \hspace*{0.617cm} VLDB};
            \node[inner sep=0pt,right=32pt of corpus_text.north east,anchor=north west,align=left] (corpus_ascii) {\small ASCII: 01010011 01001001 01000111 01001011 01000100 01000100, \\[-4pt] \small \hspace*{\widthof{ASCII:}} 01010011 01001001 01000111 01001101 01001111 01000100, \\[-4pt] \small \hspace*{0.7470cm} 01010110 01001100 01000100 01000010};
            \node[inner sep=0pt,below right=24pt and -8.67pt of corpus_ascii.west,anchor=west,align=left] (corpus_encoded) {\small Encoded: 1110 100 01 1010 001 001, \\[-4pt] \small \hspace*{\widthof{Encoded:}} 1110 100 01 1100 1101 001, \\[-4pt] \small \hspace*{1.0512cm} 1111 1011 001 000};

            \node[inner sep=0pt] (query_text) at (\queryx,\queryy) {\small Query: [SIGO, SIGS]};
            \node[inner sep=0pt,right=19.5pt of query_text.north east,anchor=north west,align=left] (query_ascii) {\small ASCII: 01010011 01001001 01000111 01001111, \\[-4pt] \small \hspace*{\widthof{ASCII:}} 01010011 01001001 01000111 01010011};
            \node[inner sep=0pt,below right=20pt and -8.67pt of query_ascii.west,anchor=west,align=left] (query_encoded) {\small Encoded: 1110 100 01 1101, \\[-4pt] \small \hspace*{\widthof{Encoded:}} 1110 100 01 1110};

            \node[inner sep=0pt,anchor=west] (fig_label) at (\figlabelx,\figlabely) {\small \textbf{A) Mehlhorn's Bisection}};
        \end{tikzpicture}
    \end{minipage}
    \hspace{12pt}
    \begin{minipage}[t]{0.35\textwidth}
        \vspace*{-8.635cm}
        \pgfdeclarelayer{background layer}
        \pgfsetlayers{background layer,main}
        \begin{tikzpicture}
            \def\samplesep{32pt}
            \def\keysep{4pt}
            \def\keylabelsep{0pt}
            \def\figlabelx{-1.45}
            \def\figlabely{2.63}

            \node[inner sep=2pt,align=center,rotate=-90] (predecessor) at (0,0) {\small 01010 0110100100101000 $\cdots$};
            \node[inner sep=2pt,above=\keylabelsep of predecessor.west,align=center] (predecessor_label) {\small Predecessor \\[-2pt] \small SIGKDD};

            \node[inner sep=2pt,align=center,rotate=-90,above right=0pt and \samplesep of predecessor.west] (key) {\small 01010 0110100100101000 $\cdots$};
            \node[inner sep=2pt,align=center] (key_label) at ($(key.west)+(0,\keylabelsep+9.4pt)$) {\small Key \\[-2pt] \small SIGMOD};
            \node[inner sep=2pt,align=center,rotate=-90,above right=0pt and 8*\keysep of key.west] (query_l) {\small 01010 0110100100101000 $\cdots$};
            \node[inner sep=2pt,align=center,rotate=-90,above right=0pt and \keysep of query_l.west] (query_r) {\small 01010 0110100100101000 $\cdots$};
            \node[inner sep=2pt,align=center] (query_label) at ($(query_l.west)+(\keysep,\keylabelsep+8.7pt)$) {\small Query \\[-2pt] \small [SIGO, SIGS]};

            \node[inner sep=2pt,align=center,rotate=-90,above right=0pt and \samplesep of query_r.west] (successor) {\small 01010 1100100110001000 $\cdots$};
            \node[inner sep=2pt,above=\keylabelsep of successor.west,align=center] (successor_label) {\small Successor \\[-2pt] \small VLDB};

            \begin{pgfonlayer}{background layer}
                \def\sharedprefixdiffy{-21.75pt}
                \def\truncatedsuffixdiffy{-83.4pt}
                \draw[decorate,decoration={brace,mirror,raise=1pt,amplitude=2pt}] ($(predecessor.south west)+(0,-2pt)$) -- ($(predecessor.south west)+(0,\sharedprefixdiffy)$)
                    node[pos=0.5,left=4pt,inner sep=1pt,align=center] (longest_common_prefix_label) {\small Longest \\[-3pt] \small Common \\[-3pt] \small Prefix};
                \draw[decorate,decoration={brace,mirror,raise=1pt,amplitude=2pt}] ($(predecessor.south west)+(0,\sharedprefixdiffy)$) -- ($(predecessor.south west)+(0,\truncatedsuffixdiffy)$)
                    node[pos=0.5,left=4pt,inner sep=1pt,align=center] (infix_label) {\small 2-Byte \\[-2pt] \small Infix};
                \fill[fill=gray!20] ($(key.south west)+(0.5pt,\sharedprefixdiffy)$) rectangle ($(query_l.north west)+(-1pt,\truncatedsuffixdiffy)$);
                \draw[densely dotted] ($(predecessor.south west)+(0,\sharedprefixdiffy)$) -- ($(successor.north west)+(0,\sharedprefixdiffy)$);
                \draw[densely dotted] ($(predecessor.south west)+(0,\truncatedsuffixdiffy)$) -- ($(successor.north west)+(0,\truncatedsuffixdiffy)$);
            \end{pgfonlayer}
            \node[inner sep=0pt,anchor=west] (fig_label) at (\figlabelx,\figlabely) {\small \textbf{B) Before Entropy Encoding: False Positive}};
        \end{tikzpicture}
        \\[12pt]
        \pgfdeclarelayer{background layer}
        \pgfsetlayers{background layer,main}
        \begin{tikzpicture}
            \def\samplesep{32pt}
            \def\keysep{4pt}
            \def\keylabelsep{4pt}
            \def\figlabelx{-1.45}
            \def\figlabely{2.0}

            \node[inner sep=2pt,align=center,rotate=-90] (predecessor) at (0,0) {\small 111 0100011010001001 $\cdots$};

            \node[inner sep=2pt,align=center,rotate=-90,above right=0pt and \samplesep of predecessor.west] (key) {\small 111 0100011100110100 $\cdots$};
            \node[inner sep=2pt,align=center,rotate=-90,above right=0pt and 8*\keysep of key.west] (query_l) {\small 111 0100011101000000 $\cdots$};
            \node[inner sep=2pt,align=center,rotate=-90,above right=0pt and \keysep of query_l.west] (query_r) {\small 111 0100011110000000 $\cdots$};

            \node[inner sep=2pt,align=center,rotate=-90,above right=0pt and \samplesep of query_r.west] (successor) {\small 111 1011001000000000 $\cdots$};

            \begin{pgfonlayer}{background layer}
                \def\sharedprefixdiffy{-13.75pt}
                \def\truncatedsuffixdiffy{-75.9pt}
                \draw[densely dotted] ($(predecessor.south west)+(0,\sharedprefixdiffy)$) -- ($(successor.north west)+(0,\sharedprefixdiffy)$);
                \draw[densely dotted] ($(predecessor.south west)+(0,\truncatedsuffixdiffy)$) -- ($(successor.north west)+(0,\truncatedsuffixdiffy)$);
            \end{pgfonlayer}
            \node[inner sep=0pt,anchor=west] (fig_label) at (\figlabelx,\figlabely) {\small \textbf{C) After Entropy Encoding: True Negative}};
        \end{tikzpicture}
    \end{minipage}
    \caption{Diva++ replaces each byte within a key by the root-to-leaf
    sequence of bits in its Mehlhorn tree (Part A). The symbols are shown in
    lexicographic order from left to right. Applying the entropy encoding
    distinguishes the infixes and turns false positives (Part B) into true
    negatives (Part C).}
    \label{fig:mehlhorn_example}
\end{figure*}

\textbf{Deletions.}
Diva deletes a key by first searching for it in its \ST. If found, it removes
the sample and merges the neighboring Infix Stores. \new{If not
found, it searches the Infix Store the key falls into for infixes whose
remaining bits match the key. Since the matching infixes may be of different
lengths, we must carefully choose which to remove to maintain the filter's
semantics. To illustrate, we consider the case where there are two matches. One
cannot remove the shorter infix without risking introducing false negatives.
The reason is that the shorter infix may correspond to a key different from the
deleted key that also disagrees with the longer infix. Removing the shorter
infix would leave the corresponding key without a representative, causing a
query targeting that key to result in a false negative. Diva avoids such issues
by always removing the match with the \emph{longest matching infix}, similarly
to other expandable filters~\cite{InfiniFilter,AlephFilter,Zeno}.}

\new{Note that, as in the case of all other dynamic
filters~\cite{CuckooFilter,GQF,InfiniFilter,AlephFilter,Zeno,Memento}, users
must ensure that deletions only target keys previously inserted into Diva. This
is because a filter has no way of verifying a deletion's validity.}

\textbf{Queries.}
Point and range queries are processed similarly to the static case, with the
only difference being that the missing bits from shorter infixes are treated as
wildcards. 

\section{Diva++}\label{sec:diva++}
As stated in Section~\ref{sec:diva_static}, Diva provides a
semi-robust FPR guarantee when the dataset follows a well-behaved (i.e.,
smooth) distribution. Yet, ``jagged'' (i.e., non-smooth) data distributions are
prevalent in many settings, including natural language~\cite{enwiki},
emails~\cite{Enron}, and URLs~\cite{MemeTracker}. In such datasets, the length
of the longest common prefix of adjacent keys is jumpy.
Figure~\ref{fig:jagged_distribution}
presents example keys from a URL dataset~\cite{MemeTracker} and the process of
deriving a 2{\nobreakdash-}byte infix for each before their insertion into an
Infix Store. For the remainder of this section, each character is encoded as a
single byte in ASCII. As shown, the longest common prefix of the Infix Store's
predecessor and successor is 10 bytes long, yet many adjacent keys between them
have much longer common prefixes (highlighted). For instance, the first three
keys share a 19-byte prefix, while the last four keys share a 14-byte prefix.
Consequently, after suffix truncation, the infixes of these two groups of keys
(delineated by the two dotted lines) become indistinguishable. Any negative
query falling within such a group of keys sharing the same infix will match
that infix, resulting in a false positive. This phenomenon can substantially
increase the FPR, especially when the query endpoints follow a distribution
similar to that of the dataset. 

To address this challenge, we introduce Diva++, an extension of Diva that
provides the same FPR guarantees for a broader class of jagged data
distributions (characterized towards the end of this section). Crucially,
Diva++ does so without the added memory overhead that simple modifications to
Diva (such as storing more \ST samples or longer infixes) would entail. It
achieves this through two complementary techniques that distinguish adjacent
infixes by removing redundancies among them: Order-Preserving Entropy Encoding
(Section~\ref{sec:order_preserving_entropy_encoding})
and Bit-Level Full Key Differentiation
(Sections~\ref{sec:static_bit-level_full_key_diff}
and
\ref{sec:dynamic_bit-level_full_key_diff}).
We discuss how these techniques affect queries in
Section~\ref{sec:diva++_queries}. We also evaluate these
techniques on jagged data distributions in
Experiment~\hyperlink{experiment:fpr_string}{3} of
Section~\ref{sec:evaluation_static}.

\subsection{Order-Preserving Entropy Encoding}\label{sec:order_preserving_entropy_encoding}
To reduce the length and jumpiness of adjacent common prefixes, we apply
entropy encoding to the input keys as a preprocessing step before inserting
them into the filter. Entropy encoding replaces long, repeating byte sequences
with shorter codes, thereby shortening common prefixes and allowing more
distinguishing information from each key to fit in its respective infix. This
increases the likelihood of infixes being distinguishable, even under jagged
data distributions. A central challenge, however, is ensuring that these
benefits are not offset by decoding overheads during query processing. To
address this, Diva++ employs an \emph{order-preserving} variant of entropy
coding, which ensures encoded keys retain the same lexicographical ordering as
the original keys. This enables answering a query by operating directly over
the encoded representations of the keys, thus eliminating the need for
decoding.

Specifically, we build upon Mehlhorn's bisection~\cite{MehlhornBisection}, an
entropy encoding algorithm that produces \emph{near-optimal order-preserving
prefix codes}. The prefix-code property ensures that no code is the prefix of
another, allowing encoded byte sequences to be unambiguously decoded without
requiring delimiters. The near-optimality property refers to achieving a
compression ratio close to the theoretical limit while preserving the
lexicographic ordering of the original keys. In particular, the average code
length produced by Mehlhorn's bisection is at most one bit longer than the
average code length of Huffman coding~\cite{Huffman} in the worst case, which
is itself at most one bit longer than the zero-order entropy of the
data~\cite{MehlhornBisection}.

\textbf{Construction.}
We apply Mehlhorn's bisection by decomposing each key into its constituent
bytes and treating each byte as a \emph{symbol} to be encoded. Based on the
occurrence frequency of each symbol across all keys, Mehlhorn's bisection
recurses top-down to construct an encoding tree. It initializes the tree's root
with the full set of symbols sorted lexicographically. It partitions this set
into two segments of lexicographically smaller and larger symbols such that the
total symbol frequency in each segment is as balanced as possible. It then
creates nodes for the left and right segments and connects the root to them
with edges labeled with~0 and 1 to reflect the original symbol ordering. Each
segment is recursed on until each symbol is assigned to a dedicated leaf node.
The code for a symbol is the string of bits written on the path from the root
to its corresponding leaf node.

Figure~\ref{fig:mehlhorn_example}{\nobreakdash-}A
illustrates an example of encoding the symbols in the three keys ``SIGKDD,''
``SIGMOD,'' and ``VLDB''. The leaf nodes above these keys store their
characters in lexicographical order from left to right. Here, the most balanced
partitioning of the root node's set~$\{B,D,G,I,K,L,M,O,S,V\}$ is
a~44{\nobreakdash-}56 split into~$\{B,D,G\}$ and $\{I,K,L,M,O,S,V\}$. Thus,
Mehlhorn's bisection splits the symbols into these two sets and recurses. For
the left child's set~$\{B,D,G\}$, the most balanced partitioning is
a~31{\nobreakdash-}13 split into~$\{B,D\}$ and $\{G\}$, assigning~$G$ to its
own leaf node. The recursion concludes by splitting the set~$\{B,D\}$ into its
two leaf nodes $\{B\}$ and $\{D\}$. The same process is applied to the right
child of the root node. To give examples of the resulting encodings, the
letter~$B$ is encoded as~$(000)_2$ and~$G$ as~$(01)_2$, preserving the original
lexicographic ordering while significantly shortening the
old~8{\nobreakdash-}bit ASCII encoding. As shown, the produced codes are~3.3
bits long on average, which is close to not only the average Huffman code
length of~3.125 bits but also to the underlying zero-order entropy of~$\approx
2.17$ bits per symbol. We show both the original ASCII encoding and the
  Mehlhorn encoding of the keys below the encoding tree in
  Figure~\ref{fig:mehlhorn_example}{\nobreakdash-}A.

To show how this entropy encoding makes false positives less likely during
query processing, consider a scenario where ``SIGKDD'' and ``VLDB'' are the
predecessor and successor samples of the key ``SIGMOD,'' which is stored as an
infix within an Infix Store. Also consider the empty range query [``SIGO,''
``SIGS''], with the ASCII and encoded representations shown at the bottom of
Figure~\ref{fig:mehlhorn_example}{\nobreakdash-}A.
The highlighted bits in
Figure~\ref{fig:mehlhorn_example}{\nobreakdash-}B
illustrate how, when using 2-byte infixes along with the original ASCII
encoding of the keys, the left query endpoint matches the infix of the key
``SIGMOD'' and results in a false positive. Applying the entropy encoding of
Figure~\ref{fig:mehlhorn_example}{\nobreakdash-}A
turns
Figure~\ref{fig:mehlhorn_example}{\nobreakdash-}B
into
Figure~\ref{fig:mehlhorn_example}{\nobreakdash-}C,
where the range query endpoints no longer match the infix of ``SIGMOD,''
yielding a true negative.

\textbf{Higher-Order Patterns and Codes.}
Real-world data, particularly data commonly used as database keys or
attributes, often exhibits higher-order patterns wherein certain symbols appear
more frequently after specific substrings. For example, in English text, the
letter following the substring ``th'' is very likely to be ``e'' and highly
unlikely to be ``x.'' In such cases, the compression ratio can be improved by
assigning a shorter code to a symbol based on its preceding substring. 

This phenomenon can be exploited through \emph{higher-order entropy encoding},
a class of encoding techniques that utilize the ``\emph{conditional
frequency}'' of the symbols. In general, a~$k${\nobreakdash-}th order encoding
computes the frequency at which each symbol occurs following each
possible~$k$-symbol substring. The simplest way of utilizing these frequencies
in the context of Mehlhorn's bisection is to construct a separate encoding tree
for every possible~$k$-symbol substring using the corresponding conditional
frequencies.\footnote{More generally, this approach can be applied to derive a
higher-order encoding scheme from any zero-order scheme~\cite{kthOrderCodes}.}
Once these trees are constructed, each symbol in the data is encoded using the
tree corresponding to its~$k$ preceding symbols. A larger value of~$k$ allows
for exploiting more patterns, yielding a higher compression ratio in exchange
for higher construction and memory overheads due to having more encoding trees.

In Diva++, we employ second-order Mehlhorn codes ($k=2$). We found that this
choice significantly improves the compression ratio while ensuring that the
encoding trees incur a negligible memory overhead relative to the filter's
size. Table~\ref{tab:entropy_encoding_stats}
demonstrates this by comparing the tree sizes and compression ratios of
second-order Mehlhorn codes with zero-, one-, and third-order codes on the
URLs~dataset~\cite{MemeTracker}.

\begin{table}
    \centering
    \caption{Higher-order Mehlhorn codes yield higher compression ratios in
    exchange for a larger memory footprint for the encoding trees. The
    compression ratios reported are from the URLs dataset~\cite{MemeTracker}
    used in our evaluation.}
    \def\arraystretch{1.1}
    \setlength{\tabcolsep}{6.0pt}
    \begin{tabular}{ccccc}
        \toprule

        & \textbf{\small $k=0$}
        & \textbf{\small $k=1$}
        & \textbf{\small $k=2$} 
        & \textbf{\small $k=3$} \\

        \midrule

        \small \textbf{Compression Ratio} & \small 1.33 & \small 1.87 & \small 2.25 & \small 2.56 \\[4pt]

        \small \makecell{\textbf{Size Ratio} \\ \textbf{(Encoding Trees vs. Filter)}} & \small $<0.0001\%$ & \small 0.002\% & \small 0.1\% & \small 8.5\% \\[8pt]
        
        \small \makecell{\textbf{Construction Time Ratio} \\ \textbf{(Encoding Trees vs. Filter)}} & \small $<0.0001\%$ & \small 0.001\% & \small 0.04\% & \small 2.9\% \\
        \bottomrule
    \end{tabular}
    \label{tab:entropy_encoding_stats}
\end{table}

\textbf{Computing and Maintaining the Encoding.}
When the dataset is static, the encoding trees are constructed only once for
the entire filter during a preprocessing pass over the data, incurring a
negligible computational overhead ($\approx 0.04\%$ of total filter
construction time, as shown in the bottom row of
Table~\ref{tab:entropy_encoding_stats}).
When the dataset is dynamic, however, the initial dataset may be too small to
accurately capture the eventual symbol frequencies. Nevertheless, it is often
feasible to construct the encoding trees using domain knowledge. For example,
an application operating on English text only requires a tree encoding the
statistical structure of English, which can be computed in advance.
Alternatively, the encoding trees can be periodically reconstructed as the
dataset evolves. To support this, our implementation includes a library for
applying Mehlhorn's bisection that can be invoked at any time to recompute the
encoding trees.

\textbf{Other Encoding Schemes.}
Aside from Mehlhorn's bisection, several other order-preserving encoding
schemes have been proposed in the
literature~\cite{AntoshenkovOrderPreservingCompression,CarstenOrderPreservingCompression,HOPE,OptimumBSTs,HuTucker,GarsiaWachs,ElementsOfInformationTheory}.
Among these, Hu-Tucker coding~\cite{OptimumBSTs,HuTucker,GarsiaWachs} is an
order-preserving variant of Huffman coding~\cite{Huffman} that achieves the
optimal compression ratio while maintaining the lexicographic ordering of the
symbols. Although Hu-Tucker coding provides a slightly better compression ratio
than Mehlhorn's bisection, we opt for the latter for two reasons:
(1)~Mehlhorn's tree is faster to construct than a Hu-Tucker tree ($O(\sigma)$
vs. $O(\sigma \log \sigma)$, where $\sigma$ is the number of unique symbols),
and (2)~Mehlhorn's bisection is simpler to implement.

Other encoding methods aim to attain a better compression ratio by capturing
even higher-order patterns in the
data~\cite{AntoshenkovOrderPreservingCompression,CarstenOrderPreservingCompression,HOPE}.
They do this by heuristically maintaining a small dictionary that maps symbol
strings of various lengths to short codes rather than storing an exhaustive set
of encoding trees. This helps most with very long keys, where very high-order
patterns are more likely to occur. However, we found that these schemes yield a
worse FPR when used with Diva++. The reason is two-fold: (1)~keys in range
filters are usually short to medium length, leaving fewer high-order patterns
for these methods to exploit, and (2)~these methods prioritize encoding and
decoding speed by producing byte-aligned codes rather than bit-aligned ones,
thus compromising the compression ratio. In contrast, a second-order Mehlhorn
code exploits many patterns while retaining discriminative information at the
finest~granularity.

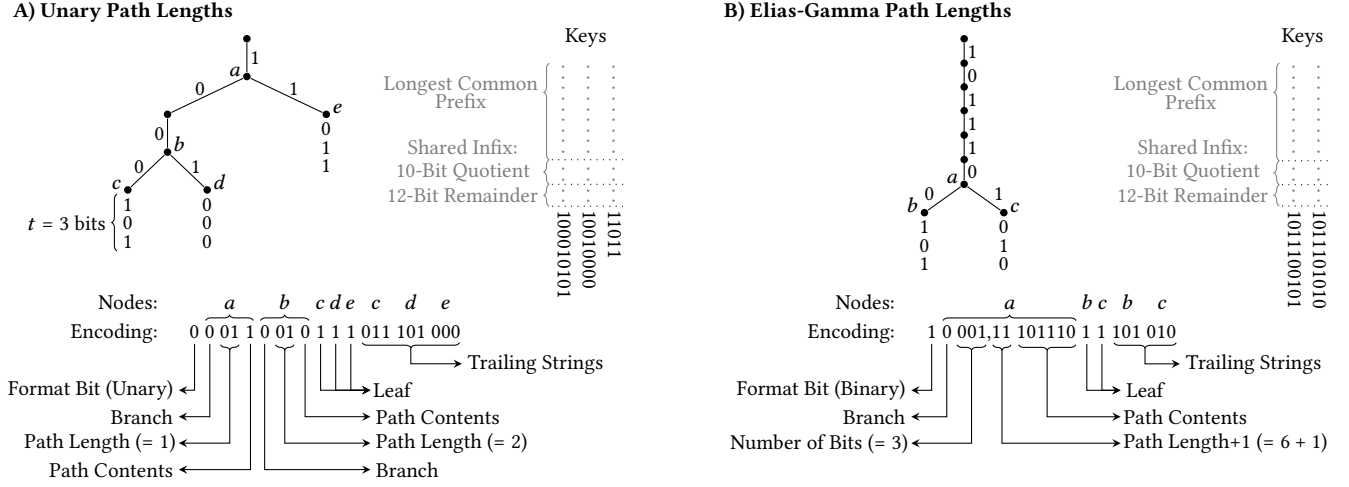
\begin{figure*}
    \centering
    \hspace*{0mm}
    \pgfdeclarelayer{background layer}
    \pgfsetlayers{background layer,main}
    \begin{tikzpicture}
        \def\encodingoffx{3.75}
        \def\encodingoffy{-3.9}
        \def\circlerad{1pt}

        \node[inner sep=0pt] (encoding) at (\encodingoffx,\encodingoffy) {\small Encoding: \hspace{8pt} 0 0 01 1 0 01 0 1 1 1 011 101 000};

        \begin{scope}[name prefix=trie_]
            \def\trieoffx{3.5}
            \def\trieoffy{0.0}
            \def\triehdiff{0.75*0.7}
            \def\trieleveldiff{0.5}

            \node[inner sep=\circlerad,fill=black,circle] (1) at (\trieoffx,\trieoffy) { };

            \node[inner sep=\circlerad,fill=black,circle] (2) at (\trieoffx,\trieoffy-1*\trieleveldiff) { };
            \draw (1) -- node[pos=0.5,inner sep=1pt,right=0pt] {\small 1} (2);
            \node[inner sep=0pt,above left=-1pt and 1pt of 2] (label_2) {\small $a$};

            \node[inner sep=\circlerad,fill=black,circle] (3) at (\trieoffx-2*\triehdiff,\trieoffy-2*\trieleveldiff) { };
            \draw (2) -- node[pos=0.3,inner sep=1pt,left=5pt] {\small 0} (3);

            \node[inner sep=\circlerad,fill=black,circle] (4) at (\trieoffx-2*\triehdiff,\trieoffy-3*\trieleveldiff) { };
            \draw (3) -- node[pos=0.5,inner sep=1pt,left=0pt] {\small 0} (4);
            \node[inner sep=0pt,above right=-1pt and 1pt of 4] (label_4) {\small $b$};

            \node[inner sep=\circlerad,fill=black,circle] (6) at (\trieoffx-3*\triehdiff,\trieoffy-4*\trieleveldiff) { };
            \draw (4) -- node[pos=0.35,inner sep=1pt,left=2pt] {\small 0} (6);
            \node[inner sep=0pt,above left=-1pt and 1pt of 6] (label_6) {\small $c$};

            \node[inner sep=\circlerad,fill=black,circle] (7) at (\trieoffx-1*\triehdiff,\trieoffy-4*\trieleveldiff) { };
            \draw (4) -- node[pos=0.35,inner sep=1pt,right=2pt] {\small 1} (7);
            \node[inner sep=0pt,above right=-1pt and 1pt of 7] (label_7) {\small $d$};

            \node[inner sep=\circlerad,fill=black,circle] (5) at (\trieoffx+2*\triehdiff,\trieoffy-2*\trieleveldiff) { };
            \draw (2) -- node[pos=0.3,inner sep=1pt,right=5pt] {\small 1} (5);
            \node[inner sep=0pt,above right=-1pt and 1pt of 5] (label_5) {\small $e$};

            \node[inner sep=1pt,below=0pt of 6,align=center] (suffix_1) {\small 1 \\[-4pt] \small 0 \\[-4pt] \small 1};
            \node[inner sep=1pt,below=0pt of 7,align=center] (suffix_2) {\small 0 \\[-4pt] \small 0 \\[-4pt] \small 0};
            \node[inner sep=1pt,below=0pt of 5,align=center] (suffix_3) {\small 0 \\[-4pt] \small 1 \\[-4pt] \small 1};

            \draw[decorate,decoration={brace,mirror,raise=1pt,amplitude=2pt}] (suffix_1.north west) -- (suffix_1.south west)
                node[pos=0.5,left=4pt,inner sep=1pt,align=center] (num_trailing_bits) {\small $t=3$ bits};

            \def\keyoffx{7.7}
            \def\cdotdist{4pt}
            \def\topoffy{44.25pt}
            \def\sharedprefixoffy{8.0pt}
            \def\sharedquotientoffy{-1.25pt}
            \def\sharedremainderoffy{-9.5pt}
            \node[inner sep=2pt,rotate=270,anchor=west] (full_suffix_1) at (\keyoffx,\trieoffy-0.5*\trieleveldiff) {\small \textcolor{gray}{$\cdots\cdots\cdots$\hspace{\cdotdist}$\cdots\cdot$} 100010101};
            \node[inner sep=2pt,rotate=270,above=1pt of full_suffix_1] (full_suffix_2) {\small \textcolor{gray}{$\cdots\cdots\cdots$\hspace{\cdotdist}$\cdots\cdot$} 10010000\hspace*{1.72pt}};
            \node[inner sep=2pt,rotate=270,above=1pt of full_suffix_2] (full_suffix_3) {\small \textcolor{gray}{$\cdots\cdots\cdots$\hspace{\cdotdist}$\cdots\cdot$} 11011\hspace*{10.9pt}};
            \node[inner sep=2pt,above left=2pt and -4pt of full_suffix_2] (full_suffix_label) {\small Keys};
            \draw[dotted] ($(full_suffix_1.south)+(0,\sharedquotientoffy)$) -- ($(full_suffix_1.south)+(28pt,\sharedquotientoffy)$);
            \draw[decorate,decoration={brace,mirror,raise=1pt,amplitude=2pt},gray] ($(full_suffix_1.south)+(0,\sharedprefixoffy)$) -- ($(full_suffix_1.south)+(0,\sharedquotientoffy)$)
                node[pos=0.5,left=7pt,inner sep=1pt,align=center,gray] (shared_infix_quotient_label) {\small 10-Bit Quotient};
            \draw[dotted] ($(full_suffix_1.south)+(0,\sharedremainderoffy)$) -- ($(full_suffix_1.south)+(28pt,\sharedremainderoffy)$);
            \draw[decorate,decoration={brace,mirror,raise=1pt,amplitude=2pt},gray] ($(full_suffix_1.south)+(0,\sharedquotientoffy)$) -- ($(full_suffix_1.south)+(0,\sharedremainderoffy)$)
                node[pos=0.5,left=4pt,inner sep=1pt,align=center,gray] (shared_infix_remainder_label) {\small 12-Bit Remainder};
            \node[inner sep=0pt,above=2pt of shared_infix_quotient_label,gray] (shared_infix_label) {\small Shared Infix:};
            \draw[dotted] ($(full_suffix_1.south)+(0,\sharedprefixoffy)$) -- ($(full_suffix_1.south)+(28pt,\sharedprefixoffy)$);
            \draw[decorate,decoration={brace,mirror,aspect=0.3,raise=1pt,amplitude=2pt},gray] ($(full_suffix_1.south)+(0,\topoffy)$) -- ($(full_suffix_1.south)+(0,\sharedprefixoffy)$)
                node[pos=0.3,left=4pt,inner sep=1pt,align=center,gray] (shared_infix_label) {\small Longest Common \\[-4pt] \small Prefix};
        \end{scope}

        \begin{scope}[name prefix=encoding_label_]
            \def\labeloffx{\encodingoffx-1.15}
            \def\labeloffy{-0.75}
            \node[inner sep=2pt,anchor=east] (format) at (\labeloffx,\encodingoffy+\labeloffy) {\small Format Bit (Unary)};
            \draw[-stealth] ($(encoding.south)+(-0.9425,0.0)$) |- (format.east);

            \node[inner sep=2pt,anchor=east] (branch) at (\labeloffx,\encodingoffy+\labeloffy-0.35) {\small Branch};
            \draw[-stealth] ($(encoding.south)+(-0.7405,0.0)$) |- (branch.east);

            \draw[decorate,decoration={brace,mirror,raise=1pt,amplitude=2pt}] ($(encoding.south)+(-0.5975,0.025)$) --($(encoding.south)+(-0.355,0.025)$) 
                node[pos=0.5,below=0.5pt,inner sep=1pt,align=center] (path_len_brace) { };
            \node[inner sep=1pt,anchor=east] (path_len) at (\labeloffx,\encodingoffy+\labeloffy-0.7) {\small Path Length ($=1$)};
            \draw[-stealth] (path_len_brace.south) |- (path_len.east);

            \node[inner sep=2pt,anchor=east] (path_contents) at (\labeloffx,\encodingoffy+\labeloffy-1.05) {\small Path Contents};
            \draw[-stealth] ($(encoding.south)+(-0.2025,0.0)$) |- (path_contents.east);

            \def\labeloffx{\encodingoffx+1.385}
            \node[inner sep=1pt,anchor=west] (leaf) at (\labeloffx,\encodingoffy+\labeloffy) {\small Leaf};
            \draw[-stealth] ($(encoding.south)+(0.725,0.0)$) |- (leaf.west);
            \draw[-stealth] ($(encoding.south)+(0.925,0.0)$) |- (leaf.west);
            \draw[-stealth] ($(encoding.south)+(1.125,0.0)$) |- (leaf.west);

            \node[inner sep=2pt,anchor=west] (path_contents_2) at (\labeloffx,\encodingoffy+\labeloffy-0.35) {\small Path Contents};
            \draw[-stealth] ($(encoding.south)+(0.52,0.0)$) |- (path_contents_2.west);

            \draw[decorate,decoration={brace,mirror,raise=1pt,amplitude=2pt}] ($(encoding.south)+(0.13,0.025)$) --($(encoding.south)+(0.3725,0.025)$) 
                node[pos=0.5,below=0.5pt,inner sep=1pt,align=center] (path_len_brace_2) { };
            \node[inner sep=2pt,anchor=west] (path_len_2) at (\labeloffx,\encodingoffy+\labeloffy-0.7) {\small Path Length ($=2$)};
            \draw[-stealth] (path_len_brace_2.south) |- (path_len_2.west);

            \node[inner sep=2pt,anchor=west] (branch_2) at (\labeloffx,\encodingoffy+\labeloffy-1.05) {\small Branch};
            \draw[-stealth] ($(encoding.south)+(-0.015,0.0)$) |- (branch_2.west);

            \draw[decorate,decoration={brace,mirror,raise=1pt,amplitude=2pt}] ($(encoding.south)+(1.275,1.5pt)$) --($(encoding.south)+(2.5625,1.5pt)$) 
                node[pos=0.5,below=1.5pt,inner sep=0pt,align=center] (trailing_bits_brace) { };
            \node[inner sep=2pt,anchor=east] (trailing_bits) at (\encodingoffx+4.5425,\encodingoffy+\labeloffy+0.35) {\small Trailing Strings};
            \draw[-stealth] (trailing_bits_brace) |- (trailing_bits.west);

            \node[inner sep=0pt,above left=4.9pt and -33.15pt of encoding] (corresponding_nodes) {\small Nodes:};
            \draw[decorate,decoration={brace,raise=1pt,amplitude=2pt}] ($(encoding.north)+(-0.795,0.0)$) --($(encoding.north)+(-0.165,0.0)$) 
                node[pos=0.5,above=4pt,inner sep=1pt,align=center] (node_1) {\small $a$};
            \draw[decorate,decoration={brace,raise=1pt,amplitude=2pt}] ($(encoding.north)+(-0.07,0.0)$) --($(encoding.north)+(0.57,0.0)$) 
                node[pos=0.5,above=4pt,inner sep=1pt,align=center] (node_2) {\small $b$};
            \node[inner sep=5pt,anchor=south] (node_3) at ($(encoding.north)+(0.720,0.0)$) {\small $c$};
            \node[inner sep=5pt,anchor=south] (node_4) at ($(encoding.north)+(0.925,0.0)$) {\small $d$};
            \node[inner sep=5pt,anchor=south] (node_5) at ($(encoding.north)+(1.120,0.0)$) {\small $e$};

            \node[inner sep=5pt,anchor=south] (trailing_1) at ($(encoding.north)+(1.46,0.0)$) {\small $c$};
            \node[inner sep=5pt,anchor=south] (trailing_2) at ($(encoding.north)+(1.92,0.0)$) {\small $d$};
            \node[inner sep=5pt,anchor=south] (trailing_3) at ($(encoding.north)+(2.38,0.0)$) {\small $e$};
        \end{scope}

        \node[inner sep=0pt,anchor=west] (fig_label) at (0.4,0.35) {\small \textbf{A) Unary Path Lengths}};
    \end{tikzpicture}
    \hspace{10.0mm}
    \pgfdeclarelayer{background layer}
    \pgfsetlayers{background layer,main}
    \begin{tikzpicture}
        \def\slotw{0.7}
        \def\sloth{0.35}
        \def\slotcnt{9}
        \def\storeoffx{1.5}
        \def\storeoffy{-3.0}
        \def\encodingoffx{3.9}
        \def\encodingoffy{-3.9}
        \def\circlerad{1pt}

        \node[inner sep=0pt] (encoding) at (\encodingoffx,\encodingoffy) {\small Encoding: \hspace{8pt} 1 0 001,11 101110 1 1 101 010};

        \begin{scope}[name prefix=trie_]
            \def\trieoffx{\storeoffx+\slotw*\slotcnt/2-1.6*\slotw}
            \def\trieoffy{0.0}
            \def\triehdiff{0.75*\slotw}
            \def\trieleveldiff{0.5}

            \node[inner sep=\circlerad,fill=black,circle] (1) at (\trieoffx,\trieoffy) { };

            \node[inner sep=\circlerad,fill=black,circle,below=0.21 of 1] (dummy_1) { };
            \node[inner sep=\circlerad,fill=black,circle,below=0.20 of dummy_1] (dummy_2) { };
            \node[inner sep=\circlerad,fill=black,circle,below=0.20 of dummy_2] (dummy_3) { };
            \node[inner sep=\circlerad,fill=black,circle,below=0.21 of dummy_3] (dummy_4) { };
            \node[inner sep=\circlerad,fill=black,circle,below=0.21 of dummy_4] (dummy_5) { };

            \node[inner sep=\circlerad,fill=black,circle] (2) at (\trieoffx,\trieoffy-3.85*\trieleveldiff) { };
            \draw (1) -- node[pos=0.5,inner sep=1pt,right=0pt,align=center] {\small 1 \\[-2pt] \small 0 \\[-2pt] \small 1 \\[-2pt] \small 1 \\[-2pt] \small 1 \\[-2pt] \small 0} (2);
            \node[inner sep=0pt,above left=-1pt and 1pt of 2] (label_2) {\small $a$};

            \node[inner sep=\circlerad,fill=black,circle] (3) at (\trieoffx-\triehdiff,\trieoffy-4.6*\trieleveldiff) { };
            \draw (2) -- node[pos=0.3,inner sep=1pt,left=5pt] {\small 0} (3);
            \node[inner sep=0pt,above left=-1pt and 1pt of 3] (label_3) {\small $b$};

            \node[inner sep=\circlerad,fill=black,circle] (4) at (\trieoffx+\triehdiff,\trieoffy-4.6*\trieleveldiff) { };
            \draw (2) -- node[pos=0.3,inner sep=1pt,right=5pt] {\small 1} (4);
            \node[inner sep=0pt,above right=-1pt and 1pt of 4] (label_4) {\small $c$};

            \node[inner sep=1pt,below=0pt of 3,align=center] (suffix_1) {\small 1 \\[-4pt] \small 0 \\[-4pt] \small 1};
            \node[inner sep=1pt,below=0pt of 4,align=center] (suffix_2) {\small 0 \\[-4pt] \small 1 \\[-4pt] \small 0};

            \def\trieoffx{\storeoffx+\slotw*\slotcnt/2}
            \def\cdotdist{4pt}
            \def\topoffy{46.25pt}
            \def\sharedprefixoffy{9.5pt}
            \def\sharedquotientoffy{0.5pt}
            \def\sharedremainderoffy{-7.5pt}
            \node[inner sep=2pt,rotate=270,anchor=west] (full_suffix_1) at (\trieoffx+3.25,\trieoffy-0.5*\trieleveldiff) {\small \textcolor{gray}{$\cdots\cdots\cdots$\hspace{\cdotdist}$\cdots\cdot$} 1011100101};
            \node[inner sep=2pt,rotate=270,above=0pt of full_suffix_1] (full_suffix_2) {\small \textcolor{gray}{$\cdots\cdots\cdots$\hspace{\cdotdist}$\cdots\cdot$} 1011101010};
            \node[inner sep=2pt,above left=2pt and 1pt of full_suffix_2] (full_suffix_label) {\small Keys};
            \draw[dotted] ($(full_suffix_1.south)+(0,\sharedquotientoffy)$) -- ($(full_suffix_1.south)+(18pt,\sharedquotientoffy)$);
            \draw[decorate,decoration={brace,mirror,raise=1pt,amplitude=2pt},gray] ($(full_suffix_1.south)+(0,\sharedprefixoffy)$) -- ($(full_suffix_1.south)+(0,\sharedquotientoffy)$)
                node[pos=0.5,left=7pt,inner sep=1pt,align=center,gray] (shared_infix_quotient_label) {\small 10-Bit Quotient};
            \draw[dotted] ($(full_suffix_1.south)+(0,\sharedremainderoffy)$) -- ($(full_suffix_1.south)+(18pt,\sharedremainderoffy)$);
            \draw[decorate,decoration={brace,mirror,raise=1pt,amplitude=2pt},gray] ($(full_suffix_1.south)+(0,\sharedquotientoffy)$) -- ($(full_suffix_1.south)+(0,\sharedremainderoffy)$)
                node[pos=0.5,left=4pt,inner sep=1pt,align=center,gray] (shared_infix_remainder_label) {\small 12-Bit Remainder};
            \node[inner sep=0pt,above=2pt of shared_infix_quotient_label,gray] (shared_infix_label) {\small Shared Infix:};
            \draw[dotted] ($(full_suffix_1.south)+(0,\sharedprefixoffy)$) -- ($(full_suffix_1.south)+(18pt,\sharedprefixoffy)$);
            \draw[decorate,decoration={brace,mirror,aspect=0.3,,raise=1pt,amplitude=2pt},gray] ($(full_suffix_1.south)+(0,\topoffy)$) -- ($(full_suffix_1.south)+(0,\sharedprefixoffy)$)
                node[pos=0.3,left=4pt,inner sep=1pt,align=center,gray] (shared_prefix_label) {\small Longest Common \\[-4pt] \small Prefix};
        \end{scope}

        \begin{scope}[name prefix=encoding_label_]
            \def\labeloffx{\encodingoffx-1.07}
            \def\labeloffy{-0.75}
            \node[inner sep=2pt,anchor=east] (format) at (\labeloffx,\encodingoffy+\labeloffy) {\small Format Bit (Binary)};
            \draw[-stealth] ($(encoding.south)+(-0.80,0.0)$) |- (format.east);

            \node[inner sep=2pt,anchor=east] (branch) at (\labeloffx,\encodingoffy+\labeloffy-0.35) {\small Branch};
            \draw[-stealth] ($(encoding.south)+(-0.60,0.0)$) |- (branch.east);

            \node[inner sep=1pt,anchor=east] (num_digits) at (\labeloffx,\encodingoffy+\labeloffy-0.7) {\small Number of Bits ($= 3$)};
            \draw[decorate,decoration={brace,mirror,raise=1pt,amplitude=2pt}] ($(encoding.south)+(-0.465,0.025)$) --($(encoding.south)+(-0.080,0.025)$) 
                node[pos=0.5,below=0.5pt,inner sep=1pt,align=center] (num_digits_brace) { };
            \draw[-stealth] (num_digits_brace.south) |- (num_digits.east);

            \def\labeloffx{\encodingoffx+1.70}
            \node[inner sep=1pt,anchor=west] (digit) at (\labeloffx,\encodingoffy+\labeloffy-0.7) {\small Path Length$+ 1$ ($= 6 + 1$)};
            \draw[decorate,decoration={brace,mirror,raise=1pt,amplitude=2pt}] ($(encoding.south)+(0.01,0.025)$) --($(encoding.south)+(0.25,0.025)$) 
                node[pos=0.5,below=0.5pt,inner sep=1pt,align=center] (digit_brace) { };
            \draw[-stealth] (digit_brace.south) |- (digit.west);

            \node[inner sep=1pt,anchor=west] (path_contents) at (\labeloffx,\encodingoffy+\labeloffy-0.35) {\small Path Contents};
            \draw[decorate,decoration={brace,mirror,raise=1pt,amplitude=2pt}] ($(encoding.south)+(0.350,0.025)$) --($(encoding.south)+(1.105,0.025)$) 
                node[pos=0.5,below=0.5pt,inner sep=1pt,align=center] (path_contents_brace) { };
            \draw[-stealth] (path_contents_brace.south) |- (path_contents.west);

            \node[inner sep=2pt,anchor=west] (leaf) at (\labeloffx,\encodingoffy+\labeloffy) {\small Leaf};
            \draw[-stealth] ($(encoding.south)+(1.2525,0.0)$) |- (leaf.west);
            \draw[-stealth] ($(encoding.south)+(1.4525,0.0)$) |- (leaf.west);

            \draw[decorate,decoration={brace,mirror,raise=1pt,amplitude=2pt}] ($(encoding.south)+(1.610,1.5pt)$) --($(encoding.south)+(2.435,1.5pt)$) 
                node[pos=0.5,below=1.5pt,inner sep=0pt,align=center] (trailing_bits_brace) { };
            \node[inner sep=2pt,anchor=east] (trailing_bits) at (\encodingoffx+4.4325,\encodingoffy+\labeloffy+0.35) {\small Trailing Strings};
            \draw[-stealth] (trailing_bits_brace) |- (trailing_bits.west);

            \node[inner sep=0pt,above left=4.9pt and -33.15pt of encoding] (corresponding_nodes) {\small Nodes:};
            \draw[decorate,decoration={brace,raise=1pt,amplitude=2pt}] ($(encoding.north)+(-0.65,0.0)$) --($(encoding.north)+(1.11,0.0)$) 
                node[pos=0.5,above=4pt,inner sep=1pt,align=center] (node_1) {\small $a$};
            \node[inner sep=5pt,anchor=south] (node_2) at ($(encoding.north)+(1.257,0.0)$) {\small $b$};
            \node[inner sep=5pt,anchor=south] (node_3) at ($(encoding.north)+(1.462,0.0)$) {\small $c$};

            \node[inner sep=5pt,anchor=south] (trailing_1) at ($(encoding.north)+(1.790,0.0)$) {\small $b$};
            \node[inner sep=5pt,anchor=south] (trailing_2) at ($(encoding.north)+(2.250,0.0)$) {\small $c$};
        \end{scope}

        \node[inner sep=0pt,anchor=west] (fig_label) at (0.35,0.35) {\small \textbf{B) Elias-Gamma Path Lengths}};
        \node[inner sep=0pt] (dummy) at (1.0,-5.85) { };
    \end{tikzpicture}
    \caption{Diva++ differentiates keys with identical infixes using binary
    D-Tries. It uses either a unary format (Part A) or a binary format (Part B)
    to save memory when encoding the lengths of a D-Trie's paths. The
    annotations above each encoding indicate which branch or leaf each part of
    the encoding represents.}
    \label{fig:d_tries}
\end{figure*}

\textbf{Improved FPR Guarantees.}
Second-order entropy encoding allows Diva++ to extend Diva's FPR bounds to all
data distributions that do not have third- or higher-order patterns in their
keys. Intuitively, when the data has no patterns spanning more than~$k$
symbols,\footnote{Formally, this corresponds to each symbol in the data being
independent of all symbols before it when conditioned on its~$k$ preceding
symbols.} $k$-th order entropy encoding ensures that each bit in an encoded key
is equally likely to be~0 or 1, yielding uniformly distributed
infixes~\cite{ShannonInformationTheory}. This uniformity bounds the FPR, since
a negative query's endpoints become less likely to match an infix.
Hypothetically, if one employs infinite-order entropy encoding,\footnote{One
can use a large language model to implement infinite-order entropy encoding.
This is because a large language model correlates each symbol with all
preceding symbols to produce a frequency distribution over the next symbol. The
trade-off would be slower construction and query times due to the model's
inference overhead.} Diva++'s FPR guarantee would extend to any data
distribution.

Second-order entropy encoding significantly improves the distinguishability of
infixes even when the keys have third- or higher-order patterns. Nevertheless,
the higher-order patterns within these datasets may still pollute the infixes
with uninformative bits. For example,
Experiment~\hyperlink{experiment:fpr_string}{3} of
Section~\ref{sec:evaluation_static} uses
the original Diva design with second-order entropy encoding to process the
English \textsc{EnWiki} dataset under a memory budget of~16~BPK. Had the
dataset contained only second- or lower-order patterns, the infixes would be
fully differentiated, and the FPR expression of
Section~\ref{sec:diva_static_sampling_keys_and_deriving_infixes}
would imply an expected FPR of~$\approx 0.01\%$. Yet, because the dataset
contains higher-order patterns, the entropy encoding reduces the proportion of
identical infixes only from~90\% to 13\%, resulting in an FPR of 17\%. In
general, when a proportion~$p$ of the infixes are identical, the FPR is at
least as large as~$p+\epsilon$, where $\epsilon$ is the expected FPR of the
filter.\footnote{When a range query's endpoints follow the same distribution as
the dataset, the range falls in-between two identical infixes with a
probability of~$p$. Such queries always return a false positive, while other
queries do so with an FPR of~$\epsilon$ or higher.} This stems from the Diva's
original design storing each infix from a group of identical infixes within a
dedicated slot, consuming space that could otherwise be used to store more bits
from each key and improve the FPR. We mitigate this shortcoming with our next
technique: bit-level full key
differentiation~(Sections~\ref{sec:static_bit-level_full_key_diff}~and~\ref{sec:dynamic_bit-level_full_key_diff}).

\subsection{Static Bit-Level Full Key Differentiation}\label{sec:static_bit-level_full_key_diff}
In the previous section, we discussed how Diva++ reduces the proportion of
identical infixes through order-preserving entropy encoding. We now show how to
fully disambiguate infixes that remain identical after entropy encoding. Diva++
achieves this by retaining a single infix from each group of identical infixes
and freeing the space taken up by the rest. It repurposes this freed space to
differentiate the group's keys, storing more bits from each key within a
compact binary trie called a \emph{\DT}. In this section, we describe the
logical structure of a \DT. We also introduce a novel \underline{bi}nary
\underline{tr}ie encoding (BITR) to encode each \DT. BITR represents both the
trie's topology and its edge labels in at most two bits per
node{\textemdash}improving on prior schemes, which can encode only the topology
in this space~\cite{LOUDS,SuRF}. This section discusses the static setting,
while Section~\ref{sec:dynamic_bit-level_full_key_diff}
covers the dynamic setting.

\textbf{\DTs.}
A \DT is a binary trie in which each edge is labeled with a single bit. It
contains a unique leaf node for each key, and the string of bits on each
root-to-leaf path represents the bits appearing after the infix of a key. That
is, each bit string is long enough to fully differentiate the corresponding
key. The top part of
Figure~\ref{fig:d_tries}{\nobreakdash-}A
presents a logical view of a \DT for the three keys on the right. As described
shortly, we encode each \DT succinctly using BITR, allowing it to take up less
space than the group of infixes it replaces.\footnote{The size of this encoding
further decreases when one uses higher-order entropy encoding or when the data
contains fewer higher-order patterns.} For example, when processing the URLs
dataset~\cite{MemeTracker}, the trie's encoding consumes~$\approx 10$ BPK on
average{\textemdash}smaller than the typical memory budget of a range filter,
i.e., 12--20 BPK.

\textbf{Trailing Strings.}
Diva++ stores more bits from each key to the right of the corresponding \DT,
using the remaining freed space from the original group of infixes. In
particular, it stores at each key's leaf node a~$t${\nobreakdash-}bit substring
called the key's \emph{trailing string}. If a key is too short to have
a~$t${\nobreakdash-}bit trailing string, we treat the bits beyond the key's end
as~0s to create the trailing string, thereby retaining the keys' original
lexicographic ordering. We tune the parameter~$t$ for each \DT during
construction so that the total size of the \DT and the trailing strings does
not exceed the size of the replaced group of infixes. Storing trailing strings
makes the \DT's keys more likely to be distinguishable from a query's endpoints
while maintaining the same space usage as the original Diva design.
Figure~\ref{fig:d_tries}{\nobreakdash-}A
illustrates storing a trailing string of length~$t=3$ bits for each key.

\textbf{Encoding \DTs with BITR.}
We design a succinct and self-delimiting encoding for \underline{bi}nary
\underline{tr}ies, dubbed \emph{BITR}. BITR represents the trie's branch and
leaf nodes (henceforth referred to as branches and leaves) in depth-first
order, followed by the trailing strings in the same order. For every branch and
leaf within the \DT, BITR describes the path connecting it to its nearest
ancestor branch. Specifically, it stores the length of this path and the bit
labels along it. To optimize for memory, BITR encodes the path lengths in one
of two formats, as discussed shortly. It indicates the format used with a
\emph{format bit} at the beginning of the encoding. It then represents the
\DT's branches and leaves as follows:

\emph{For each branch,}
BITR stores three fields: (1)~a~0 bit indicating that it is a branch, (2)~the
length of the path leading to it from its nearest ancestor branch (or the root
if no such ancestor branch exists), and (3)~the string of bits on the edges of
the aforementioned path. We save~$\approx 1$ BPK when encoding the bit strings
by noting that, for every branch except the first (i.e., the highest), the bit
on the first edge of its respective path need not be stored. The omitted bit is
inferable from two observations: (a)~the \DT's binary fanout ensures that the
labels on the left and right edges of any branch are~0~and~1, respectively, and
(b)~the depth-first ordering of the encoding indicates whether the omitted
bit's edge is the left or right edge of the ancestor branch, thus determining
its label.

\emph{For each leaf,} there is no need to encode a path. The reason is that,
unlike branches, the path leading to the leaf from the nearest ancestor branch
is always a single edge whose bit label is already determined by the preceding
encoding. This path is a single edge because the \DT keeps just enough bits
from each key (aside from the trailing strings) to differentiate it. Hence,
BITR stores only a single~1 bit to indicate that the current node is a leaf. 

\textbf{Path Length Formats.}
BITR allows representing path lengths of branches in (A)~\emph{unary} and
(B)~\emph{Elias-Gamma} formats. Unary is more efficient for \DTs with short
average path lengths (less than 3 edges), whereas Elias-Gamma is better for
\DTs with long path lengths. BITR chooses whichever format minimizes the total
encoding length for each \DT.

\textbf{(A)~Unary Lengths.}
Diva++'s order-preserving entropy encoding (described in
Section~\ref{sec:order_preserving_entropy_encoding})
often captures most of the data's patterns, yielding shallow \DTs with short
paths connecting branches (1-3 edges long in the URLs
dataset~\cite{MemeTracker}). The unary format targets these tries by
representing path lengths using unary codes. Specifically, it encodes lengths
of~$1,2,3,\dots$ as $1,01,001,\dots$. As a special case, the first (i.e.,
highest) branch in the \DT can have a path length of~0, since it may be the
root. Hence, for the first branch, we encode the lengths~$0,1,2,\dots$ as
$1,01,001,\dots$. Unary codes are space-optimal when the average path is~1-3
edges long~\cite{InfiniFilter,Huffman}.

The \DT at the top of
Figure~\ref{fig:d_tries}{\nobreakdash-}A
is encoded with unary path lengths at the bottom of the figure. The encoding
begins with a format bit of~0. The next~0 bit indicates that the first node
encoded in depth-first order, $a$, is a branch. The subsequent~01 bits indicate
that the path from the root to~$a$ is one edge long, and the next~1 is the bit
label of this edge. BITR then indicates the branch~$b$ with a~0 followed by a
path length of~01, meaning that the path from the nearest ancestor branch ($a$)
to~$b$ is two edges long. BITR stores the bit label~0 on the lower edge of this
path while omitting the bit label on its upper edge, as the upper edge is~$a$'s
left edge and therefore is known to be labeled~0. The encoding finishes with
three~1s to indicate the leaves~$c$, $d$, and $e$. The three trailing strings
of the keys are stored in the depth-first order of the \DT.

\textbf{(B)~Elias-Gamma Lengths.}
When the dataset contains many higher-order bit patterns not captured by the
entropy encoding, adjacent keys may have more common bits after their infixes.
This causes the corresponding \DT to contain longer paths. For example, in the
URLs dataset~\cite{MemeTracker}, approximately~14\% of the \DTs have paths
longer than~50 edges. Unary codes would waste significant space for such \DTs,
as they consume as much space as the path lengths. The Elias-Gamma format
avoids this wastage using Elias-Gamma coding~\cite{EliasGamma}, which is a
self-delimiting binary encoding. An Elias-Gamma code represents a non-negative
integer~$v$ by first storing the number of bits required to represent~$v$,
i.e., $\lceil\log_2(v+1)\rceil$, in unary (with 1, 01, 001, $\dots$, indicating
1, 2, 3, $\dots$ bits). Then, the code stores the bits in $v$'s binary
representation. The binary representation, in turn, has two cases: (1)~If~$v$'s
plain binary representation consists of more than one bit, the code stores all
its bits except for the most significant one, as it is always~1. (2)~If~$v$ has
a single bit, the code stores that single bit, as it \mbox{can be either~0 or
1}.

For a path of length~$l$, BITR stores the value~$v=(l+1)$ instead of~$l$ using
the Elias-Gamma code. This is to reserve the value~$v=0$ for representing keys
ending in internal nodes. We elaborate on this case when discussing dynamic
\DTs in Section~\ref{sec:dynamic_bit-level_full_key_diff}.

Figure~\ref{fig:d_tries}{\nobreakdash-}B
encodes the \DT at its top using the Elias-Gamma format. This format is
signaled by the leading~1 bit of the encoding shown at the bottom of the
figure. The following bit of the encoding is~0, meaning that the first encoded
node in depth-first order, $a$, is a branch. The next bits form an Elias-Gamma
code. The unary part of the code is~001, which indicates a 3-bit value.
Therefore, the next two bits, i.e., $(11)_2$, are the two lower-order bits of
the encoded value. Since the most-significant bit of this value is~1, the code
represents~$v=(111)_2=7$, so the path from the root to~$a$ has
length~$l=v-1=6$. The next six bits represent the bit labels along this path.
The remainder of the encoding follows the same structure as BITR's unary
format.

\textbf{Encoding Trailing Strings.}
After encoding the \DT's nodes, BITR represents the trailing string of each
key. For each trailing string, encoding its bits as edge labels within the \DT
would incur an unnecessary metadata overhead to represent the edges' topology.
Specifically, in the unary and Elias-Gamma formats, this overhead would be
linear and logarithmic in the trailing string's length, respectively. BITR
avoids this overhead entirely by packing the trailing strings and storing them
\mbox{to the right of the \DT encoding}.

\textbf{Storing and Locating \DTs within Runs.}
Diva++ stores each \DT next to its shared infix's remainder within the same
run. To preserve slot alignment, a small amount of space in the rightmost slot
storing the \DT and the trailing strings may go unused.
Figure~\ref{fig:d_trie_in_run}
illustrates this storage layout as applied to the \DT of
Figure~\ref{fig:d_tries}{\nobreakdash-}A.
Here, the run stores a~12-bit remainder for each infix. The first and last
slots contain infix remainders without \DTs. The middle three slots are where
the remainders of the three keys of the aforementioned \DT would have been
stored. These slots now contain the keys' shared infix, their \DT's encoding,
and their trailing strings. The two bits after the trailing strings are~unused.

\begin{figure}
    \centering
    \pgfdeclarelayer{background layer}
    \pgfsetlayers{background layer,main}
    \begin{tikzpicture}
        \def\slotw{1.53}
        \def\sloth{0.35}
        \def\slotcnt{5}
        \def\storeoffx{1.5}
        \def\storeoffy{-3.0}
        \def\circlerad{1pt}
        \def\trieleveldiff{0.5}

        \draw[very thick] (\storeoffx,\storeoffy) rectangle (\storeoffx+\slotcnt*\slotw,\storeoffy-\sloth);
        \node[inner sep=4pt,anchor=east] (run_label) at (\storeoffx,\storeoffy-\sloth/2) {\small Run};
        \foreach \x in {1, ..., \slotcnt} {
            \draw[dotted] (\storeoffx+\slotw*\x,\storeoffy-\sloth) -- (\storeoffx+\slotw*\x,\storeoffy);
        }

        \node at (\storeoffx+0*\slotw+0.5*\slotw,\storeoffy-0.5*\sloth) {\footnotesize 000000000101};
        \node at (\storeoffx+1*\slotw+0.5*\slotw,\storeoffy-0.5*\sloth) {\footnotesize 100000000011};
        \node at (\storeoffx+2*\slotw+0.5*\slotw,\storeoffy-0.5*\sloth) {\footnotesize 0\hspace{1.0pt}00011001011};
        \node at (\storeoffx+3*\slotw+0.41*\slotw,\storeoffy-0.5*\sloth) {\footnotesize 1011101000};
        \def\dtriestart{2.123}
        \def\dtrieend{3.8}
        \begin{pgfonlayer}{background layer}
            \fill[pattern=north east lines,pattern color=gray!50] (\storeoffx+\dtrieend*\slotw,\storeoffy-\sloth) rectangle (\storeoffx+4*\slotw,\storeoffy);
            \fill[fill=gray!20] (\storeoffx+\dtriestart*\slotw,\storeoffy-\sloth) rectangle (\storeoffx+\dtrieend*\slotw,\storeoffy);
        \end{pgfonlayer}
        \node at (\storeoffx+4*\slotw+0.5*\slotw,\storeoffy-0.5*\sloth) {\footnotesize 111111111111};

        \begin{scope}[name prefix=label_]
            \draw[decorate,decoration={brace,raise=1pt,amplitude=2pt}] (\storeoffx+1*\slotw,\storeoffy) -- (\storeoffx+2*\slotw,\storeoffy)
                node[pos=0.5,above=5.8pt,inner sep=1pt,align=center] (shared_infix) {\small Shared Infix};
            \draw[decorate,decoration={brace,raise=1pt,amplitude=2pt}] (\storeoffx+\dtriestart*\slotw,\storeoffy) -- (\storeoffx+\dtrieend*\slotw,\storeoffy)
                node[pos=0.5,above=4pt,inner sep=1pt,align=center] (d_trie) {\small D-Trie Encoding};

            \draw[decorate,decoration={brace,mirror,raise=1pt,amplitude=2pt}] (\storeoffx+1*\slotw,\storeoffy-\sloth) -- (\storeoffx+3*\slotw,\storeoffy-\sloth)
                node[pos=0.5,below=1.5pt,inner sep=0pt,align=center] (escape_sequence_connector) {};
            \node[inner sep=1pt,align=center] (escape_sequence_label) at (\storeoffx+1.02*\slotw,\storeoffy-2.1*\sloth) {\small Escape Sequence};
            \draw[-stealth] (escape_sequence_connector) |- (escape_sequence_label);
            \node[inner sep=1pt,align=center] (padding_label) at (\storeoffx+2.7*\slotw,\storeoffy-2.1*\sloth) {\small Padding};
            \draw[-stealth] (\storeoffx+2.071*\slotw,\storeoffy-0.85*\sloth) |- (padding_label);
            \draw[decorate,decoration={brace,mirror,raise=1pt,amplitude=2pt}] (\storeoffx+\dtrieend*\slotw,\storeoffy-\sloth) -- (\storeoffx+4*\slotw,\storeoffy-\sloth)
                node[pos=0.5,below=1.5pt,inner sep=0pt,align=center] (unused_space_connector) {};
            \node[inner sep=1pt,align=center] (unused_space_label) at (\storeoffx+\dtrieend*\slotw/2+2*\slotw,\storeoffy-2.1*\sloth) {\small Unused Space};
        \end{scope}
    \end{tikzpicture}
    \caption{Diva++ stores a D-Trie next to the corresponding shared infix's
    remainder within the same run.}
    \label{fig:d_trie_in_run}
\end{figure}
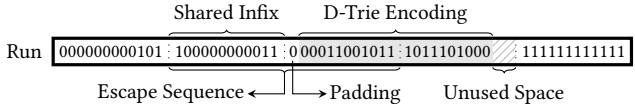

To enable unambiguous decoding, we store the infixes within each run in
increasing order, from left to right. This allows for signaling the existence
of a \DT for an infix by introducing a decrease in the run's slot values,
thereby creating an \emph{escape sequence}. We create this decrease by
0-padding the \DT's encoding on the left. Specifically, we add just enough 0s
so that the bit position corresponding to the highest-order~1 of the shared
infix contains a~0 bit. Doing so avoids potential ambiguities in the bit
position of the \DT's encoding. This process adds~1{\nobreakdash-}2 bits of
padding to each \DT on average. The reason is that the entropy encoding of the
keys makes each bit of the shared infix almost equally likely to be either~0 or
1, implying that its highest-order~1 most likely falls within
its~1{\nobreakdash-}2 most-significant bits. 

Figure~\ref{fig:d_trie_in_run}
presents an example. Here, the slot to the right of the shared infix contains a
smaller value, forming an escape sequence that indicates the existence of a
\DT. The encoding of the \DT, shown in gray, is padded with a single~0 bit
(corresponding to the infix's highest-order~1) to ensure unambiguous decoding.

\textbf{Fanout Discussion.}
Many trie-based data structures employ a high fanout to reduce the trie's depth
and lower the search cost~\cite{ART,HOT,Wormhole,SuRF}. Unlike these designs,
\DTs use a fanout of two, which is more space-efficient in our context for two
reasons: (1)~When differentiating two keys up to their first differing bit, a
higher fanout represents the keys at a coarser granularity than bits, thereby
storing more bits than necessary. (2)~A high-fanout node's edge labels often
have common prefixes that should ideally be stored only once. Known techniques
for suppressing these common prefixes (such as Golomb~\cite{Golomb} or
Elias-Fano~\cite{EliasFanoElias,EliasFanoFano} coding) waste memory by
introducing unnecessary metadata overheads. A fanout of two allows Diva++ to
store the entropy-encoded keys at the bit granularity without duplicate bits
across edge labels. As a result, Diva++ is more space-efficient than range
filters employing high-fanout tries at the same FPR level (e.g., as much
as~$1.4\times$ smaller than SuRF~\cite{SuRF}), as shown in
Experiment~\hyperlink{experiment:fpr_string}{3}
of Section~\ref{sec:evaluation_static}.

\textbf{Other Trie Encodings.}
While one can encode a \DT using any succinct, self-delimiting encoding
scheme~\cite{TreeEncoding1,TreeEncoding2,TreeEncoding3,TreeEncoding4,TreeEncoding5,TreeEncoding6,TreeEncoding7,TreeEncodingSODA,LOUDS},
we design BITR for two reasons: (1)~Existing encoding schemes are wasteful when
applied to \DTs, since they target general tries that can have nodes with an
arbitrarily high fanout. For example, the LOUDS encoding~\cite{LOUDS} used by
SuRF~\cite{SuRF} represents only the topology of a trie (without the edge
labels) by storing each node's degree in unary, incurring a memory footprint
of~2 bits per node. Completing this encoding requires storing the bit label of
each single-child node's edge, as this bit cannot be inferred from the topology
alone. Doing so incurs an additional overhead of up to~1 bit per node,
depending on the proportion of single-child nodes. Under well-behaved data
distributions common in practice (e.g., uniform), this translates to a total
memory footprint of~$\approx 2.25$ bits per node, since half of the nodes are
internal nodes and each has a single child with probability~$0.5$. Unlike
LOUDS, BITR jointly represents both the topology and the edge labels in at
most~2 bits per node by exploiting the \DTs' binary fanout (proven in
Section~\ref{sec:theoretical_analysis_and_results}).
BITR's memory footprint is optimal, which can be shown by relating the total
number of binary tries to Catalan numbers~\cite{TAOCP1}, as
in~\cite{TreeEncodingLowerBound}. (2)~For dynamic \DTs, introduced in
Section~\ref{sec:dynamic_bit-level_full_key_diff},
special cases arise wherein the bits stored for a key may end in an internal
node instead of a leaf. Standard trie encodings cannot represent such keys,
whereas BITR can do so using its Elias-Gamma format.

\begin{table}
    \centering
    \caption{Both order-preserving entropy encoding and bit-level full key
    differentiation play a significant role in Diva++'s improved FPR. Here, the
    memory budget is fixed to~16 BPK, and the percentages are the attained
    FPRs.}
    \def\arraystretch{1.1}
    \setlength{\tabcolsep}{5.5pt}
    \begin{tabular}{cccc}
        \toprule

        \small \textbf{Dataset} & \small \textsc{EnWiki}~\cite{enwiki} & \small \textsc{Emails}~\cite{Enron} & \small \textsc{URLs}~\cite{MemeTracker} \\[1pt]

        \midrule

        \small \textbf{Diva} & \small 63\% & \small 84\% & \small 88\% \\[4pt]

        \small \makecell{\textbf{Diva + Order-Preserving} \\ \textbf{Entropy Encoding}} & \small 16\% & \small 61\% & \small 58\% \\[8pt]

        \small \makecell{\textbf{Diva + Bit-Level} \\ \textbf{Full Key Differentiation}} & \small 21\% & \small 37\% & \small 40\% \\[8pt]
        
        \small \textbf{Diva++} & \small 7\% & \small 19\% & \small 15\% \\

        \bottomrule
    \end{tabular}
    \label{tab:diva++_ablation}
\end{table}

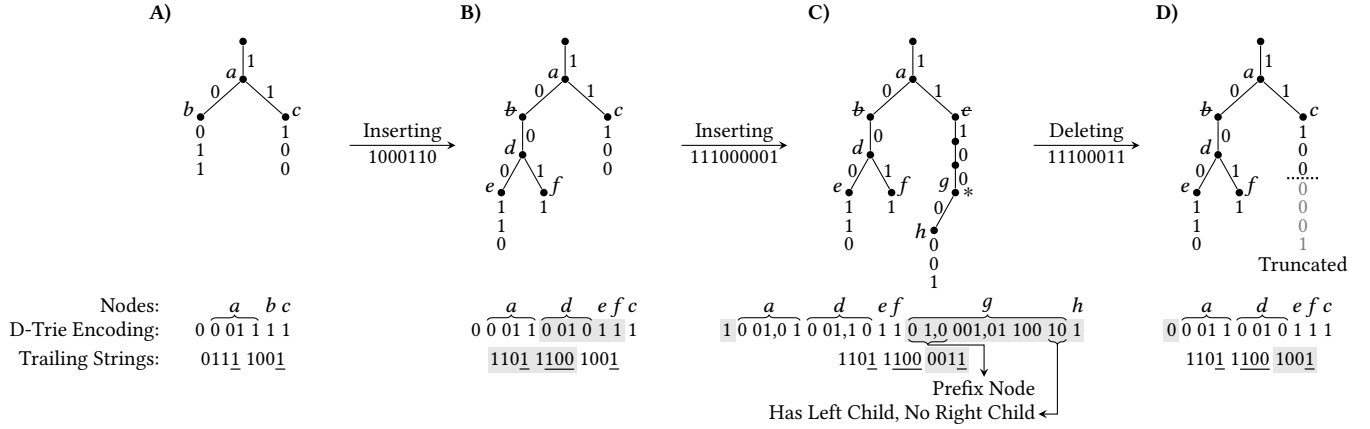
\begin{figure*}
    \centering
    \pgfdeclarelayer{background layer}
    \pgfsetlayers{background layer,main}
    \begin{tikzpicture}
        \def\circlerad{1pt}
        \def\trieoffy{0.0}
        \def\trieleveldiff{0.5}
        \def\triehdiff{0.75*0.75}
        \def\triecenteroff{0.25*\triehdiff}
        \def\colonex{0.0}
        \def\coltwox{4.12}
        \def\colthreex{8.72}
        \def\colfourx{13.32}
        \def\arrowonex{2.11}
        \def\arrowtwox{6.50}
        \def\arrowthreex{11.16}
        \def\arrowlen{1.4}
        \def\arrowy{\trieoffy-2.75*\trieleveldiff}
        \def\labeloffx{1.25}
        \def\labely{1.0*\trieleveldiff}
        \def\encodingy{\trieoffy-8.0775*\trieleveldiff}

        \def\trieoffx{\colonex}

        \begin{scope}[name prefix=t1_]
            \node[inner sep=\circlerad,fill=black,circle] (1) at (\trieoffx,\trieoffy) { };

            \node[inner sep=\circlerad,fill=black,circle] (2) at (\trieoffx,\trieoffy-1*\trieleveldiff) { };
            \draw (1) -- node[pos=0.5,inner sep=1pt,right=0pt] {\small 1} (2);
            \node[inner sep=0pt,above left=-1pt and 1pt of 2] (label_2) {\small $a$};

            \node[inner sep=\circlerad,fill=black,circle] (3) at (\trieoffx-\triehdiff,\trieoffy-2*\trieleveldiff) { };
            \draw (2) -- node[pos=0.3,inner sep=1pt,left=2pt] {\small 0} (3);
            \node[inner sep=0pt,above left=-1pt and 1pt of 3] (label_3) {\small $b$};

            \node[inner sep=\circlerad,fill=black,circle] (4) at (\trieoffx+\triehdiff,\trieoffy-2*\trieleveldiff) { };
            \draw (2) -- node[pos=0.3,inner sep=1pt,right=2pt] {\small 1} (4);
            \node[inner sep=0pt,above right=-1pt and 1pt of 4] (label_4) {\small $c$};

            \node[inner sep=1pt,below=0pt of 3,align=center] (suffix_1) {\small 0 \\[-4pt] \small 1 \\[-4pt] \small 1};
            \node[inner sep=1pt,below=0pt of 4,align=center] (suffix_2) {\small 1 \\[-4pt] \small 0 \\[-4pt] \small 0};
        \end{scope}

        \node[inner sep=0pt,align=center] (encoding_1) at (\colonex,\encodingy) {\small 0 0 01 1 1 1 \\[0pt] \small 011\underline{1} 100\underline{1}};
        \node[inner sep=0pt,above left=-18.25pt and 13pt of encoding_1,align=right] (encoding_label) {\small Nodes: \\[-2pt] \small \DT Encoding: \\[0pt] \small Trailing Strings:};

        \begin{scope}[name prefix=t1_label_]
            \draw[decorate,decoration={brace,raise=1pt,amplitude=2pt}] ($(encoding_1.north)+(-0.425,-0.015)$) --($(encoding_1.north)+(0.20,-0.015)$)
                node[pos=0.5,above=3pt,inner sep=1pt,align=center] (node_1) {\small $a$};
            \node[inner sep=5pt,anchor=south] (node_2) at ($(encoding_1.north)+(0.37,-1.4pt)$) {\small $b$};
            \node[inner sep=5pt,anchor=south] (node_3) at ($(encoding_1.north)+(0.57,-1.4pt)$) {\small $c$};
        \end{scope}

        \node[inner sep=0pt,align=left,anchor=north west] (fig_label_1) at (\colonex-\labeloffx,\labely) {\small \textbf{A)}};

        \draw[-stealth] (\arrowonex-\arrowlen/2,\arrowy) -- (\arrowonex+\arrowlen/2,\arrowy)
            node[pos=0.5,above=-8.1pt,inner sep=1pt,align=center] (transition_1) {\small Inserting \\[-1pt] \small 1000110};

        \def\trieoffx{\coltwox+\triecenteroff}

        \begin{scope}[name prefix=t2_]
            \node[inner sep=\circlerad,fill=black,circle] (1) at (\trieoffx,\trieoffy) { };

            \node[inner sep=\circlerad,fill=black,circle] (2) at (\trieoffx,\trieoffy-1*\trieleveldiff) { };
            \draw (1) -- node[pos=0.5,inner sep=1pt,right=0pt] {\small 1} (2);
            \node[inner sep=0pt,above left=-1pt and 1pt of 2] (label_2) {\small $a$};

            \node[inner sep=\circlerad,fill=black,circle] (3) at (\trieoffx-\triehdiff,\trieoffy-2*\trieleveldiff) { };
            \draw (2) -- node[pos=0.25,inner sep=1pt,left=2pt] {\small 0} (3);
            \node[inner sep=0pt,above left=-1pt and 1pt of 3] (label_3) {\small \sout{$b$}};

            \node[inner sep=\circlerad,fill=black,circle] (5) at (\trieoffx-\triehdiff,\trieoffy-3*\trieleveldiff) { };
            \draw (3) -- node[pos=0.5,inner sep=1pt,right=0pt] {\small 0} (5);
            \node[inner sep=0pt,above left=-1pt and 1pt of 5] (label_5) {\small $d$};

            \node[inner sep=\circlerad,fill=black,circle] (6) at (\trieoffx-1.5*\triehdiff,\trieoffy-4*\trieleveldiff) { };
            \draw (5) -- node[pos=0.4,inner sep=1pt,left=0.25pt] {\small 0} (6);
            \node[inner sep=0pt,above left=-1pt and 1pt of 6] (label_6) {\small $e$};

            \node[inner sep=\circlerad,fill=black,circle] (7) at (\trieoffx-0.5*\triehdiff,\trieoffy-4*\trieleveldiff) { };
            \draw (5) -- node[pos=0.4,inner sep=1pt,right=0.25pt] {\small 1} (7);
            \node[inner sep=0pt,above right=-2.75pt and 1pt of 7] (label_7) {\small $f$};

            \node[inner sep=\circlerad,fill=black,circle] (4) at (\trieoffx+\triehdiff,\trieoffy-2*\trieleveldiff) { };
            \draw (2) -- node[pos=0.25,inner sep=1pt,right=2pt] {\small 1} (4);
            \node[inner sep=0pt,above right=-1pt and 1pt of 4] (label_4) {\small $c$};

            \node[inner sep=1pt,below=0pt of 6,align=center] (suffix_1) {\small 1 \\[-4pt] \small 1 \\[-4pt] \small 0};
            \node[inner sep=1pt,below=0pt of 7,align=center] (suffix_2) {\small 1};
            \node[inner sep=1pt,below=0pt of 4,align=center] (suffix_3) {\small 1 \\[-4pt] \small 0 \\[-4pt] \small 0};
        \end{scope}

        \def\trieoffx{\coltwox}
        \node[inner sep=0pt,align=center] (encoding_2) at (\trieoffx,\encodingy) {\small 0 0 01 1 0 01 0 1 1 1 \\[0pt] \small 110\underline{1} 1\underline{100} 100\underline{1}};
        \begin{pgfonlayer}{background layer}
            \def\bith{0.15}
            \fill[fill=gray!20] (\trieoffx-0.205,\encodingy-\bith+0.215) rectangle (\trieoffx+0.93,\encodingy+\bith+0.215);
            \fill[fill=gray!20] (\trieoffx-0.875,\encodingy-\bith-0.190) rectangle (\trieoffx+0.30,\encodingy+\bith-0.155);
        \end{pgfonlayer}

        \begin{scope}[name prefix=t2_label_]
            \draw[decorate,decoration={brace,raise=1pt,amplitude=2pt}] ($(encoding_2.north)+(-0.89,-0.015)$) -- ($(encoding_2.north)+(-0.255,-0.015)$)
                node[pos=0.5,above=3pt,inner sep=1pt,align=center] (node_1) {\small $a$};
            \draw[decorate,decoration={brace,raise=1pt,amplitude=2pt}] ($(encoding_2.north)+(-0.17,-0.015)$) -- ($(encoding_2.north)+(0.49,-0.015)$)
                node[pos=0.5,above=3pt,inner sep=1pt,align=center] (node_2) {\small $d$};
            \node[inner sep=5pt,anchor=south] (node_3) at ($(encoding_2.north)+(0.635,-1.4pt)$) {\small $e$};
            \node[inner sep=5pt,anchor=south] (node_4) at ($(encoding_2.north)+(0.84,-3.225pt)$) {\small $f$};
            \node[inner sep=5pt,anchor=south] (node_5) at ($(encoding_2.north)+(1.035,-1.4pt)$) {\small $c$};
        \end{scope}

        \node[inner sep=0pt,align=left,anchor=north west] (fig_label_2) at (\coltwox-\labeloffx,\labely) {\small \textbf{B)}};

        \draw[-stealth] (\arrowtwox-\arrowlen/2,\arrowy) -- (\arrowtwox+\arrowlen/2,\arrowy)
            node[pos=0.5,above=-8.1pt,inner sep=1pt,align=center] (transition_2) {\small Inserting \\[-1pt] \small 111000001};

        \def\trieoffx{\colthreex+\triecenteroff}

        \begin{scope}[name prefix=t3_]
            \node[inner sep=\circlerad,fill=black,circle] (1) at (\trieoffx,\trieoffy) { };

            \node[inner sep=\circlerad,fill=black,circle] (2) at (\trieoffx,\trieoffy-1*\trieleveldiff) { };
            \draw (1) -- node[pos=0.5,inner sep=1pt,right=0pt] {\small 1} (2);
            \node[inner sep=0pt,above left=-1pt and 1pt of 2] (label_2) {\small $a$};

            \node[inner sep=\circlerad,fill=black,circle] (3) at (\trieoffx-\triehdiff,\trieoffy-2*\trieleveldiff) { };
            \draw (2) -- node[pos=0.25,inner sep=1pt,left=2pt] {\small 0} (3);
            \node[inner sep=0pt,above left=-1pt and 1pt of 3] (label_3) {\small \sout{$b$}};

            \node[inner sep=\circlerad,fill=black,circle] (5) at (\trieoffx-\triehdiff,\trieoffy-3*\trieleveldiff) { };
            \draw (3) -- node[pos=0.5,inner sep=1pt,right=0pt] {\small 0} (5);
            \node[inner sep=0pt,above left=-1pt and 1pt of 5] (label_5) {\small $d$};

            \node[inner sep=\circlerad,fill=black,circle] (6) at (\trieoffx-1.5*\triehdiff,\trieoffy-4*\trieleveldiff) { };
            \draw (5) -- node[pos=0.4,inner sep=1pt,left=0.25pt] {\small 0} (6);
            \node[inner sep=0pt,above left=-1pt and 1pt of 6] (label_6) {\small $e$};

            \node[inner sep=\circlerad,fill=black,circle] (7) at (\trieoffx-0.5*\triehdiff,\trieoffy-4*\trieleveldiff) { };
            \draw (5) -- node[pos=0.4,inner sep=1pt,right=0.25pt] {\small 1} (7);
            \node[inner sep=0pt,above right=-2.75pt and 1pt of 7] (label_7) {\small $f$};

            \node[inner sep=\circlerad,fill=black,circle] (4) at (\trieoffx+\triehdiff,\trieoffy-2*\trieleveldiff) { };
            \draw (2) -- node[pos=0.25,inner sep=1pt,right=2pt] {\small 1} (4);
            \node[inner sep=0pt,above right=-1pt and 1pt of 4] (label_4) {\small \sout{$c$}};

            \node[inner sep=\circlerad,fill=black,circle] (8) at (\trieoffx+\triehdiff,\trieoffy-4.0*\trieleveldiff) { };
            \draw (4) -- node[pos=0.5,inner sep=1pt,right=0pt,align=center] {\small 1 \\[-2pt] \small 0 \\[-2pt] \small 0} (8);
            \node[inner sep=0pt,right=1pt of 8] (prefix_node_indicator) {$*$};
            \node[inner sep=0pt,above left=-1pt and 1pt of 8] (label_8) {\small $g$};

            \node[inner sep=\circlerad,fill=black,circle,below=0.21 of 4] (dummy_1) { };
            \node[inner sep=\circlerad,fill=black,circle,below=0.20 of dummy_1] (dummy_2) { };

            \node[inner sep=\circlerad,fill=black,circle] (9) at (\trieoffx+\triehdiff/2,\trieoffy-5.0*\trieleveldiff) { };
            \draw (8) -- node[pos=0.375,inner sep=1pt,left=0pt] {\small 0} (9);
            \node[inner sep=0pt,left=1pt of 9] (label_9) {\small $h$};

            \node[inner sep=1pt,below=0pt of 6,align=center] (suffix_1) {\small 1 \\[-4pt] \small 1 \\[-4pt] \small 0};
            \node[inner sep=1pt,below=0pt of 7,align=center] (suffix_2) {\small 1};
            \node[inner sep=1pt,below=0pt of 9,align=center] (suffix_3) {\small 0 \\[-4pt] \small 0 \\[-4pt] \small 1};
        \end{scope}

        \def\trieoffx{\colthreex}
        \node[inner sep=0pt,align=center] (encoding_4) at (\trieoffx,\encodingy) {\small 1 0 01,0 1 0 01,1 0 1 1 0 1,0 001,01 100 10 1 \\[0pt] \small 110\underline{1} 1\underline{100} 001\underline{1}};
        \begin{pgfonlayer}{background layer}
            \def\bith{0.15}
            \fill[fill=gray!20] (\trieoffx-2.409,\encodingy-\bith+0.215) rectangle (\trieoffx-2.204,\encodingy+\bith+0.215);
            \fill[fill=gray!20] (\trieoffx+0.040,\encodingy-\bith+0.215) rectangle (\trieoffx+2.407,\encodingy+\bith+0.215);
            \fill[fill=gray!20] (\trieoffx+0.3,\encodingy-\bith-0.190) rectangle (\trieoffx+0.88,\encodingy+\bith-0.155);
        \end{pgfonlayer}

        \begin{scope}[name prefix=t3b_label_]
            \draw[decorate,decoration={brace,raise=1pt,mirror,amplitude=2pt}] (\trieoffx+0.084,\encodingy+0.135) -- (\trieoffx+0.594,\encodingy+0.135)
            node[pos=0.5,below=1.8pt,inner sep=0pt] (prefix_key_indicator_arrow_base) { };
            \node[inner sep=1pt,below right=12pt and 0pt of prefix_key_indicator_arrow_base] (prefix_node_indicator) {\small Prefix Node};
            \draw[-stealth] (prefix_key_indicator_arrow_base.south) -| (prefix_node_indicator.north);

            \draw[decorate,decoration={brace,raise=1pt,mirror,amplitude=2pt}] (\trieoffx+1.935,\encodingy+0.135) -- (\trieoffx+2.170,\encodingy+0.135)
            node[pos=0.5,below=1.5pt,inner sep=0pt] (children_bits_arrow_base) { };
            \node[inner sep=1pt,below=10pt of prefix_node_indicator.east,anchor=east] (children_bits) {\small Has Left Child, No Right Child};
            \draw[-stealth] (children_bits_arrow_base) |- (children_bits);

            \draw[decorate,decoration={brace,raise=1pt,amplitude=2pt}] ($(encoding_4.north)+(-2.172,-0.015)$) --($(encoding_4.north)+(-1.352,-0.015)$)
                node[pos=0.5,above=3pt,inner sep=1pt,align=center] (node_1) {\small $a$};
            \draw[decorate,decoration={brace,raise=1pt,amplitude=2pt}] ($(encoding_4.north)+(-1.26,-0.015)$) --($(encoding_4.north)+(-0.40,-0.015)$)
                node[pos=0.5,above=3pt,inner sep=1pt,align=center] (node_2) {\small $d$};
            \node[inner sep=5pt,anchor=south] (node_3) at ($(encoding_4.north)+(-0.26,-1.4pt)$) {\small $e$};
            \node[inner sep=5pt,anchor=south] (node_4) at ($(encoding_4.north)+(-0.065,-3.225pt)$) {\small $f$};
            \draw[decorate,decoration={brace,raise=1pt,amplitude=2pt}] ($(encoding_4.north)+(0.084,-0.015)$) --($(encoding_4.north)+(2.170,-0.015)$)
                node[pos=0.5,above=3pt,inner sep=1pt,align=center] (node_5) {\small $g$};
            \node[inner xsep=1pt,inner ysep=5pt,anchor=south] (node_6) at ($(encoding_4.north)+(2.32,-1.4pt)$) {\small $h$};
        \end{scope}

        \node[inner sep=0pt,align=left,anchor=north west] (fig_label_3) at (\colthreex-\labeloffx,\labely) {\small \textbf{C)}};

        \draw[-stealth] (\arrowthreex-\arrowlen/2,\arrowy) -- (\arrowthreex+\arrowlen/2,\arrowy)
            node[pos=0.5,above=-8.1pt,inner sep=1pt,align=center] (transition_3) {\small Deleting \\[-1pt] \small 11100011};

        \def\trieoffx{\colfourx+\triecenteroff}

        \begin{scope}[name prefix=t4_]
            \node[inner sep=\circlerad,fill=black,circle] (1) at (\trieoffx,\trieoffy) { };

            \node[inner sep=\circlerad,fill=black,circle] (2) at (\trieoffx,\trieoffy-1*\trieleveldiff) { };
            \draw (1) -- node[pos=0.5,inner sep=1pt,right=0pt] {\small 1} (2);
            \node[inner sep=0pt,above left=-1pt and 1pt of 2] (label_2) {\small $a$};

            \node[inner sep=\circlerad,fill=black,circle] (3) at (\trieoffx-\triehdiff,\trieoffy-2*\trieleveldiff) { };
            \draw (2) -- node[pos=0.25,inner sep=1pt,left=2pt] {\small 0} (3);
            \node[inner sep=0pt,above left=-1pt and 1pt of 3] (label_3) {\small \sout{$b$}};

            \node[inner sep=\circlerad,fill=black,circle] (5) at (\trieoffx-\triehdiff,\trieoffy-3*\trieleveldiff) { };
            \draw (3) -- node[pos=0.5,inner sep=1pt,right=0pt] {\small 0} (5);
            \node[inner sep=0pt,above left=-1pt and 1pt of 5] (label_5) {\small $d$};

            \node[inner sep=\circlerad,fill=black,circle] (6) at (\trieoffx-1.5*\triehdiff,\trieoffy-4*\trieleveldiff) { };
            \draw (5) -- node[pos=0.4,inner sep=1pt,left=0.25pt] {\small 0} (6);
            \node[inner sep=0pt,above left=-1pt and 1pt of 6] (label_6) {\small $e$};

            \node[inner sep=\circlerad,fill=black,circle] (7) at (\trieoffx-0.5*\triehdiff,\trieoffy-4*\trieleveldiff) { };
            \draw (5) -- node[pos=0.4,inner sep=1pt,right=0.25pt] {\small 1} (7);
            \node[inner sep=0pt,above right=-2.75pt and 1pt of 7] (label_7) {\small $f$};

            \node[inner sep=\circlerad,fill=black,circle] (4) at (\trieoffx+\triehdiff,\trieoffy-2*\trieleveldiff) { };
            \draw (2) -- node[pos=0.25,inner sep=1pt,right=2pt] {\small 1} (4);
            \node[inner sep=0pt,above right=-1pt and 1pt of 4] (label_4) {\small $c$};

            \node[inner sep=1pt,below=0pt of 6,align=center] (suffix_1) {\small 1 \\[-4pt] \small 1 \\[-4pt] \small 0};
            \node[inner sep=1pt,below=0pt of 7,align=center] (suffix_2) {\small 1};
            \node[inner sep=1pt,below=0pt of 4,align=center] (suffix_3) {\small 1 \\[-4pt] \small 0 \\[-4pt] \small 0};
            \node[inner sep=1pt,below=0pt of suffix_3,align=center] (truncated_suffix) {\small \textcolor{gray}{0} \\[-4pt] \small \textcolor{gray}{0} \\[-4pt] \small \textcolor{gray}{0} \\[-4pt] \small \textcolor{gray}{1}};
        \end{scope}

        \def\truncatedlineydiff{-21.4pt}
        \draw[thick, densely dotted] ($(t4_4.south)+(-6pt,\truncatedlineydiff)$) -- ($(t4_4.south)+(6pt,\truncatedlineydiff)$);
        \node[inner sep=1pt,below=0pt of t4_truncated_suffix,align=center] (truncated_suffix_label) {\small Truncated};

        \def\trieoffx{\colfourx}
        \node[inner sep=0pt,align=center] (encoding_6) at (\trieoffx,\encodingy) {\small 0 0 01 1 0 01 0 1 1 1 \\[0pt] \small 110\underline{1} 1\underline{100} 100\underline{1}};
        \begin{pgfonlayer}{background layer}
            \def\bith{0.15}
            \fill[fill=gray!20] (\trieoffx-1.145,\encodingy-\bith+0.215) rectangle (\trieoffx-0.935,\encodingy+\bith+0.215);
            \fill[fill=gray!20] (\trieoffx+0.3,\encodingy-\bith-0.190) rectangle (\trieoffx+0.88,\encodingy+\bith-0.155);
        \end{pgfonlayer}

        \begin{scope}[name prefix=t4b_label_]
            \draw[decorate,decoration={brace,raise=1pt,amplitude=2pt}] ($(encoding_6.north)+(-0.90,-0.015)$) --($(encoding_6.north)+(-0.25,-0.015)$)
                node[pos=0.5,above=3pt,inner sep=1pt,align=center] (node_1) {\small $a$};
            \draw[decorate,decoration={brace,raise=1pt,amplitude=2pt}] ($(encoding_6.north)+(-0.175,-0.015)$) --($(encoding_6.north)+(0.50,-0.015)$)
                node[pos=0.5,above=3pt,inner sep=1pt,align=center] (node_2) {\small $d$};
            \node[inner sep=5pt,anchor=south] (node_3) at ($(encoding_6.north)+(0.625,-1.4pt)$) {\small $e$};
            \node[inner sep=5pt,anchor=south] (node_4) at ($(encoding_6.north)+(0.835,-3.225pt)$) {\small $f$};
            \node[inner sep=5pt,anchor=south] (node_5) at ($(encoding_6.north)+(1.03,-1.4pt)$) {\small $c$};
        \end{scope}

        \node[inner sep=0pt,align=left,anchor=north west] (fig_label_4) at (\colfourx-\labeloffx,\labely) {\small \textbf{D)}};
    \end{tikzpicture}
    \caption{Diva++ inserts a key into a D-Trie by adding a branching path
    (Parts~A and B). As a result, the keys' trailing strings can shorten and
    vary in length, which Diva++ handles by padding the trailing strings to the
    same length with unary codes (Part B). If an inserted key matches an
    existing key along its path within the D-Trie and its trailing string,
    Diva++ extends the existing key's path. It marks the location along this
    path where the existing key was truncated with a prefix node (Part C,
    annotated with~$*$). Diva++ deletes a key from a D-Trie by removing the
    matching key with the longest path (Part~D).}
    \label{fig:updating_binary_tries}
\end{figure*}

\textbf{Ablation Study.}
Both order-preserving entropy encoding
(Section~\ref{sec:order_preserving_entropy_encoding})
and bit-level full key differentiation (this section and
Section~\ref{sec:dynamic_bit-level_full_key_diff})
play key roles in Diva++'s improved FPR. Table~\ref{tab:diva++_ablation}
illustrates this by comparing the FPR of Diva++ against variants applying each
technique in isolation, as well as the original Diva. We consider the three
datasets \textsc{EnWiki}~\cite{enwiki}, \textsc{Emails}~\cite{Enron}, and
\textsc{URLs}~\cite{MemeTracker}, and enforce a memory budget of~16 BPK. As
shown, the two techniques complement each other across all datasets, reducing
the FPR by as much as~$9\times$ relative to the original Diva design.
Experiment~\hyperlink{experiment:fpr_string}{3} of
Section~\ref{sec:evaluation_static}
further verifies this behavior under \mbox{different memory budgets}.

\subsection{Query Processing}\label{sec:diva++_queries}
To answer a range query, Diva++ first applies order-preserving entropy encoding
(Section~\ref{sec:order_preserving_entropy_encoding})
to the query endpoints. This allows the query endpoints to be compared directly
with the infixes of the entropy-encoded keys, without decompression. To answer
the entropy-encoded range query, Diva++ adopts a procedure similar to Diva's
(Section~\ref{sec:diva_queries}), where it
finds and searches the relevant runs within the corresponding Infix Store. The
core difference lies in how it checks the query against the infixes within a
run. Diva++ scans a run from left to right, skipping each infix lying outside
the query range. If such an infix has a \DT
(Section~\ref{sec:static_bit-level_full_key_diff}),
Diva++ skips it as well. As a \DT's size is unknown in advance, skipping it
entails sequentially parsing it to determine its size. Once a matching infix is
found, there are two cases: 

\emph{If the infix lies strictly in-between the encoded query endpoints,}
Diva++ immediately returns a true positive. Crucially, if this infix has a \DT,
Diva++ does not search or parse it. This is because any key within such a \DT
would also be strictly in-between the query endpoints, automatically leading to
a true positive.

\emph{If the infix matches either encoded query endpoint,}
Diva++ considers whether it has a \DT. If it does not, the query results in a
positive. Otherwise, Diva++ parses and searches the infix's \DT for a path and
a trailing string matching the query. Here, if a path is found that is strictly
in-between the entropy-encoded endpoints, the query results in a true positive,
similarly to how infixes strictly within the range yield true positives. 

If neither of the two cases above occur with a positive result, the query
returns a negative. Although parsing and searching \DTs raises the query cost
relative to Diva, it allows for better differentiating the keys and query
endpoints, significantly lowering the FPR.
Experiment~\hyperlink{experiment:fpr_string}{3} of
Section~\ref{sec:evaluation_static}
illustrates this trade-off.

\subsection{Dynamic Bit-Level Full Key Differentiation}\label{sec:dynamic_bit-level_full_key_diff}
In
Section~\ref{sec:static_bit-level_full_key_diff},
we described how Diva++ fully differentiates identical infixes using \DTs.
Under dynamic workloads, however, full differentiation is not always possible.
The difficulty arises because existing keys are truncated to save space.
Consequently, the bits of an existing key needed to distinguish it from an
insertion may have already been discarded. In this section, we describe how to
adapt the \DT structure to handle such cases and support both insertions and
deletions. We begin with insertions, which fall into three cases, each
requiring a different technique.

\textbf{Insertions Case (1): Divergence within the \DT.}
The first case occurs when the key being inserted diverges from some existing
key at a bit represented within the \DT. This case is straightforward: we
simply introduce a new branch at the first differentiating bit, and the two
keys remain fully distinguishable. As in
Section~\ref{sec:static_bit-level_full_key_diff},
if the inserted key is too short to provide a full-length trailing string, we
treat the bits beyond the key's end as~0s to form its trailing string while
preserving lexicographic ordering. Since an insertion enlarges a \DT's
encoding, Diva++ makes room for the new encoding by shifting physically
adjacent runs within the Infix Store; the same also applies to the other two
insertion cases.

\textbf{Insertions Case (2): Divergence along a Trailing String.}
The second case occurs when the inserted key matches some existing key in the
\DT, and first diverges from that key along its trailing string.
Figures~\ref{fig:updating_binary_tries}{\nobreakdash-}A
and
\ref{fig:updating_binary_tries}{\nobreakdash-}B
illustrate this case before and after inserting a key that matches the leftmost
existing key up to the second bit of its trailing string, where the two keys
differ.

We handle this case by converting the matching bits of the existing key's
trailing string into \DT edge labels, up to and including the first bit
differentiating it from the inserted key.
Figure~\ref{fig:updating_binary_tries}{\nobreakdash-}B
presents the resulting \DT after the insertion. As shown, the insertion leaves
a one-bit trailing string at leaf~$f$, because the existing key's truncation
has left too few bits to maintain a trailing string the same length as before.
Thus, unlike in the static setting of
Section~\ref{sec:static_bit-level_full_key_diff},
trailing strings may have variable length \mbox{in the dynamic setting}.

To accommodate variable-length trailing strings, we pad the trailing strings to
a common width with unary codes of the form~\underline{$10\dots00$} (analogous
to how Diva pads variable-length infixes to the slot width in
Section~\ref{sec:diva_dynamic}). We append a
unary code to the right of each trailing string, allowing the
delimiter~\underline{1} bit to separate the trailing string from the padding.
The delimiter bit elongates a trailing string's encoding by one bit. 

The encoding at the bottom of
Figure~\ref{fig:updating_binary_tries}{\nobreakdash-}A
illustrates the unary padding (underlined). In this case, the trailing strings
are all of full length, so only the delimiter bit of each unary code is stored.
Due to an insertion, the first two bits of the right key's trailing string
become edge labels within the \DT in
Figure~\ref{fig:updating_binary_tries}{\nobreakdash-}B.
This causes the trailing string's remaining bit (1) to be stored with unary
padding~\underline{100} following it, as shown at the bottom of the figure. The
newly introduced path and the modified trailing string are highlighted. Note
that since node~$b$ becomes a non-branching internal node after the insertion,
it no longer appears explicitly in the encoding, which is indicated by its
crossed-out label in
Figure~\ref{fig:updating_binary_tries}{\nobreakdash-}B.

\textbf{Insertions Case~(3): No Divergence.}
The third case occurs when the inserted key fully matches an existing key along
both its \DT path and trailing string. That is, the portion of the existing key
remaining after truncation is a prefix of the inserted key. Because the bits
distinguishing the existing key from the inserted key were discarded, the two
keys cannot be fully differentiated.
Figure~\ref{fig:updating_binary_tries}{\nobreakdash-}C
shows an example where the inserted key matches the rightmost key at leaf~$c$
in
Figure~\ref{fig:updating_binary_tries}{\nobreakdash-}B.

To handle this case, Diva++ converts the entire trailing string of the existing
key into \DT edge labels and adds a new leaf and trailing string for the
insertion to the end of the resulting path. We must now explicitly mark where
the existing key ends on the new path to prevent a range query including this
key but not the newly inserted key from returning a false negative. To this
end, we introduce a new type of internal node called a \emph{prefix
node}.\footnote{Prefix nodes also prove useful in static \DTs. In that setting,
we may have a key that is a proper prefix of another key, and the longer key
has only~0s past the end of the shorter key. When this happens, the usual
strategy of treating the bits beyond the shorter key's end as~0s does not help
with differentiating it. Therefore, we use a prefix node to represent the
shorter key instead.} Node~$g$ in
Figure~\ref{fig:updating_binary_tries}{\nobreakdash-}C
is a prefix node (marked with~$*$).

\textbf{Indicating a Prefix Node.}
We now extend BITR to allow encoding prefix nodes within a \DT. The central
challenge is indicating that a node in the encoding is a prefix node. Recall
from Section~\ref{sec:static_bit-level_full_key_diff}
that each node's encoding begins with a bit indicating whether it is a branch
or a leaf. A straightforward approach would be to extend this field to two bits
and reserve a bit combination for prefix nodes. However, doing so would enlarge
the encoding of every branch and leaf, substantially increasing the size of a
\DT even when it has no prefix nodes.

Our design takes a different approach, motivated by two observations: First,
BITR uses the Elias-Gamma format much less frequently in practice (e.g., on the
\textsc{Emails} and \textsc{URLs} datasets, only~15\%~and~14\% of the \DTs use
it, respectively). Second, insertions in most dynamic workloads fall into
cases~(1) or (2), and thus do not introduce prefix nodes (e.g., on both the
\textsc{Emails} and \textsc{URLs} datasets, less than~$0.1\%$ of the \DTs
contain a prefix node).

We therefore require any \DT containing a prefix node to use the Elias-Gamma
format and augment this format to indicate a prefix node using an escape
sequence. Doing so confines the overhead of supporting prefix nodes to the
relatively small fraction of \DTs that use the Elias-Gamma representation,
leaving the more common unary representation unchanged. In
Figure~\ref{fig:updating_binary_tries}{\nobreakdash-}C,
for example, the encoding transitioned from unary to the Elias-Gamma format to
accommodate the new prefix node, as shown by the encoding at the bottom of the
figure. The format bit is now set to~1, and the nodes~$a$~and~$b$, which
previously used the unary representation in
Figure~\ref{fig:updating_binary_tries}{\nobreakdash-}B,
are re-encoded using Elias-Gamma.

We define the required escape sequence as a~0~bit, which normally signals a
branch, along with the path length value~1,0 previously reserved in
Section~\ref{sec:static_bit-level_full_key_diff}.
Thus, the escape sequence for a prefix node is 010.
Figure~\ref{fig:updating_binary_tries}{\nobreakdash-}C
indicates that node~$g$ is a prefix node using this escape sequence. 

The reserved path length slightly increases the size of the Elias-Gamma
encoding of other path lengths, yet because relatively few \DTs use this
format, its effect on the average \DT encoding size remains small. In fact,
even when many \DTs use the Elias-Gamma format, the additional space overhead
(of both the escape sequence and the Elias-Gamma codes) is offset by BITR not
having to store a trailing string for the key represented by the prefix node.

\textbf{Encoding the Rest of Prefix Node.}
After the escape sequence, we encode the path from the prefix node's nearest
ancestor branch (or prefix node) as in
Section~\ref{sec:static_bit-level_full_key_diff}.
That is, we store an Elias-Gamma code for the path length followed by the
corresponding edge labels. As before, we omit the path's first edge label,
since it is inferrable from the \DT's topology. In
Figure~\ref{fig:updating_binary_tries}{\nobreakdash-}C,
the bits following the escape sequence encode a path length of four followed
by~100, which is the path's edge labels except for the omitted first one.

Finally, unlike a branch, which always has two children, a prefix node may have
only a left child, only a right child, or both. We therefore append two bits
indicating the presence of the left and right children, respectively. For
node~$g$ in
Figure~\ref{fig:updating_binary_tries}{\nobreakdash-}C,
these bits are~10, indicating that the left child is present and \mbox{the
right child is not}.

\textbf{Deletions.}
Diva++ handles deletions similarly to
Section~\ref{sec:diva_dynamic}. That is, it
removes a matching existing key with the longest \DT path to avoid introducing
false negatives in the future. As before, users must ensure deletions only
target existing keys only, since a filter cannot verify a deletion's validity.

A deletion may reduce the subtree of an ancestor of the removed node to a path
with no branches or prefix nodes. In
Figures~\ref{fig:updating_binary_tries}{\nobreakdash-}C~and~\ref{fig:updating_binary_tries}{\nobreakdash-}D,
for example, a deletion removes the prefix node~$g$, reducing the path from
node~$c$ to leaf~$h$ into a branchless, prefix-node-free path. Diva++ avoids
the metadata overhead of encoding the path's topology in BITR by removing the
path from the \DT and prepending its edge labels to the trailing string
previously stored at its end. If the resulting trailing string has more
than~$t$ bits (the original trailing string length, not counting unary
padding), Diva++ truncates it to~$t$ bits to ensure unambiguous decoding. In
Figure~\ref{fig:updating_binary_tries}{\nobreakdash-}D,
this process leads to a seven-bit trailing string~1000001 at leaf~$c$, which is
truncated to the original trailing string length of~$t=3$ bits.

When a \DT has no prefix nodes remaining after a deletion, encoding it in the
unary format may be more space-efficient. To determine whether this is the
case, we use the unary format's specification to compute the encoding's length
in this format. If it is smaller than the current encoding size, BITR
re-encodes the \DT in the unary format.
Figure~\ref{fig:updating_binary_tries}{\nobreakdash-}D
shows an example where switching back to the unary format saves space. As
shown, the format bit is set to~0 and the path lengths of the branches~$a$
and~$d$ revert to using the unary codes they had in
Figure~\ref{fig:updating_binary_tries}{\nobreakdash-}C.

\textbf{Queries.}
Diva++ answers a query similarly to
Section~\ref{sec:diva++_queries}. The only
difference is in how it handles the case where one of the query endpoints
matches an infix with a \DT. Since the bits missing from a key are unknown,
Diva++ compares only the available bits to the query endpoints when searching
the \DT for a key that lies in the range. Specifically, when Diva++ encounters
a prefix node, it compares only the bits on the path leading to the prefix node
with the endpoints. Similarly, when it encounters a trailing string with
missing bits, it compares only the available bits. The query returns a positive
if the compared bits match one of the endpoints or lie strictly in-between
them.

\subsection{Concurrency Support}\label{sec:concurrency}
Diva++ extends Diva to support concurrent operations using read-write locks.

\textbf{Read-Write Locks.}
We protect each Infix Store using a read-write lock. Since each Infix Store
contains~$\approx T = 1024$ keys (i.e., a small fraction of the dataset), there
is little lock contention when updates are independent and identically
distributed. The \ST also experiences low contention, since an insertion (resp.
deletion) of a sample occurs once in every~\mbox{$\approx T = 1024$} insertions
(resp. deletions). We use a concurrent variant of Wormhole that uses read-write
locks as our \ST.

Query threads take read locks over the \ST and the relevant Infix Stores.
Insertion and deletion threads take read locks over the \ST and write locks
over the Infix Stores. They upgrade their read locks on the \ST to write locks
only when it must be modified.

\textbf{Crabbing.}
Diva++ further reduces contention by employing the well-known crabbing
technique used in concurrent B-trees~\cite{Crabbing1,Crabbing2}: Threads first
take the required \ST locks, proceed to lock their target Infix Store to ensure
that it is not split or merged by another thread, and release the \ST lock if
the \ST does not require modification. This technique allows each thread to
only hold locks on its Infix Store and free up the \ST's structures for
modification by other threads. It also ensures
linearizability~\cite{Linearizability}, meaning that once a thread finishes an
update operation, the effects are visible to all future queries and updates.

\textbf{Read-Write Lock Implementation.}
We implement a read-write lock using compare-and-swap loops over a single byte,
thereby incurring a negligible memory overhead ($\approx 0.008$ BPK). The~7
lower-order bits of this byte track the number of querying threads accessing
the shared resource, and the higher-order bit indicates if an update thread is
modifying the resource. This layout supports a total of~128 threads. Using
larger data types for the lock would allow for supporting more threads.

\section{Theoretical Analysis and Results}\label{sec:theoretical_analysis_and_results}
\new{We prove Diva's semi-robustness by defining well-behaved data
distributions and using this notion to bound the FPR when the dataset comes
from such a distribution. We further argue that Diva++ moves this guarantee
closer to a robust one by bounding the FPR for a larger class of data
distributions. We also show that both Diva and Diva++ have low operation costs,
and that their FPR-space trade-off is optimal.} The last four rows in
Table~\ref{tab:method_complexities}
summarize this section's results.

\textbf{Setting.}
Formally, we consider the dataset's keys to be sampled independently from a
distribution with a cumulative distribution function (CDF) $F$ and probability
density function (PDF) $f$. For our analysis, we treat string keys as real
numbers in the range~$[0, 1)$ by interpreting their binary representation as a
fractional binary value. For example, we interpret the strings~$(1000)_2$ and
$(1100)_2$ as the real numbers $(0.1000)_2=0.5$ and $(0.1100)_2=0.75$.

\new{
The notion of FPR we adopt and analyze is a ``\emph{data-oblivious}'' one: the
query is chosen independently from the dataset. That is, the user has a query
in mind before they gain any knowledge of the specific dataset drawn from the
data distribution. This is a weaker notion than the classic
``\emph{data-aware}'' FPR, which also allows the query to depend on the
specific dataset (e.g., the query may be chosen to be just to the right of the
smallest key in the dataset). Nevertheless, data-oblivious robustness is
sufficient for filter data structures, as filters are used in settings where
users lack knowledge of the dataset and seek to gain information to avoid
redundant accesses to storage or the network. Interestingly, Goswami et al.'s
memory lower bound for range filters~\cite{Goswami} assumes a data-aware notion
of FPR, and does not apply to range filters with data-oblivious FPR guarantees.
Because of this, we are able to prove that Diva and Diva++ attain a better
FPR-memory footprint trade-off than the one dictated by this lower bound.
}

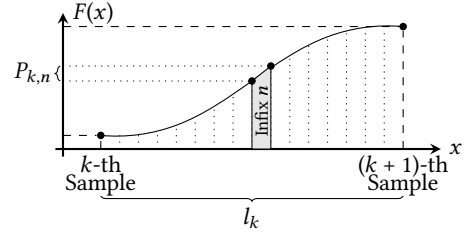
\begin{figure}
    \centering
    \begin{tikzpicture}
        \def\ox{-2}
        \def\oy{-2}
        \def\xaxislength{5}
        \def\yaxislength{1.8}
        \def\ooverlen{0.1}
        \def\xk{0.1}
        \def\xkk{0.9}
        \def\pk{0.1}
        \def\pkk{0.9}

        \draw[thick,-stealth] (\ox-\ooverlen,\oy) -- (\ox+\xaxislength,\oy);
        \draw[thick,-stealth] (\ox,\oy-\ooverlen) -- (\ox,\oy+\yaxislength);
        \node[inner sep=1pt,circle,fill=black,align=center] (start) at (\ox+\xk*\xaxislength,\oy+\pk*\yaxislength) {};
        \node[inner sep=1pt,circle,fill=black,align=center] (end) at (\ox+\xkk*\xaxislength,\oy+\pkk*\yaxislength) {};
        \draw plot[hobby] coordinates { (start) 
        (\ox+\xk*\xaxislength+0.2*\xaxislength,\oy+\pk*\yaxislength+0.075*\yaxislength) 
        (\ox+\xkk*\xaxislength-0.2*\xaxislength,\oy+\pkk*\yaxislength-0.075*\yaxislength) 
        (end) };
        \node[inner sep=0pt,align=center] (x_label) at (\ox+\xaxislength+2*\ooverlen,\oy) {$x$};
        \node[inner sep=0pt,align=center] (y_label) at (\ox+4*\ooverlen,\oy+\yaxislength) {$F(x)$};

        \draw[dashed] (\ox+\xk*\xaxislength,\oy) -- (start);
        \node[inner sep=2pt,anchor=north,align=center] (x_k) at (\ox+\xk*\xaxislength,\oy) {$k$-th \\[-4pt] Sample};
        \draw[dashed] (\ox+\xkk*\xaxislength,\oy) -- (end);
        \node[inner sep=2pt,anchor=north,align=center] (x_k1) at (\ox+\xkk*\xaxislength,\oy) {$(k+1)$-th \\[-4pt] Sample};

        \draw[dashed] (\ox,\oy+\pk*\yaxislength) -- (start);
        \draw[dashed] (\ox,\oy+\pkk*\yaxislength) -- (end);

        \draw[decorate,decoration={brace,raise=1pt,mirror,amplitude=2pt}] (\ox+\xk*\xaxislength,\oy-0.325*\yaxislength) -- (\ox+\xkk*\xaxislength,\oy-0.325*\yaxislength)
        node[pos=0.5,below=1pt,inner sep=4pt,align=center] (l_k) {$l_k$};

        \draw[dotted] (\ox+0.15*\xaxislength,\oy) -- (\ox+0.15*\xaxislength,\oy+0.10*\yaxislength);
        \draw[dotted] (\ox+0.20*\xaxislength,\oy) -- (\ox+0.20*\xaxislength,\oy+0.125*\yaxislength);
        \draw[dotted] (\ox+0.25*\xaxislength,\oy) -- (\ox+0.25*\xaxislength,\oy+0.13*\yaxislength);
        \draw[dotted] (\ox+0.30*\xaxislength,\oy) -- (\ox+0.30*\xaxislength,\oy+0.175*\yaxislength);
        \draw[dotted] (\ox+0.35*\xaxislength,\oy) -- (\ox+0.35*\xaxislength,\oy+0.225*\yaxislength);
        \draw[dotted] (\ox+0.40*\xaxislength,\oy) -- (\ox+0.40*\xaxislength,\oy+0.295*\yaxislength);
        \draw[dotted] (\ox+0.45*\xaxislength,\oy) -- (\ox+0.45*\xaxislength,\oy+0.39*\yaxislength);
        \draw[dotted] (\ox+0.50*\xaxislength,\oy) -- (\ox+0.50*\xaxislength,\oy+0.50*\yaxislength);
        \draw[dotted] (\ox+0.55*\xaxislength,\oy) -- (\ox+0.55*\xaxislength,\oy+0.61*\yaxislength);
        \draw[dotted] (\ox+0.60*\xaxislength,\oy) -- (\ox+0.60*\xaxislength,\oy+0.70*\yaxislength);
        \draw[dotted] (\ox+0.65*\xaxislength,\oy) -- (\ox+0.65*\xaxislength,\oy+0.78*\yaxislength);
        \draw[dotted] (\ox+0.70*\xaxislength,\oy) -- (\ox+0.70*\xaxislength,\oy+0.83*\yaxislength);
        \draw[dotted] (\ox+0.75*\xaxislength,\oy) -- (\ox+0.75*\xaxislength,\oy+0.85*\yaxislength);
        \draw[dotted] (\ox+0.80*\xaxislength,\oy) -- (\ox+0.80*\xaxislength,\oy+0.88*\yaxislength);
        \draw[dotted] (\ox+0.85*\xaxislength,\oy) -- (\ox+0.85*\xaxislength,\oy+0.89*\yaxislength);

        \draw[fill=gray!20] (\ox+0.50*\xaxislength,\oy) -- (\ox+0.55*\xaxislength,\oy) -- (\ox+0.55*\xaxislength,\oy+0.61*\yaxislength) -- (\ox+0.50*\xaxislength,\oy+0.50*\yaxislength) -- cycle;
        \node[inner sep=0pt,align=center,rotate=90] (infix_label) at (\ox+0.525*\xaxislength,\oy+0.28*\yaxislength) {\footnotesize Infix $n$};
        \node[inner sep=1pt,circle,fill=black,align=center] at (\ox+0.50*\xaxislength,\oy+0.50*\yaxislength) {};
        \node[inner sep=1pt,circle,fill=black,align=center] at (\ox+0.55*\xaxislength,\oy+0.61*\yaxislength) {};

        \draw[dotted] (\ox,\oy+0.50*\yaxislength) -- (\ox+0.50*\xaxislength,\oy+0.50*\yaxislength);
        \draw[dotted] (\ox,\oy+0.61*\yaxislength) -- (\ox+0.55*\xaxislength,\oy+0.61*\yaxislength);
        \draw[decorate,decoration={brace,raise=1pt,amplitude=2pt}] (\ox,\oy+0.50*\yaxislength) -- (\ox,\oy+0.61*\yaxislength) 
        node[pos=0.5,left=3pt,inner sep=1pt] (P_kn) {$P_{k,n}$};
    \end{tikzpicture}
    \caption{Infix Stores partition the key space into short ranges. Each of
    these ranges is partitioned into equal-width sub-ranges, each of which is
    condensed into an infix.}
    \label{fig:infix_key_space_partitioning_paper}
\end{figure}

\textbf{Infix Probability.}
Diva returns a false positive when a query and an input key from the data
distribution map to the same infix within the same Infix Store. We bound the
probability of this event by calculating the probability that a specific infix
exists. 

Intuitively, Diva's Infix Stores partition the key space into shorter ranges.
Diva discretizes each of these ranges by deriving infixes, further partitioning
it into~$\frac{2T}{\epsilon}$ equal-width sub-ranges, one for each possible
infix. Figure~\ref{fig:infix_key_space_partitioning_paper} illustrates this
partitioning for the $k$-th Infix Store. The probability that a particular
input key maps to a specific infix is equal to the probability that it lies in
the corresponding sub-range, which depends on the data distribution's CDF. We
denote this probability for infix~$n$ in the $k$-th Infix Store as~$P_{k,n}$.
We can guarantee an FPR of at most~$\epsilon$ by showing that~$P_{k,n} \leq
\frac{\epsilon}{N}$ and applying a union bound to all keys in the dataset.

\textbf{Example: Uniform Data Distribution.}
To give intuition why $P_{k,n}$ is small for well-behaved data distributions,
we first discuss the case where the data distribution is uniform, i.e.,
$F(x)=x$. Here, $P_{k,n}$ is exactly equal to the length of its infix's
sub-range. That is, it is equal to~$\frac{\epsilon}{2T} \cdot l_k$, where $l_k$
is the length of the range spanned by the $k$-th Infix Store, as shown in
Figure~\ref{fig:infix_key_space_partitioning_paper}.
Since there are~$\frac{N}{T}$ many Infix Stores, $l_k$ is in expectation
$\frac{T}{N}$. Thus, the probability of an input key landing in this sub-range
is equal to~$\frac{\epsilon}{2N}$, which is~less~than~$\frac{\epsilon}{N}$.

\textbf{Well-Behaved Data Distributions.}
In general, the CDF of the data distribution could make~$P_{k,n}$ arbitrarily
large, depending on how abruptly it changes within that particular sub-range.
Nevertheless, we can make a similar argument for \emph{well-behaved} data
distributions, defined as:
\begin{definition}[Well-Behaved Data Distribution]
    Let~$f'$ be the derivative of~$f$. A data distribution is well-behaved if,
    for~$\delta = T \cdot \frac{\log N}{N}$, a constant~$\alpha \in [0,1/4]$
    exists such that for all~$x$, 
    $ \delta^{\alpha} \leq f(x) \leq \delta^{-1 + 4\alpha} $
    and $f'(x) \leq \frac{1}{3} \cdot \delta^{-\alpha}. $
    \label{definition:well-behaved}
\end{definition}
The conditions above make well-behaved data distributions ``smooth.'' The condition
on~$f$ keeps the distribution's CDF from radically changing, while the
condition on~$f'$ limits the acceleration of change. Together, they control how
much the CDF can climb in an infix's sub-range, which is~$P_{k,n}$, as shown in
Figure~\ref{fig:infix_key_space_partitioning_paper}. We remark that the bounds
on these quantities grow polynomially in~$N$, thus becoming easier to satisfy
as the dataset grows. For example, $f'$ is allowed to take values as high
as~$\frac{\delta^{-\alpha}}{3} \geq \frac{\delta^{-1/4}}{3} \approx
\frac{N^{1/4}}{3}$. Intuitively, this is because~$P_{k,n}$ corresponds to a
smaller sub-range with larger~$N$, so our analysis can afford less smooth
distributions.

Many natural data distributions are well-behaved, as one can verify by plugging
their PDFs into
Definition~\ref{definition:well-behaved}.
For example, the uniform distribution is trivially well-behaved since it
has~$f(x)=1$ and $f'(x)=0$. Normal distributions with a standard deviation
of~$\sigma \gtrapprox \frac{3}{\sqrt{\pi} \cdot N^{1/5}}$ and power law
distributions (which generalize the Zipfian distribution) with an exponent
of~$\lambda \leq \frac{\log_2 N}{4}$ are also well-behaved. Note that~$\sigma$
and $\lambda$ often satisfy these conditions.


\new{In \apx{proof:uniformity-new},} we use the properties of well-behaved data
distributions in an argument similar to that of the uniform distribution to
show:
\begin{restatable}{theorem}{FPRnew}
    If the input distribution is well-behaved, then, with high probability over
    Diva's samples, in the static (dynamic) case, a particular infix exists
    with (expected) probability at most~$\epsilon$. 
    \label{theorem:uniformity-new}
\end{restatable}

\textbf{FPR.}
We now use Theorem~\ref{theorem:uniformity-new}
to bound Diva's FPR when the dataset comes from a well-behaved distribution:
\begin{theorem}[Static FPR]
    When the input distribution is well-behaved, the static variant of Diva
    returns a false positive for a point query~$q$ with probability at
    most~$\epsilon$. Moreover, it returns a false positive for a range
    query~$[q_l,q_r]$ with a probability of at most~$2\epsilon$.
    \label{theorem:static_FPR}
\end{theorem}
\begin{proof}
    Consider the point query~$q$. Since we are analyzing the probability of a
    false positive, we can suppose that~$q$ is not in the dataset and thus does
    not equal any sample in Diva's \ST. Therefore, Diva checks if an infix
    corresponding to~$q$ exists to answer the query. By
    Theorem~\ref{theorem:uniformity-new},
    such an infix exists with a probability of at most~$\epsilon$, implying an
    FPR of at most~$\epsilon$ for~$q$.

    Now, consider the range query~$[q_l, q_r]$. As before, we can suppose that
    this is an empty range, implying that it includes none of the samples in
    the \ST and thus falls in a single Infix Store. Denoting the infixes
    corresponding to~$q_l$ and $q_r$ by $l$ and $r$, one can see that there
    must be no other infix in the Infix Store that is strictly between~$l$ and
    $r$, as it would cause the filter to return a true positive. Thus, Diva can
    only mistakenly return a false positive if some infix equals either~$l$ or
    $r$. By applying
    Theorem~\ref{theorem:uniformity-new},
    one can show that this happens with a probability of at most~$2\epsilon$.
\end{proof}

The FPR analysis of the dynamic variant of Diva follows a proof similar to that
of Theorem~\ref{theorem:static_FPR}. More specifically, since the
filter is initially constructed using the static bulk-loading method, it starts
with point and range query FPRs of $\epsilon$ and $2\epsilon$. After inserting
$N$ new keys into the filter with no distribution shifts, one can use an
analysis similar to that of InfiniFilter~\cite{InfiniFilter} to show that:
\begin{theorem}[Dynamic FPR]
    The dynamic variant of Diva returns a false positive for point query $q$
    with expected probability at most $\epsilon/2 \cdot (\log_2 N + 2)$.
    Moreover, it returns a false positive for a range query $[q_l,q_r]$ with
    expected probability at most $\epsilon \cdot (\log_2 N + 2)$.
    \label{theorem:dynamic_FPR}
\end{theorem}

\new{
Since Diva++ expands upon Diva's techniques, it provides the same semi-robust
FPR guarantee for well-behaved data distributions. Moreover, when using~$k$-th
order entropy encoding, Diva++'s guarantees extend to datasets from jagged
distributions that do not have patterns of a higher order than~$k$. This is
because~$k$-th order entropy encoding ensures that the infixes resulting from
these datasets are uniformly distributed~\cite{ShannonInformationTheory},
making Theorem~\ref{theorem:uniformity-new}
and thus the same FPR guarantees as Theorems~\ref{theorem:static_FPR}
and \ref{theorem:dynamic_FPR} trivially hold. In
particular, this implies that an infinite-order entropy encoding scheme would
grant Diva++ a data-oblivious FPR guarantee that holds for any data
distribution, making it robust.
}

\textbf{Memory Footprint.}
Both static and dynamic Diva variants keep a~$\frac{1}{T}$ fraction of the keys
in their \ST. For each, they also store a total of 96 bits of metadata,
representing whether the sample is a partial sample, the size of the Infix
Store to its right, the number of infixes it stores, and a pointer to its
array. Assuming that the keys have an average length of $L$ bits, the \ST
consumes a total of~$\frac{(96+L) \cdot N}{T}$ bits.

In the static variant, each Infix Store has $T$ slots. It consumes $T$ bits for
the \texttt{occupieds} bitmap, one bit per slot for the \texttt{runends}
bitmap, and $\log_2 \frac{2}{\epsilon}$ bits for each slot. Since there are a
total of~$N/T$ Infix Stores, the total memory footprint of the Infix Stores is
$N/T \cdot (T \cdot (3 + \log_2 \frac{1}{\epsilon}))$ bits. Summing up the \ST
and Infix Store costs, noting that $T=1024$, and dividing by $N$ results in a
memory footprint of $\approx
3.09+\frac{L}{1024}+\log_2\frac{1}{\epsilon}$~BPK.

In the dynamic variant, an Infix Store with $n$ infixes uses at most
$\frac{n}{\alpha^2}$ slots. Since there are a total of $\frac{T-1}{T} \cdot N$
infixes, this translates to a total of $\frac{T-1}{T \cdot \alpha^2} \cdot N
\approx \frac{1}{\alpha^2} \cdot N$ slots. Moreover, each slot is one bit wider
than the static case to accommodate the unary counter. As there are an expected
of $N/T$ Infix Stores, this results in a total memory footprint of $\approx
1.09+\frac{L}{1024}+\frac{1}{\alpha^2}\cdot(2+\log_2 \frac{1}{\epsilon})$~BPK.

\new{
When the dataset follows a well-behaved distribution, Diva++ maintains the same
memory footprint as Diva in both static and dynamic settings, as it store store
a negligible number of \DTs due to the rarity of identical infixes. Under
jagged data distributions, Diva++ can potentially store more bits from each key
in its \DTs if there are many higher order patterns in the data that are
unexploited by the entropy encoding. In our experiments, the \DTs did not incur
an noticable memory overhead, as they occupied space reclaimed by removing
duplicate infixes. This is because, as the following theorem states, the
encoding of a \DT consumes at most~2 bits per node:
\begin{theorem}[D-Trie Memory Footprint]
    A \DT's encoding, save for the trailing bits, consumes at most~2 bits per
    node of space.
\end{theorem}
\begin{proof}
    A \DT's encoding uses whichever of the unary and the binary formats consume
    less space. As such, we prove the desired claim by bounding the memory
    footprint of the unary format. Observe that a branch node's encoding
    consumes exactly two bits for each node on the path it encodes. Combined
    with how each leaf is encoded using a single bit, this implies that the
    entire encoding consumes at most~2 bits per node of space.
\end{proof}
}

\textbf{Performance.}
Each of Diva \new{and Diva++'s} operations search for the
predecessor and successor of the query endpoints in the \ST. Since we use
Wormhole as the underlying \ST, and because Wormhole searches for an~$L$-bit
key using at most~$O(\log_2 L)$ cache misses, this search incurs a total
of~$O(\log_2 L)$ cache misses. Moreover, as the \texttt{occupieds} and
\texttt{runends} bitmaps are $T=1024$ and~$\frac{T}{\alpha}\approx1078$ bits
long, and since a typical cache line is 512 bits long, applying rank and select
operations on them incurs a constant number of cache misses.
\new{Finally, reading/modifying the slots in the Infix Store incurs
a single cache miss in expectation. Thus, Diva and Diva++} incur an expected
of~$O(\log_2 L + 1) = O(\log_2 L)$ cache misses. Notice that when keys are
fixed-length, the number of cache misses becomes constant.

\begin{table}
    \centering
    \caption{The memory footprint lower bound for supporting range queries of
    length~$R$ with an FPR of~$\epsilon$ weakens when the desired FPR guarantee
    is relaxed from a data-aware one to a data-oblivious one.}
    \def\arraystretch{2.0}
    \begin{tabular}{ccc}
        \toprule \\[-28pt]

        & \small \textbf{Robust} & \small \textbf{Semi-Robust} \\[-2pt]

        \midrule

        \small \textbf{Data-Aware} 
            & \makecell[c]{\small $\log_2 \frac{R}{\epsilon} - O(1)$ \\ \small Goswami et al.~\cite{Goswami}} 
            & \makecell[c]{\small $(1-o(1)) \cdot \log_2 \frac{R}{\epsilon} - O(1)$ \\ \small Theorem~\ref{theorem:data_aware_memory_lower_bound}} \\

        \small \textbf{Data-Oblivious} 
            & \makecell[c]{\small $\log_2 \frac{1}{\epsilon} - O(1)$ \\ \small Theorem~\ref{theorem:data_oblivious_memory_lower_bound}}
            & \makecell[c]{\small $\log_2 \frac{1}{\epsilon} - O(1)$ \\ \small Theorem~\ref{theorem:data_oblivious_memory_lower_bound}} \\

        \bottomrule
    \end{tabular}
    \label{tab:fpr_notions}
\end{table}

\new{
\textbf{Theoretical Limits.}
Diva and Diva++ achieve a lower memory footprint than the~$\log_2
\frac{R}{\epsilon}$ BPK lower bound of Goswami et al.~\cite{Goswami}. Here, $R$
is the maximum range query length one would want to support with an FPR
of~$\epsilon$. This improvement is because Diva and Diva++ provide a
data-oblivious FPR guarantee as opposed to a data-aware one. A natural
question, therefore, is whether one can improve upon Diva and Diva++'s FPR
guarantees. We answer this question negatively by proving several memory lower
bounds for various robustness and FPR requirements.
Table~\ref{tab:fpr_notions}
summarizes our resutls and compares them with the result of Goswami et al.
}

\new{
We first show that one cannot significantly improve upon Diva and Diva++'s
data-oblivious FPR guarantee at the same memory footprint. We do this in
\apx{proof:data_oblivious_memory_lower_bound} through an encoding argument
similar to that of Goswami et al., yielding the following theorem:
\begin{theorem}[Data-Oblivious Memory Lower Bound]
    Any range filter, semi-robust or robust, that attains a data-oblivious FPR
    of at most~$\epsilon$ must use a total memory of~$N \cdot \log_2
    \frac{1}{\epsilon} - O(N)$ bits on average, which is equivalent
    to~$\log_2\frac{1}{\epsilon}-O(1)$ BPK.
    \label{theorem:data_oblivious_memory_lower_bound}
\end{theorem}
}

\new{
Additionally, in \apx{proof:data_aware_memory_lower_bound}, we provide a simple
extension of Goswami et al.'s lower bound~\cite{Goswami}, showing that it is
impossible to design a range filter that extends Diva's data-oblivious
semi-robust guarantee to a data-aware semi-robust guarantee without incurring a
memory footprint proportional to the range query length: 
\begin{theorem}[Data-Aware Memory Lower Bound]
    Any semi-robust range filter that answers queries of length~$R$ with a
    data-aware FPR of at most~$\epsilon$ must use a total memory of~$(1-o(1))
    \cdot N \cdot \log_2 \frac{R}{\epsilon} - O(N)$ bits, which is equivalent
    to~$(1-o(1)) \cdot \log_2\frac{R}{\epsilon}-O(1)$ BPK. In other words,
    semi-robust range filters must use the same amount of memory as robust ones
    if the FPR must be data-aware.
    \label{theorem:data_aware_memory_lower_bound}
\end{theorem}
In particular, this theorem implies that supporting variable-length keys or
range queries with a data-aware FPR requires storing the full dataset, since in
these cases we have~$R=\infty$.
}

\new{
\textbf{Simultaneous Data-Awareness and Data-Obliviousness.}
The following natural question remains unaswered: \emph{Is it possible to
design a semi-robust range filter with a data-oblivious FPR guarantee that also
supports ``short or local'' range queries with a robust data-aware FPR
guarantee?} Formally, such a range filter has the following two properties when
processing fixed-length keys:
\begin{enumerate}
    \item For ``short'' range queries of length at most~$R$, the filter attains
        a data-aware FPR of at most~$\epsilon$.
    \item For any range query, it attains a semi-robust data-oblivious FPR of
        at most~$\epsilon$.
\end{enumerate}
A simple and dynamic approach to achieving these two properties over
fixed-length keys is to combine a Memento filter~\cite{Memento} instance (with
a memory footprint of~$\approx \log_2 \frac{R}{\epsilon}$ BPK) and Diva (with a
memory footprint of~$\approx \log_2 \frac{1}{\epsilon}$ BPK). The idea is to
use Memento filter to answer range queries of length at most~$R$ with a
data-aware FPR guarantee, while using Diva to answer all other range queries
with a data-oblivious FPR guarantee.
\ifarxiv
\Apx{proof:combined_range_filtering_lower_bound}
\fi
shows that this approach is space-optimal:
\begin{theorem}
    When the keys are fixed length, any range filter that answers range queries
    of length at most~$R$ with a data-aware FPR of at most~$\epsilon$ and
    data-oblivious range queries with a semi-robust FPR of at~$\epsilon$ must
    use a total memory of~$N \cdot \log_2 \frac{R^{1-O(\epsilon)}}{\epsilon^2}
    - O(N)$ bits on average, which is equivalent
    to~$\log_2{\frac{R}{\epsilon}}+\log_2{\frac{1}{\epsilon}}-O(1)$ BPK.
    \label{theorem:combined_range_filtering_lower_bound}
\end{theorem}
}

\section{Evaluation}\label{sec:evaluation}
We evaluate \new{Diva and Diva++} against existing range filters in
static and dynamic settings in Sections~\ref{sec:evaluation_static}
and \ref{sec:evaluation_dynamic},
respectively. We also conduct end-to-end evaluations on top of
WiredTiger~\cite{WiredTiger}, a popular B-Tree-based key-value store, in
Section~\ref{sec:evaluation_dynamic}.
\new{We implement and open source an extensible benchmarking suite
to conduct our experiments for easy reproducibility.}

\begin{figure*}
    \centering
    \includegraphics[width=0.85\textwidth]{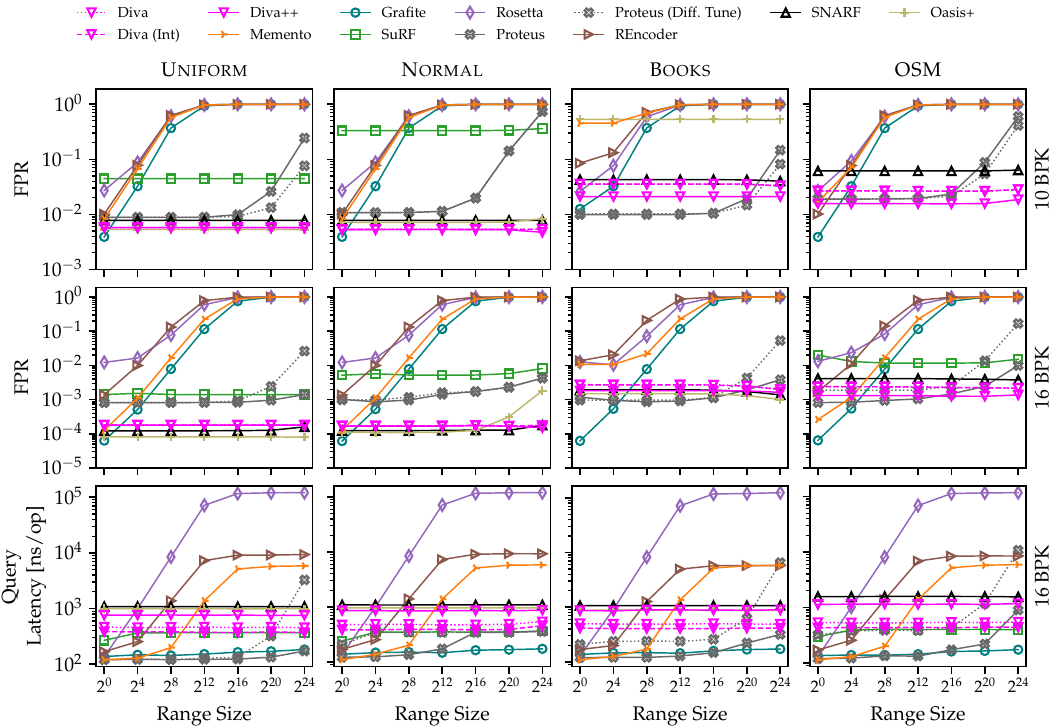}
    \caption{\new{Diva and Diva++ provide} the best balance between
    FPR and query latency for any workload across range query sizes.}
    \label{fig:experiment_1_fpr}
\end{figure*}

\textbf{Platform.}
We run experiments on a Fedora 39 machine with an Intel Xeon w7-2495X processor
(4.8 GHz) with 24 cores and 48 hyper-threads. Our machine has an 80 KB L1 cache
and a 2 MB L2 cache for each core, a 45 MB shared L3 cache, and 64 GB of main
memory. It also has two SK Hynix 512 GB PC611 M.2 2280 80mm SSDs, which are
used in the end-to-end experiments only.

\subsection{Static Evaluation}\label{sec:evaluation_static}
\textbf{Baselines.}
We compare \new{Diva and Diva++} to all existing range filters
except for bloomRF~\cite{bloomRF}, as its implementation is closed-source.
Because most other range filters only support integer keys, we specialize a
version of Diva for integers to enable a fair comparison while still
benchmarking the general-purpose version of Diva that supports variable-length
keys. We implement these versions in \texttt{C++}. All other
baselines we compare to are implemented in \texttt{C/C++} as well. We compile
all filters using \texttt{gcc-13}.

\textbf{Integer Datasets.}
Following prior
work~\cite{SuRF,Rosetta,REncoder,REncoder_Journal,bloomRF,Proteus,SNARF,Oasis,Grafite,Memento},
we employ the following synthetic and real-world~\cite{SOSD-vldb,SOSD-neurips}
datasets for our evaluation over integer data:
\begin{itemize}
    \item \textsc{Uniform}: 200M uniformly sampled 64-bit integers.
    \item \textsc{Normal}: 200M 64-bit integer samples from
        $\mathcal{N}(2^{63}, 2^{50})$.
    \item \textsc{Books}: 200M popularity scores for books on Amazon. This
        dataset is heavily skewed.
    \item \textsc{OSM}: 200M geocoordinates from the Open Street Map. It
        features several densely populated regions in the key space.
\end{itemize}

\textbf{Query Workloads.}
For integer datasets, we sample a start key~$x$ and a length $R$ to generate
the range $[x, x + R - 1]$. We vary~$R$ to showcase the effect of range lengths
on filter performance and FPR. For synthetic workloads, we choose the starting
key $x$ by sampling from the same distribution as the dataset. For real
workloads, we sample $x$ from the dataset and subsequently remove it from the
set. We generate point queries by using $R=1$.

Our workloads issue empty range queries, and we measure the FPR of each filter
by dividing the number of positive results by the size of the query batch. We
also conduct a separate experiment gauging performance for non-empty queries.
Our benchmarks focus on filter CPU times. All workloads issue a total of 10M
queries.

\hypertarget{experiment:fpr}{\textbf{Experiment 1: Integer FPR vs. Query Size Trade-off.}}
The first and second rows of Figure~\ref{fig:experiment_1_fpr}
depict the FPR of our baselines with a memory budget of 10 and 16 BPK,
respectively. \new{Here, the query size varies along the~$x$-axes.}
The bottom row of
Figure~\ref{fig:experiment_1_fpr}
compares the filters' query latencies for the experiments of the second row. We
omit query latency measurements for the first row since they are similar.
Figure~\ref{fig:experiment_1_fpr}
has both solid and dotted plots for Proteus. In the solid plots, Proteus is
tuned using a sample of the queries, while in the dotted plots, it is tuned
with range queries with uniformly distributed endpoints, simulating a workload
shift. Such workload shifts do not impact the other range filters. We tune
Rosetta and Memento filter to assume a maximum range query length of~$R=128$ to
keep their FPR for short range queries low and their memory footprint in line
with the other filters. If one keeps the memory footprint constant and tunes
Rosetta for longer ranges, its query speed deteriorates. Tuning Memento filter
for longer ranges yields faster queries at the expense of the FPR.

The FPRs of robust range filters, i.e., Rosetta, Grafite, and Memento filter,
approach one as the range query length increases. As mentioned in
Section~\ref{sec:problem_analysis}, this is due
to the filters issuing extra probes in their internal structures. This also
increases the number of cache misses in Rosetta and Memento filter, slowing
down range queries by orders of magnitude. Due to similar reasons, REncoder's
FPR and speed also rapidly deteriorate with longer queries.

While SuRF approximately matches Diva's query performance, its FPR is at least
an order of magnitude higher across experiments. \new{This is
because Diva is able to store more bits from each key at the same memory
footprint, since it learns the dataset's distribution and exploits its local
uniformity to implicitly encode many bits from an infix. In contrast, SuRF's
trie explicitly stores a prefix of each key long enough to differentiate it
from all other keys. As a result, SuRF cannot curb its memory footprint on some
datasets.} For instance, in
Figure~\ref{fig:experiment_1_fpr},
SuRF is missing from the \textsc{Books} column since it requires at least~21
BPK to create a trie that differentiates all of the dataset's keys, exceeding
the allotted memory budget.

SNARF and Oasis+ match Diva's FPR but exhibit~$\approx2\times$ slower queries
due to their slower binary searches. They also do not support variable-length
keys and dynamic operations. Moreover, in some scenarios, such as in the
\textsc{OSM} dataset, Diva achieves better FPRs than SNARF by a factor
of~$\approx2\times$. This is due to floating point errors in SNARF's model of
the data's distribution. We exclude Oasis+ from the \textsc{OSM} column since
its construction takes more than~12 hours.

Proteus' FPR is slightly lower than Diva's under the \textsc{Books} and
\textsc{OSM} datasets due to some redundancy of the keys carrying over into
Diva's \new{infixes. However, Proteus' FPR is worse than Diva by
multiple orders of magnitude under the \textsc{Uniform} and \textsc{Normal}
datasets, as the shallow SuRF instance within Proteus does not provide
significant filtering on such smooth datasets despite consuming a sizeable
amount of memory. Moreover, Proteus' dotted plots show that in the face of a
workload shift, its FPR shoots up for medium to long queries on all datasets,
resulting in orders of magnitude higher FPRs than Diva.} Although Proteus
achieves faster queries than Diva in the absence of workload shifts, it becomes
slower than Diva by an order of magnitude in the face of a workload shift as it
performs more Bloom filter probes.

\new{
While Diva++ attains the same FPR as Diva on the \textsc{Uniform} and
\textsc{Normal} datasets, it achieves a~$2\times$ lower FPR across all range
query lengths on the \textsc{Books} and \textsc{OSM} datasets. This is because
the skew in these datasets introduces a noticeable number of identical infixes
within Diva, driving up its FPR. Diva++'s bit-level full-key differentiation
counters this issue by creating \DTs. In exchange, Diva++ must parse these \DTs
and incur the associated branch-miss costs, leading to~$2\times$ slower queries
than Diva. Order-preserving entropy encoding does not significantly impact
Diva++'s FPR on these workloads, as the smooth distribution makes the keys'
bytes almost uniformly distributed and thus almost incompressible.
}

As shown in
Figure~\ref{fig:experiment_1_fpr},
any filter that outperforms \new{Diva or Diva++} in terms of FPR
has significantly slower queries and vice-versa. Hence, \new{Diva
and Diva++ provide} the best balance between FPR and query speed across the
board for variable-length range queries. Moreover, as shown by the range
size~$2^0$, \new{Diva and Diva++ are} also competitive for point
queries. As we show in subsequent experiments, \new{Diva and Diva++
are also more general-purpose than their competitors since they achieve all of
the other range filtering goals outlined in
Section~\ref{sec:introduction}.}

\hypertarget{experiment:fpr_memory}{\textbf{Experiment 2: Integer FPR vs. Memory Trade-off.}}
Figure~\ref{fig:experiment_2_fpr_memory}
depicts the FPRs of the baselines as we vary the memory budget while fixing the
query length to~$R=2^{10}$. Due to its semi-robust FPR guarantee, Diva achieves
a low FPR across the board, closely matching the best competitors. Although
Proteus, SNARF, and Oasis+ exhibit better filtering under certain datasets,
they have higher FPRs under others. Moreover, unlike Diva, none support
variable-length keys and dynamicity. Because Diva's memory footprint is
independent of~$R$ and its metadata overhead is small, it is operational under
budgets as low as 8 BPK (unlike SuRF or the robust filters).
\new{As in Experiment~\hyperlink{experiment:fpr}{1}, Diva++ attains
a similar FPR to Diva under the \text{Uniform} dataset. It also attains
a~$2\times$ lower FPR than Diva on the \textsc{OSM} dataset due to bit-level
full key differentiation better accommodating the data's skew.}

\hypertarget{experiment:fpr_string}{\textbf{Experiment 3: String FPR vs. Memory Trade-off.}}
We now experiment with datasets and workloads containing variable-length keys.
\new{
We use the following datasets:
\begin{itemize}
    \item \textsc{Normal}: 200M keys where the length of each key, in bytes, is
        chosen at random from the set $\{8, 16, 32, 64, 128, 256\}$, the first
        64 bits of each key is normally distributed~$\mathcal{N}(2^{63},
        2^{50})$, and the other bits are chosen uniformly at random. Generating
        keys in this way results in a smooth data distribution.
    \item \textsc{EnWiki}: 2.9M English words extracted from Wikipedia
        titles~\cite{enwiki}. This dataset has a moderately jagged
        distribution.
    \item \textsc{Emails}: 100K email addresses from the Enron email
        classification corpus~\cite{Enron}. We reverse the host name of each
        email address (e.g., ``com.domain@foo'') to reflect how emails are
        commonly indexed~\cite{SuRF}. The resultant hierarchical structure also
        makes the data distribution significantly jagged.
    \item \textsc{URLs}: 77M URLs extracted from the quotes listing of the
        MemeTracker dataset~\cite{MemeTracker}. The URLs within this dataset
        have deeply nested heirachies, leading to a heavily jagged data
        distribution.
\end{itemize}
We generate queries by sampling and removing pairs of adjacent keys from each
dataset.
}

\begin{figure}
    \centering
    \includegraphics[width=1.0\columnwidth]{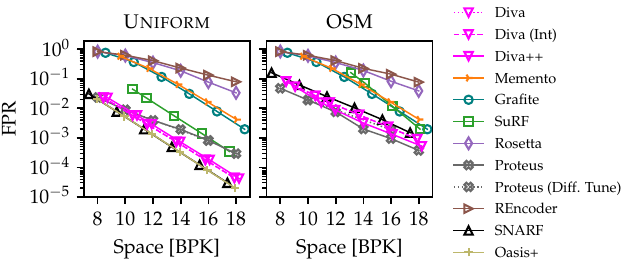}
    \caption{\new{Diva and Diva++ achieve} the best FPR scaling
    with~memory while \new{remaining operational} under stringent
    memory budgets.}
    \label{fig:experiment_2_fpr_memory}
\end{figure}

\begin{figure*}
    \centering
    \hspace*{-7mm}
    \includegraphics[width=0.9\textwidth]{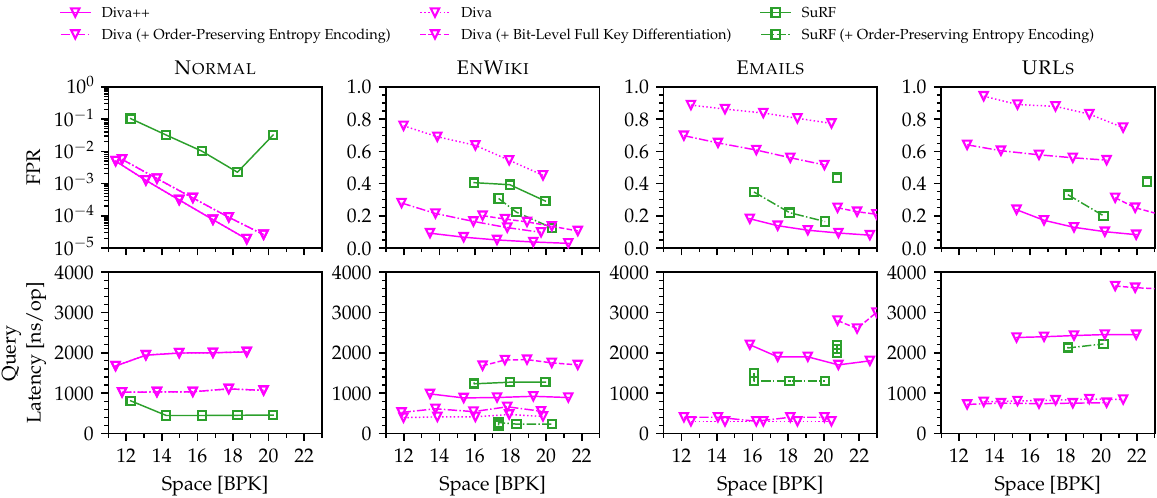}
    \caption{When operating on variable-length keys, Diva++ attains a better
    FPR than SuRF by several orders of magnitude in exchange for slower
    queries.}
    \label{fig:experiment_3_fpr_string}
\end{figure*}

\new{
Figure~\ref{fig:experiment_3_fpr_string}
evaluates Diva, Diva++, and SuRF, on the above string datasets. These filters
are the only ones that support variable-length keys. To isolate the benefits of
Diva++'s techniques, i.e., order-preserving entropy encoding
(Section~\ref{sec:order_preserving_entropy_encoding})
and bit-level full key differentiation
(Sections~\ref{sec:static_bit-level_full_key_diff}
and
\ref{sec:dynamic_bit-level_full_key_diff}),
we also benchmark versions of Diva that incorporate each technique separately.
We also illustrate how SuRF would benefit from Diva++'s techniques by including
a version of it that utilizes the order-preserving entropy encoding.
}

\new{
The top and bottom rows of Figure~\ref{fig:experiment_3_fpr_string}
plot the filters' FPR and query speed against the memory budget. On the
\textsc{Normal} dataset, which has a smooth data distribution, the original
version of Diva provides an FPR several orders of magnitude lower than both
versions of SuRF. However, as queries may terminate in
the upper levels of SuRF's trie, SuRF attains a modest~35\% query-speed advantage
over Diva. As shown in
Figure~\ref{fig:experiment_3_fpr_string},
Diva++ (and the version of Diva employing bit-level full key differentiation)
further improves Diva's FPR by as much as~$3\times$, in exchange for~$2\times$
slower queries. This is because bit-level full key differentiation
distinguishes the keys better using \DTs, which are more complex to parse than
simple infixes. Order-preserving entropy encoding contributes little on this
dataset, since the randomness of the keys makes their constituent bytes
incompressible. Trading slower queries for a lower FPR is worthwhile, since
range filters are typically used to reduce redundant accesses to slower storage
media such as disks{\textemdash}the core bottleneck of most systems.
}

\new{
On the real-world datasets \textsc{EnWiki}, \textsc{Emails}, and \textsc{URLs},
the original versions of Diva and SuRF suffer from a significantly worse FPR
of~40\% or higher. In particular, unlike Diva, SuRF is unable to differentiate
the keys using less than~20 BPK of memory when processing the \textsc{Emails}
and \textsc{URLs} datasets. These issues stem from the jagged distributions of
these datasets. Applying order-preserving entropy encoding in isolation
improves the FPR of both Diva and SuRF by as much as~$4\times$. Applying
bit-level full key differentiation to Diva in isolation improves the FPR by a
similar factor. However, it forces Diva to use at least~20 BPK of memory on the
\textsc{Emails} and \textsc{URLs} datasets, similarly to SuRF. By combining
both techniques, Diva++ attains significantly lower FPRs than all the other
baselines. In particular, Diva++ outperforms the enhanced version of SuRF by at
least~$2\times$ across the board while outperforming the original version of
SuRF by as much as an order of magnitude. These results demonstrate the
complementary nature of order-preserving entropy encoding and bit-level full
key differentiation. 
}

\hypertarget{experiment:true}{\textbf{Experiment 4: Non-Empty Query Performance.}}
Until now, we only issued empty queries in our evaluations. We now benchmark
non-empty query latency in Figure~\ref{fig:experiment_5_true}. The left
subfigure varies the memory budget while fixing the range query length
to~$R=2^{10}$, and the right subfigure varies the query size while fixing the
memory budget to 16~BPK. We only consider the \textsc{Uniform} integer dataset,
as other datasets yield similar results. Henceforth, Rosetta's hierarchy depth
and Memento filter's suffix size are tuned to the longest range query length to
enable their best performance. 

\begin{figure}
    \centering
    \includegraphics[width=0.9\columnwidth]{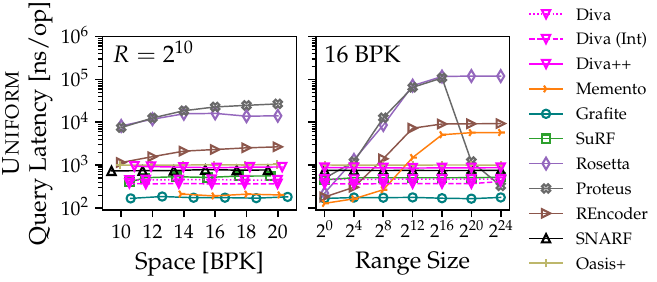}
    \vspace*{-3mm}
    \caption{Diva brings competitive performance to the table when processing
    non-empty range queries.}
    \label{fig:experiment_5_true}
    \vspace*{-3mm}
\end{figure}

Diva is faster than all range filters, except for Grafite and Memento filter,
across all memory budgets and range query lengths. \new{This speed
difference is due to other filters incurring} many cache misses whilst binary
searching or checking for existing keys in Bloom filters and tries. The closest
other range filter to Diva's speed, SuRF, is slower by as much as 35\%, as it
has to traverse all the levels in its trie down to a leaf node for a non-empty
query. Grafite and Memento filter achieve their speed by localizing similar
keys to the same memory region, forgoing support of variable-length keys and
queries, both of which are supported by Diva. Rosetta, REncoder, and Proteus
incur more cache misses as the memory budget or query length grows since they
check more bits in their Bloom filters. In contrast, \new{Diva has
a query complexity that is} independent of the memory budget and query size,
yielding consistently fast queries. \new{As in
Experiment~\hyperlink{experiment:fpr}{1}, Diva++ is slower than Diva on
non-empty queries by~$2\times$ due to the overhead of parsing \DTs.}

We have also experimented with mixed workloads, but exclude the plots, as all
filters behave similarly to Figure~\ref{fig:experiment_5_true}, except for
Proteus, which behaves similarly to the last row of
Figure~\ref{fig:experiment_1_fpr}.

\begin{figure}
    \centering
    \includegraphics[width=0.8\columnwidth]{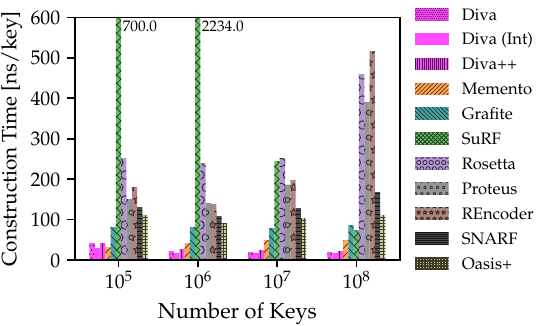}
    \caption{\new{Diva and Diva++ have} the fastest construction
    times among range filters with a large margin.}
    \label{fig:experiment_6_construction}
\end{figure}

\hypertarget{experiment:construction}{\textbf{Experiment 5: Construction Times.}}
Figure~\ref{fig:experiment_6_construction}
presents the construction times of the range filters on integer datasets of
varying sizes under a memory budget of 16 BPK. We use the \textsc{Uniform}
dataset, though the dataset choice has little influence on filter construction
times (except for Oasis+). \new{Diva and Diva++ are} significantly
faster to construct than all other range filters, with the closest range filter
being $2.7\times$ slower. This gap is due to \new{Diva and
Diva++'s} sequential memory access pattern in the input array and
\new{their} Infix Stores (which leverages the hardware prefetcher)
and due to \new{them} only hashing the sample keys in its $y$-Fast
trie, which comprise~$\approx0.1\%$ of all the keys.

\subsection{Dynamic Evaluation}\label{sec:evaluation_dynamic}
\textbf{Baselines.}
We compare \new{Diva and Diva++} to Rosetta~\cite{Rosetta},
REncoder~\cite{REncoder,REncoder_Journal}, SNARF~\cite{SNARF}, and Memento
filter~\cite{Memento}, as they are the only filters that provide insertion
APIs. Here, we implement an improved version of Memento filter that employs
Aleph filter's techniques~\cite{AlephFilter,PaghExpandability} to support
infinite expansions. \new{We then showcase how Diva and Diva++'s
performance scales with the number of threads present in the system.} In the
presence of deletions, we compare \new{Diva and Diva++} to SNARF
and Memento filter since they are the only other filters that have deletion
APIs. Finally, we integrate \new{Diva, Diva++, and Memento filter}
into WiredTiger~\cite{WiredTiger} to conduct an end-to-end evaluation on a
popular B-Tree-based key-value store. By default, the memory budget of each
filter is set to~16~BPK.

\textbf{Datasets and Workloads.}
We use the \textsc{Books} dataset with~200M keys. We construct the filters on a
randomly chosen~$\frac{1}{64}$ fraction of this dataset and insert the
remaining keys in a random order. We issue queries of length~$R=2^{10}$ (and
$R=2^{20}$ in Experiment~\hyperlink{experiment:expansion}{6}) with endpoints
from the same distribution and collect measurements as the dataset grows. Other
datasets result in similar measurements.

\hypertarget{experiment:expansion}{\textbf{Experiment 6: FPR, Space, and Insert Latency.}}
Figure~\ref{fig:experiment_7_expansion}
evaluates the baselines' FPR, space, and insert latency
as new keys are inserted. The $x$-axis shows the fraction of the dataset
ingested. The top row plots the filters' FPR, with the left and right
subfigures issuing shorter ($R=2^{10}$) and longer ($R=2^{20}$) range queries.
The left and right subfigures of the bottom row present the baselines' memory
footprint and insertion latency. We measure the memory footprint of the
baselines, as some deviate from the user-defined memory budget.

\begin{figure}
    \centering
    \includegraphics[width=0.85\columnwidth]{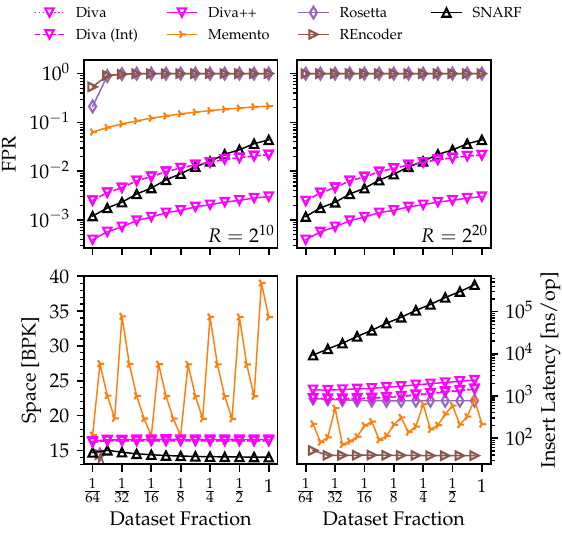}
    \caption{\new{Diva and Diva++ have} the best FPR scaling with
    insertions. \new{They are also the only efficient filters that
    have} stable memory footprints.}
    \label{fig:experiment_7_expansion}
\end{figure}

Since Rosetta and REncoder are based on Bloom filters, they are not expandable.
As such, new insertions increase the proportion of bits set to 1s in their
Bloom filters, causing their FPR to quickly converge to one, rendering them
useless as filters.

SNARF's initial FPR is better than Diva by~$\approx2.1\times$, as it does not
pay the extra cost of storing unary counters. However, its FPR deteriorates
with insertions, becoming worse than Diva by more than~$2\times$. This is
because SNARF does not update its distribution model. This is also what leads
to the slight drop in its memory footprint. In contrast, Diva maintains an
accurate distribution model by adding samples to its \ST. Moreover, SNARF's
insertion speed deteriorates with the dataset's size due to its increasing
block sizes, making it slower than Diva by as much as three orders of
magnitude. 

Since Memento filter does not support variable-length queries, its FPR is
higher than Diva by~$\approx25.3\times$ in the top-left subfigure of
Figure~\ref{fig:experiment_7_expansion}. Memento filter is missing from the
top-right subfigure since it requires more than 16~BPK of memory to support
these longer queries. Although Memento filter initially has faster inserts than
Diva, it becomes similar to Diva as the the dataset grows due to shifting more
slots in its hash table. Moreover, its memory consumption exceeds the
user-defined budget at times. This is because it expands to twice its original
size to accommodate new insertions, wasting as much as 50\% of its capacity
until it fills up again. In contrast, Diva expands its Infix Stores in small
increments, maintaining a constant memory footprint and avoiding wastage.

\new{
Diva++ attains a lower FPR than all other baselines, improving upon Diva by an
order of magnitude. Notably, this gap between Diva and Diva++ is significantly
larger than the~$2\times$ gap of Experiment~\hyperlink{experiment:fpr}{1}. This
is because, unlike Experiment~\hyperlink{experiment:fpr}{1}, where Diva builds
its \ST over the entire dataset from the onset, Diva begins this experiment
with an \ST built over only a~$\frac{1}{64}$ fraction of the keys. A smaller
initial dataset size makes the Infix Stores cover wider regions of the key
space, forcing Diva to truncate more of each key. Combined with the dataset's
skew, this truncation introduces many more identical infixes than in
Experiment~\hyperlink{experiment:fpr}{1}, inflating the FPR. Diva++ avoids this
issue altogether through bit-level full key differentiation, which retains
enough bits of each key within a \DT to differentiate it from the other keys no
matter how coarse the \ST is. In exchange, insertions into Diva++
are~$2.4\times$ slower than into Diva, as they have to parse and modify \DTs.
Nevertheless, Diva++'s insertion speed is comparable with Memento and Rosetta
in the long run, and is faster than SNARF's by \mbox{multiple orders of
magnitude}. }

\begin{figure}
    \centering
    \includegraphics[width=0.719\columnwidth]{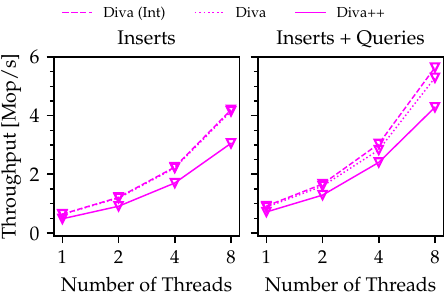}
    \vspace{1.75mm}
    \caption{\new{Diva and Diva++'s insertions and queries scale
    almost linearly with the number of threads.}}
    \label{fig:experiment_8_concurrency}
\end{figure}

In sum, \new{Diva and Diva++ maintain} the lowest FPR for range
queries of any length across expansions while supporting fast insertions.
\new{They are also the only dynamic filters that respect the memory
budget.}

\new{
\hypertarget{experiment:concurrency}{\textbf{Experiment 7: Concurrency.}}
We now measure how Diva and Diva++'s insertion and query speed scale with
multiple threads. Figure~\ref{fig:experiment_8_concurrency}
considers the same workload as the previous experiment, which inserts the keys
from the \textsc{Books} dataset in a random order and issues queries with
endpoints chosen and removed at random from the same dataset. As shown in the
left subfigure, both Diva and Diva++'s insertion speed scale almost linearly
with the number of threads in the system. This is because random insertions
following the same data distribution often fall into different Infix Stores,
resulting in low contention on the read-write locks. As the filter grows, the
number of Infix Stores increases, and the aforementioned contention drops even
further. For these reasons, the same near-linear scaling appears when the
workload mixes insertions and queries. The right subfigure of
Figure~\ref{fig:experiment_8_concurrency}
illustrates this for a workload composed~50\% insertions and~50\% queries.
}

\hypertarget{experiment:deletion}{\textbf{Experiment 8: Deletions.}}
Figure~\ref{fig:experiment_8_9_delete_wiredtiger}{\nobreakdash-}A measures the
deletion performance of the filters that provide a deletion API as the dataset
grows. It shows the latency of issuing 500K deletions each time the dataset
size doubles.
Figure~\ref{fig:experiment_8_9_delete_wiredtiger}{\nobreakdash-}A
demonstrates that Diva exhibits the fastest deletions, beating Memento filter
and SNARF by as much as $2.54\times$ and $300\times$. SNARF's deletion speed
deteriorates with insertions since its blocks grow, forcing it to rewrite
larger blocks. Memento filter also slows down since the number of slots it has
to shift increases. In contrast, Diva controls the amount of shifting and
rewriting it does by splitting Infix Stores. \new{Diva++'s
deletions are slower than Diva's by~$3\times$, as it incurs the additional
overhead of maintaining \DTs. Nevertheless, Diva++'s performance is stable, as
it employs Diva's techniques for limiting the amount of shifting within an
Infix Store.}

\hypertarget{experiment:wiredtiger}{\textbf{Experiment 9: End-to-End Query Latency.}}
We integrate \new{Diva, Diva++, and Memento filter} with WiredTiger
and measure the end-to-end empty range query latency. We exclude SNARF from
this experiment since it is unsuitable for dynamic workloads due to its
inefficient update APIs.
Figure~\ref{fig:experiment_8_9_delete_wiredtiger}{\nobreakdash-}B
measures end-to-end query latency each time the dataset size doubles. Each key
in this workload is associated with a 504-byte value, resulting in 512-byte
key-value pairs on disk. We configure WiredTiger with a buffer pool that is 1\%
of the data size on disk. We trade some of this memory for the range filter to
draw a fair comparison. The curve labeled ``Baseline'' in
Figure~\ref{fig:experiment_8_9_delete_wiredtiger}{\nobreakdash-}B
represents WiredTiger without a filter.

\new{Diva and Diva++ speed up WiredTiger's query processing by as
much as three orders of magnitude. Moreover, since they support variable-length
queries, they beat Memento filter's query latency~by~$\approx85\times$.}

\begin{figure}
    \centering
    \includegraphics[width=0.9\columnwidth]{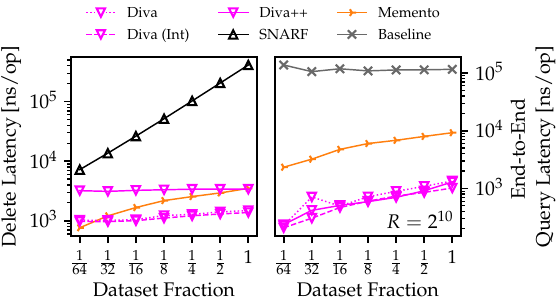}
    \begin{tikzpicture}
        \node[inner sep=1pt,align=center] (A) at (0,0) {\textbf{A)}};
        \node[inner sep=1pt,right=2.4 of A,align=center] (B) {\textbf{B)}};
        \node[inner sep=1pt,right=0.175 of B,align=center] (dummy) {};
    \end{tikzpicture}
    \vspace{-3mm}
    \caption{Diva achieves the best deletion and end-to-end query performance
    under dynamic workloads.}
    \label{fig:experiment_8_9_delete_wiredtiger}
\end{figure}
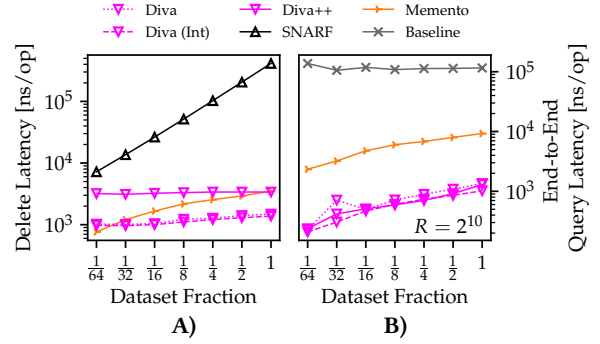

\section{Future Directions}
Diva and Diva++ open up multiple fronts for future research. We are integrating
Diva into OrcaDB, our lab’s fork of RocksDB~\cite{RocksDB}. Slow range queries
have long been a fundamental challenge for
LSM-trees~\cite{LSMTree,RocksDB,SuRF,Rosetta,REncoder,REncoder_Journal,bloomRF,Proteus,SNARF,Oasis,Grafite,Memento,Aeris},
and Diva offers a general-purpose path forward for real-world datasets and
workloads. Beyond range queries, Diva and Diva++ open new research directions
in storage engines, including adaptivity~\cite{ACF,AQF,Aeris} (proactively
mitigating recurring false positives) and efficient range
deletes~\cite{DontForgetRangeDelete} (marking entire key ranges as deleted
while preventing access to deleted keys in storage).

\begin{acks}
    We thank the reviewers for their insightful comments. This research was
    supported by the NSERC grant \#RGPIN-2023-03580.
\end{acks}

\section{Conclusion}
We introduced Diva, the first range filter to simultaneously attain the six
range filtering goals. We showed, both theoretically and empirically, that Diva
provides support for variable-length queries and keys with excellent FPR,
dynamicity, and performance.


\bibliographystyle{ACM-Reference-Format}
\bibliography{diva}

\ifarxiv

\pagebreak

\appendix
\section{Proof of Lemma~\ref{lemma:total_infixes}}\label{proof:total_infixes}
\begin{proof}
    Denote by $l$ and $r$ the first $m_{\text{infix}}$ bits of what remains of
    the predecessor and successor, respectively, after the removal of the
    $m^i_{\text{shared}}$ and $m^i_{\text{redundant}}$ bits. As infixes can be
    any value between $l$ and $r$, we analyze the difference $(r-l+1)$ to bound
    the number of potential values they can take. Consider an incremental
    process where we start with two-bit prefixes of $l$ and $r$, and increase
    their length until reaching $m_{\text{infix}}$. Initially, due to the
    removed $m^i_{\text{redundant}}$ bits, either $l=0$ and $r=2$, or $l=0$ and
    $r=3$, or $l=1$ and $r=3$. One can see that incrementing $m_{\text{infix}}$
    will turn the value $(r-l+1)$ into $2(r-l)$, $2(r-l)+1$, or $2(r-l)+2$.
    This implies that the difference $(r-l+1)$ increases to at least $2(r-l)$
    each iteration. Moreover, since $(r-l+1)$ was initially at least 3, one can
    conclude that after $k$ iterations, $(r-l+1)$ becomes at least $3\cdot2^k -
    2 \cdot (1+2+\dots+2^{k-1}) = 2^k+2$. Thus, with $m_{\text{infix}}-2=\lceil
    \log_2 \frac{T}{\epsilon} \rceil-2$ steps, the number of potential values
    becomes at least $\frac{T}{4\cdot\epsilon}+2$.

    Note that, depending on the predecessor and successor, $m_{\text{infix}} =
    \log_2 \frac{T}{\epsilon}$-bit infixes may result in as little as
    $\frac{T}{4 \cdot \epsilon}+2$ possible values, lower than what the lemma
    is trying to guarantee. Diva handles this case by incrementing
    $m_{\text{infix}}$ for only the considered Infix Store to enable
    representing at least $\frac{T}{2 \cdot \epsilon}$ different valid infixes.
\end{proof}

\section{Proof of Theorem~\ref{theorem:uniformity-new}}\label{proof:uniformity-new}
\begin{proof}
    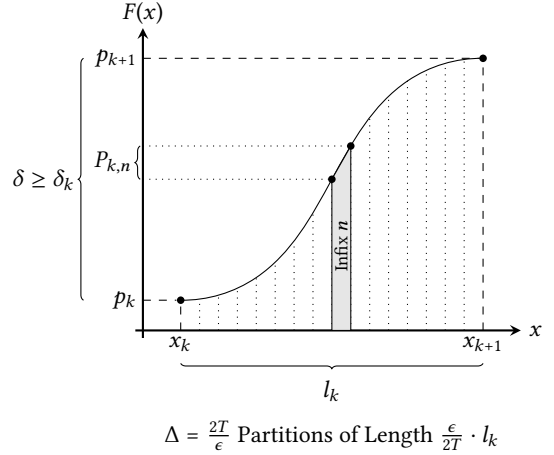
\begin{figure}
        \centering
        \begin{tikzpicture}
            \def\ox{-2}
            \def\oy{-2}
            \def\xaxislength{5}
            \def\yaxislength{4}
            \def\ooverlen{0.1}
            \def\xk{0.1}
            \def\xkk{0.9}
            \def\pk{0.1}
            \def\pkk{0.9}

            \draw[thick,-stealth] (\ox-\ooverlen,\oy) -- (\ox+\xaxislength,\oy);
            \draw[thick,-stealth] (\ox,\oy-\ooverlen) -- (\ox,\oy+\yaxislength);
            \node[inner sep=1pt,circle,fill=black,align=center] (start) at (\ox+\xk*\xaxislength,\oy+\pk*\yaxislength) {};
            \node[inner sep=1pt,circle,fill=black,align=center] (end) at (\ox+\xkk*\xaxislength,\oy+\pkk*\yaxislength) {};
            \draw plot[hobby] coordinates { (start) 
            (\ox+\xk*\xaxislength+0.2*\xaxislength,\oy+\pk*\yaxislength+0.075*\yaxislength) 
            (\ox+\xkk*\xaxislength-0.2*\xaxislength,\oy+\pkk*\yaxislength-0.075*\yaxislength) 
            (end) };
            \node[inner sep=0pt,align=center] (x_label) at (\ox+\xaxislength+2*\ooverlen,\oy) {$x$};
            \node[inner sep=0pt,align=center] (y_label) at (\ox,\oy+\yaxislength+2*\ooverlen) {$F(x)$};

            \draw[dashed] (\ox+\xk*\xaxislength,\oy) -- (start);
            \node[inner sep=2pt,anchor=north] (x_k) at (\ox+\xk*\xaxislength,\oy) {$x_k$};
            \draw[dashed] (\ox+\xkk*\xaxislength,\oy) -- (end);
            \node[inner sep=2pt,anchor=north] (x_k1) at (\ox+\xkk*\xaxislength,\oy) {$x_{k+1}$};

            \draw[dashed] (\ox,\oy+\pk*\yaxislength) -- (start);
            \node[inner sep=2pt,anchor=east] (p_k) at (\ox,\oy+\pk*\yaxislength) {$p_k$};
            \draw[dashed] (\ox,\oy+\pkk*\yaxislength) -- (end);
            \node[inner sep=2pt,anchor=east] (p_k1) at (\ox,\oy+\pkk*\yaxislength) {$p_{k+1}$};

            \draw[decorate,decoration={brace,raise=1pt,mirror,amplitude=2pt}] (\ox+\xk*\xaxislength,\oy-0.1*\yaxislength) -- (\ox+\xkk*\xaxislength,\oy-0.1*\yaxislength)
            node[pos=0.5,below=3pt,inner sep=4pt,align=center] (l_k) {$l_k$ \\[6pt] $\Delta=\frac{2T}{\epsilon}$ Partitions of Length $\frac{\epsilon}{2T} \cdot l_k$};
            \draw[decorate,decoration={brace,raise=1pt,amplitude=2pt}] (\ox-0.15*\xaxislength,\oy+\pk*\yaxislength) -- (\ox-0.15*\xaxislength,\oy+\pkk*\yaxislength) 
            node[pos=0.5,left=3pt,inner sep=1pt] (delta_k) {$\delta\geq\delta_k$};

            \draw[dotted] (\ox+0.15*\xaxislength,\oy) -- (\ox+0.15*\xaxislength,\oy+0.10*\yaxislength);
            \draw[dotted] (\ox+0.20*\xaxislength,\oy) -- (\ox+0.20*\xaxislength,\oy+0.125*\yaxislength);
            \draw[dotted] (\ox+0.25*\xaxislength,\oy) -- (\ox+0.25*\xaxislength,\oy+0.13*\yaxislength);
            \draw[dotted] (\ox+0.30*\xaxislength,\oy) -- (\ox+0.30*\xaxislength,\oy+0.175*\yaxislength);
            \draw[dotted] (\ox+0.35*\xaxislength,\oy) -- (\ox+0.35*\xaxislength,\oy+0.225*\yaxislength);
            \draw[dotted] (\ox+0.40*\xaxislength,\oy) -- (\ox+0.40*\xaxislength,\oy+0.295*\yaxislength);
            \draw[dotted] (\ox+0.45*\xaxislength,\oy) -- (\ox+0.45*\xaxislength,\oy+0.39*\yaxislength);
            \draw[dotted] (\ox+0.50*\xaxislength,\oy) -- (\ox+0.50*\xaxislength,\oy+0.50*\yaxislength);
            \draw[dotted] (\ox+0.55*\xaxislength,\oy) -- (\ox+0.55*\xaxislength,\oy+0.61*\yaxislength);
            \draw[dotted] (\ox+0.60*\xaxislength,\oy) -- (\ox+0.60*\xaxislength,\oy+0.70*\yaxislength);
            \draw[dotted] (\ox+0.65*\xaxislength,\oy) -- (\ox+0.65*\xaxislength,\oy+0.78*\yaxislength);
            \draw[dotted] (\ox+0.70*\xaxislength,\oy) -- (\ox+0.70*\xaxislength,\oy+0.83*\yaxislength);
            \draw[dotted] (\ox+0.75*\xaxislength,\oy) -- (\ox+0.75*\xaxislength,\oy+0.85*\yaxislength);
            \draw[dotted] (\ox+0.80*\xaxislength,\oy) -- (\ox+0.80*\xaxislength,\oy+0.88*\yaxislength);
            \draw[dotted] (\ox+0.85*\xaxislength,\oy) -- (\ox+0.85*\xaxislength,\oy+0.89*\yaxislength);

            \draw[fill=gray!20] (\ox+0.50*\xaxislength,\oy) -- (\ox+0.55*\xaxislength,\oy) -- (\ox+0.55*\xaxislength,\oy+0.61*\yaxislength) -- (\ox+0.50*\xaxislength,\oy+0.50*\yaxislength) -- cycle;
            \node[inner sep=0pt,align=center,rotate=90] (infix_label) at (\ox+0.525*\xaxislength,\oy+0.28*\yaxislength) {\footnotesize Infix $n$};
            \node[inner sep=1pt,circle,fill=black,align=center] at (\ox+0.50*\xaxislength,\oy+0.50*\yaxislength) {};
            \node[inner sep=1pt,circle,fill=black,align=center] at (\ox+0.55*\xaxislength,\oy+0.61*\yaxislength) {};

            \draw[dotted] (\ox,\oy+0.50*\yaxislength) -- (\ox+0.50*\xaxislength,\oy+0.50*\yaxislength);
            \draw[dotted] (\ox,\oy+0.61*\yaxislength) -- (\ox+0.55*\xaxislength,\oy+0.61*\yaxislength);
            \draw[decorate,decoration={brace,raise=1pt,amplitude=2pt}] (\ox,\oy+0.50*\yaxislength) -- (\ox,\oy+0.61*\yaxislength) 
            node[pos=0.5,left=3pt,inner sep=1pt] (P_kn) {$P_{k,n}$};
        \end{tikzpicture}
        \caption{The infix values partition the range between two samples~$x_k$ and
        $x_{k+1}$ into $\frac{2T}{\epsilon}$ equi-width partitions. The shaded
        partition in the figure delineates the partition of the key space whose
        keys create an infix equal to~$n$. The change in CDF in this partition
        corresponds to~$P_{k,n}$.}
        \label{fig:infix_key_space_partitioning}
    \end{figure}

    Let~$x_1,\dots,x_{N/T}$ be Diva's samples, ordered by their CDF values. Now
    consider the~$k$-th Infix Store defined by the sample keys~$x_k$ and
    $x_{k+1}$ and let~$l_k$ denote the length of its corresponding range, i.e.,
    $l_k = x_{k+1}-x_k$. Diva partitions this range into~$\Delta =
    \frac{2T}{\epsilon}$ equal-width sub-ranges, each condensing a unique
    infix. 

    We first focus on the static case and let~$\mathcal{S}_k$ denote the set
    of~$T$ input keys that go into the~$k$-th infix store. (The dynamic case
    differs only in that there might be more than~$T$ input keys in the~$k$-th
    infix store, even though there are in expectation~$T$ of them.) Since these
    keys are independently drawn from the distribution defined by~$F$, each is
    included in~$\mathcal{S}_k$ with probability~$\delta_k=F(x_{k+1})-F(x_k)$.
    In order to map to a specific infix~$n$, the keys further have to be
    sampled in the interval~$(s_n,s_{n+1}]$, where $s_i = x_k + i \cdot
    l_k/\Delta$. We denote the probability of being sampled in this interval
    by~$P_{k,n}=F(s_{n+1})-F(s_n)$. Thus, using conditional probability, we can
    express the probability of a specific key in~$\mathcal{S}_k$ generating
    infix~$n$ as
    $$ \frac{F(s_{n+1}) - F(s_n)}{F(x_{k+1}) - F(x_k)} = \frac{P_{k,n}}{\delta_k}. $$

    In other words, given that~$x$ already is between~$x_k$ and $x_{k+1}$, we
    are interested in the probability that it lies in the specific
    interval~$(s_n, s_{n+1}]$. Fig.~\ref{fig:infix_key_space_partitioning}
    visualizes these quantities.

    Our goal is to show that~$\frac{P_{k,n}}{\delta_k} \leq \frac{2}{\Delta} =
    \frac{\epsilon}{T}$, which we then use in a union bound over all~$T$ keys
    in~$\mathcal{S}_k$ to show that the probability of seeing an infix equal
    to~$n$ in the~$k$-th Infix Store is at most~$\epsilon$. To this end, we
    would ideally like to claim that each of the sub-ranges~$(s_i,s_{i+1}]$
    have equal probability mass~$F(s_{i+1})-F(s_i)=\frac{\delta_k}{\Delta}$,
    i.e., a key is equally likely to land in any of them. If this were the
    case, we would immediately get that~$\frac{P_{k,n}}{\delta_k} \leq
    \frac{1}{\Delta} = \frac{\epsilon}{2T}$. While this is true for the uniform
    distribution (due to the fact that the sub-ranges have equal length), this
    is not true in general. Nevertheless, we will show that, if for
    well-behaved distributions, $\frac{P_{k,n}}{\delta_k}$ is not bigger
    than~$\frac{2}{\Delta}$. We do this by first bounding~$P_{k,n}$ in terms of
    the distance~$l_k = x_{k+1}-x_{k}$ and then bounding~$\delta_k$ in terms
    of~$l_k$. For this, we recall the following special case of a Taylor
    approximation:
    \begin{theorem}[Taylor's Theorem,  page 274 of\cite{CalculusAdams}]
        If the second derivative of a function~$g(\cdot)$ exists for all points
        in an interval around~$x$ containing~$x+y$, we have that, for
        some~$\theta \in [0,1]$
        $$ g(x+y) - g(x) = \frac{d}{dt} g(x) \cdot y + \frac{\frac{d^2}{dt^2} g(x+\theta y)}{2} \cdot y^2. $$
    \end{theorem}

    Intuitively, both~$P_{k,n}$ and ${\delta_k}$ can be expressed as special
    instances of the form~$F(x+y)-F(y)$ and we will use Taylor's approximation
    to bound them in terms of the distance~$y$, which ends up being
    proportional to~$l_k$. Formally, we get that:
    \begin{claim}\label{claim:eq1} If~$f' (x)\leq M_1$ for all~$x$, then:
        \begin{align*}
            \frac{P_{k,n}}{l_k/\Delta}& \leq  f\left(x_k\right) + E_1\;,  
        \end{align*}
        where~$E_1 = \frac{3}{2} M_1 l_k$.
    \end{claim}

    Before we present the proof, we note that intuitively, the~$f(x_k)$ term
    controls how far~$P_{k,n}$ is from~$l_k/\Delta$ and the~$E_1$ term should
    be thought of as a low-order error term (as long as~$M_1$ is small with
    respect to~$f(x_k)$).  For the uniform distribution in particular, we would
    have that~$f(x_k)=1$ and $E_1=0$.
    \begin{proof}  

        As mentioned, we bound~$P_{k,n}=F(s_{n+1})-F(s_n)$ by using the above
        Taylor approximation of~$F$ around $s_n$ and get that, for some~$\theta
        \in [0,1]$:
        \begin{align*}
            F(s_{n+1}) - F(s_n)& = \frac{d}{dt} F\left(s_n\right) \cdot \frac{l_k}{\Delta} + \frac{\frac{d^2}{dt^2} F\left(s_n+\theta \cdot \frac{l_k}{\Delta}\right)}{2} \cdot \left(\frac{l_k}{\Delta}\right)^2 \\
            & = f\left(s_n\right) \cdot \frac{l_k}{\Delta} + \frac{1}{2} \cdot f'\left(s_n+\theta \cdot \frac{l_k}{\Delta}\right) \cdot \left(\frac{l_k}{\Delta}\right)^2.
        \end{align*}	

        We replace the dependency on~$f(s_n)$ with a dependency on~$f(x_k)$. By
        the Mean Value Theorem, we have that there exists an~$x\in (x_k, s_n)$
        such that~$f(s_n)-f(x_k)=f'(x)\cdot(s_n-x_k)$. Since~$s_n \leq
        x_{k+1}$, we get that: 
        $$ f(s_n) - f(x_k) \leq f'(x) \cdot (x_{k+1} - x_k) \leq M_1 \cdot l_k. $$
        \noindent We now upper bound~$f'\left(s_n+\theta \cdot
        \frac{l_k}{\Delta}\right)$ by $M_1$ and get that:
        $$ \frac{P_{k,n}}{l_k/\Delta} \leq  f(x_k) + M_1 l_k + \frac{1}{2}\cdot M_1 \cdot \frac{l_k}{\Delta} \leq f(x_k) +\frac{3}{2} M_1 l_k. $$
        \noindent where in the last inequality, we used the fact that~$\Delta
        \geq 1$.
    \end{proof}
    \noindent We now bound~$\frac{l_k}{\delta_k}$ in terms of~$\delta_k$ and
    $f(x_k)$ as such:
    \begin{claim}\label{claim:eq2}  
        If~$f'(x) \leq M_1$ and $f(x) \geq M_2$ for all~$x$, then
        \begin{align*}
            \frac{l_{k}}{\delta_k} & \leq \frac{1}{f(x_k)} + E_2,
        \end{align*}
        where $E_2 = \frac{\delta_k}{2} \cdot \frac{M_1}{M_2^3}$.
    \end{claim}

    Similarly to the case of Claim~\ref{claim:eq1}, the~$\frac{1}{f(x_k)}$ is
    the dominant term in the right-hand side, while the~$E_2$  can be
    considered to be a low-order error-term as long as~$M_1$ and $M_2$ are
    bounded with respect to~$f(x_k)$.
    \begin{proof}
        We first define~$p_{k+1} = F(x_{k+1})$, $p_k = F(x_k)$ and note
        that~$\delta_k = p_{k+1}-p_k$ by definition. We write the Taylor
        approximation of~$F^{-1}$ around~$p_k$ to show that there
        exists~$\theta \in [0,1]$ such that
        \begin{align*}
            l_k & = F^{-1}(p_{k+1}) - F^{-1}(p_k) \\ 
            & = \frac{d}{dt} F^{-1}(p_k) \cdot \delta_k + \frac{1}{2} \cdot \frac{d^2}{dt^2} F^{-1}\left(p_k + \theta\delta_k\right) \cdot (\delta_k)^2.
        \end{align*}
        \noindent We simplify this expression by noticing that:
        $$ \frac{d}{dt} F^{-1}(p_k) = \frac{1}{\frac{d}{dt} F\left(F^{-1}(p_k)\right)} = \frac{1}{\frac{d}{dt} F(x_k)} = \frac{1}{f(x_k)}. $$ 
        \noindent Similarly, by denoting~$F^{-1}(p_k+\theta\delta_k) = y_k$, we
        get that:
        \begin{align*} 
            \frac{d^2}{dt^2} F^{-1}(p_k+\theta\delta_k) & = \frac{\frac{d^2}{dt^2} F\left(F^{-1}(p_k+\theta\delta_k)\right)}{\left(\frac{d}{dt} F\left(F^{-1}(p_k+\theta\delta_k)\right)\right)^3} \\
            & = \frac{\frac{d^2}{dt^2} F\left(y_k\right)}{\left(\frac{d}{dt} F\left(y_k\right)\right)^3} = \frac{f'(y_k)}{f^3(y_k)} \leq \frac{M_1}{M_2^3}\;.
        \end{align*}
        \noindent Putting everything together, we get the claim.
    \end{proof}
    \noindent Multiplying the bounds in Claim~\ref{claim:eq1} and
    Claim~\ref{claim:eq2}, we get that:
    \begin{align*}
        \frac{P_{k,n}}{\delta_k} & \leq \frac{1}{\Delta} \cdot \left( f(x_k) + E_1 \right)\cdot \left(\frac{1}{f(x_k)} + E_2\right) \\
        & = \frac{1}{\Delta} \cdot \left( 1 + f(x_k) \cdot E_2 + E_1 \cdot  \left(\frac{1}{f(x_k)} + E_2\right)\right)
    \end{align*}

    Recall that our goal is to show that~$\frac{P_{k,n}}{\delta_k} \leq
    \frac{2}{\Delta}$, which requires that the extra terms depending
    on~$f(x_k)$, $E_1$, and $E_2$ be smaller than~$1$. This motivates the
    conditions for a well-behaved data distribution: namely, that~$f(x_k)$
    cannot be too small or large, and that~$f'(x_k)$ cannot be too big.
    Specifically, we have the following:
    \begin{claim}\label{claim:conditions} 
        Let~$\delta = T \cdot \frac{\log N}{N}$ and suppose there exists a
        constant~$\alpha \in [0,1/4]$ such that:
        \begin{itemize}
            \item $\delta^{\alpha} = M_2 \leq f(x) \leq M_3=\delta^{-1+4\alpha}$,
            \item $f'(x) \leq M_1=1/3\cdot \delta^{-\alpha}$.
        \end{itemize}
        If $\delta_k \leq \delta$, then:
        \begin{align*}
            f(x_k) \cdot E_2 + E_1 \cdot \left(\frac{1}{f(x_k)} + E_2\right) &\leq 1.
        \end{align*}
    \end{claim}

    \begin{proof}
        We bound the sum one term at a time. First, we show that~$f(x_k)\cdot
        E_2 \leq 1/6$ using Claim~\ref{claim:eq2} as such:
        \begin{align*}
            f(x_k) \cdot E_2 & \leq  M_3\cdot \frac{\delta_k}{2} \cdot \frac{M_1}{(M_2)^3} \leq \delta^{-1+4\alpha} \cdot \frac{\delta}{2} \cdot \frac{\delta^{-\alpha}}{3\delta^{3\alpha}} =\frac{1}{6}.
        \end{align*}	
        \noindent We now bound the term $\frac{1}{f(x_k)} + E_2$ as such:
        \begin{align}
            \frac{1}{f(x_k)} + E_2 & = \frac{1}{f(x_k)} \cdot \left( 1 + f(x_k) \cdot E_2\right) \nonumber \\
            & \leq \delta^{-\alpha} \cdot \left(1 + \frac{1}{6} \right) = \frac{7}{6} \cdot \delta^{-\alpha} \label{eq:1}.
        \end{align}

        It then follows that~$l_k = \delta_k \cdot \left(\frac{1}{f(x_k)} +
        E_2 \right) \leq \frac{7}{6} \cdot \delta^{1-\alpha}$. We are now ready
        to bound~$E_1$ using Claim~\ref{claim:eq1}:
        \begin{align}\label{eq:2}
            E_1 & = \frac{3}{2} M_1 l_k \leq \frac{3}{2} \cdot \frac{1}{3} \cdot \delta^{-\alpha} \cdot \frac{7}{6} \cdot \delta^{1-\alpha} = \frac{7}{12} \cdot \delta^{1-2\alpha}, 
        \end{align}

        Combining Equations (\ref{eq:1}) and (\ref{eq:2}), and noticing
        that~$\delta\leq 1$ and $1-3\alpha>0$, we get
        \begin{align*}
            E_1 \cdot \left(\frac{1}{f(x_k)} + E_2\right) & \leq \frac{7}{12} \cdot \delta^{1-2\alpha} \cdot \frac{7}{6} \cdot \delta^{-\alpha}= \frac{49}{72} \cdot \delta^{1-3\alpha} \leq \frac{49}{72}.
        \end{align*}
        The error terms add up to at most~$1/6+49/72$ which is at most $1$,
        implying the main claim.
    \end{proof}

    Finally, we note that the extra condition that~$\delta_k \leq \delta$ holds
    with high probability because of the following:

    \begin{lemma}
        With high probability over the samples, it holds for all~$k$ that
        $\delta_k \leq \delta = T \cdot \frac{\log N}{N}$.
    \end{lemma}
    \begin{proof}
        Let~$x_k$ denote the~$k$-th sample key and let~$p_k=F(x_k)$. We analyze
        the static and dynamic variants of Diva separately:
        \begin{itemize}
            \item \textbf{Static:} 
                Here, one can see that~$p_k$ is the CDF values of the $kT$-th
                key in the dataset. Therefore, since the CDF values of the keys
                in the dataset are uniform samples from~$\mathcal{U}(0,1)$, the
                set of all~$p_k$ is equivalent to a set of multiple-of-$T$
                order statistics from~$N$ samples from~$\mathcal{U}(0,1)$. It
                is known that the gap between any two consecutive such samples
                from~$\mathcal{U}(0,1)$ is upper bounded by $\frac{\log N}{N}$
                with high
                probability~\cite{OrderStatisticsSpacings,OrderStatisticsIrwin,OrderStatisticsFisher}.
                Therefore, since there are $T-1$ keys in the dataset
                between~$x_k$ and $x_{k+1}$, the gap between~$p_k$
                and~$p_{k+1}$ for any~$k$ is at most~$T\cdot\frac{\log
                N}{N}=\delta$ with high probability.

            \item \textbf{Dynamic:}
                Here, the set of all~$p_k$ form a set of~$N/T$ i.i.d. samples
                from~$\mathcal{U}(0,1)$, as the~$x_k$ are i.i.d. samples from
                the dataset's distribution. Therefore, by a similar argument,
                the gap between~$p_k$ and~$p_{k+1}$ for any~$k$ is at
                most~$\frac{\log N/T}{N/T} \leq T\cdot\frac{\log N}{N}=\delta$
                with high probability. 
        \end{itemize}
    \end{proof}

    We thus get that, with high probability, the conditions in
    Claim~\ref{claim:conditions} holds and hence,~$\frac{P_{k,n}}{\delta_k}
    \leq 2/\Delta$. In the static case, there are~$T$ keys in
    the~$\mathcal{S}_k$ (i.e., the $k^{\text{th}}$ Infix Store), so we can do a
    union bound over all such~$T$ keys and get that the probability that any
    one of them creates infix~$n$ is at most 
    $$ \frac{2}{\Delta} \cdot T = \frac{\epsilon}{T} \cdot T = \epsilon. $$

    In general, the previous bound on the size of~$\mathcal{S}_k$ hold only for
    the static case. In the dynamic case, we are not guaranteed that at
    most~$T$ keys lie between~$x_k$ and $x_{k+1}$. What is, true, however, is
    that in \emph{expectation}, the number of keys in~$\mathcal{S}_k$ is $T$.
    We then get that the expected probability of seeing infix~$n$ in the~$k$-th
    store is at most~$\epsilon$. 
\end{proof}

\section{Proof of Theorem~\ref{theorem:data_oblivious_memory_lower_bound}}\label{proof:data_oblivious_memory_lower_bound}
\begin{proof}
    Our proof rests upon an encoding argument. We will show that if a range
    filter~$F$ with a data-oblivious FPR guarantee of~$\epsilon$ consumes
    significantly less than~$\log_2 \frac{1}{\epsilon}$ bits per key of memory,
    it can be used to encode a set of keys that are chosen uniformly at random
    beyond the entropy limit, which must be impossible. Our argument applies to
    both semi-robust and robust filters.

    We consider sets of~$N$ keys drawn uniformly at random from a universe of
    size~$u \gg N$. We denote such a set by~$S$. The data-oblivious FPR
    guarantee of~$F$, whether it is robust or semi-robust, must ensure that for
    a random~$S$, the probability of~$F$ returning a false positive for a given
    query is at most~$\epsilon$. That is, for any fixed query~$q$, we have 
    $$ \Pr_S[F \text{ returns a false positive for } q] \leq \epsilon. $$
    Here, the probability is over the random choice of the set~$S$. Through the
    linearity of expectation, this implies that the expected number of keys
    from the universe that~$F$ returns a false positive for is at
    most~$\epsilon \cdot u$. That is, denoting the number of positives returned
    by~$F$ on the universe's keys on a dataset~$S$ by $E(S)$, we have
    $$ \mathbb{E}_S[E(S)] \leq \epsilon \cdot u + N, $$
    where once again, the
    expectation ranges over the random choice of~$S$. Note that the addition
    of~$N$ on the right-hand side is to account for the true positives returned
    by~$F$.

    Given~$F$, we can represent the exact set~$S$ as follows. We first store
    the internal representation of~$F$. We denote the size of this
    representation, in bits, when~$F$ encodes the set~$S$, by~$s(S)$. We then
    simply store which of the~$E(S)$ positives $F$ returns are true positives.
    One can do this using an efficient code of size~$\log_2 {E(S) \choose N}$.
    Since representing a random~$S$ requires at least as much memory as its
    entropy, which is~$\log_2 {u \choose N}$, it must be that
    $$ \mathbb{E}_S \left[ s(S) + \log_2 {E(S) \choose N} \right] \geq \log_2 {u \choose N}. $$
    Using the linearity of expectation and rearranging, this implies that
    $$ \mathbb{E}_S [s(S)] \geq \log_2 {u \choose N} - \mathbb{E}_S \left[ \log_2 {E(S) \choose N} \right]. $$
    Now, using the well-known inequalities~${n \choose k} \geq (n/k)^k$ and ${n
    \choose k} \leq (en/k)^k$~\cite{MitzenmacherProbability}, we lower bound
    the right-hand side as
    \begin{align*}
        \mathbb{E}_S [s(S)] & \geq \log_2 \left(\frac{u}{N}\right)^N - \mathbb{E}_S \left[ \log_2 \left(\frac{eE(S)}{N}\right)^N \right] \\
        & = N \cdot \log_2 u - N \cdot \mathbb{E}_S[\log_2 E(S)] - O(N). 
    \end{align*}
    Applying Jensen's inequality~\cite{MitzenmacherProbability} to the
    right-hand side then yields
    \begin{align*}
        \mathbb{E}_S [s(S)] & \geq N \cdot \log_2 u - N \cdot \log_2 \mathbb{E}_S \left[ E(S) \right] - O(N) \\
        & \geq N \cdot \log_2 u - N \cdot \log_2 \mathbb{E}_S \left[ \epsilon \cdot u + N \right] - O(N) \\
        & = N \cdot \log_2 \frac{1}{\epsilon} - O(N),
    \end{align*}
    proving that on average, $F$ must use~$\log_2 \frac{1}{\epsilon} -
    O(1)$ BPK of memory.
\end{proof}

\section{Proof of Theorem~\ref{theorem:data_aware_memory_lower_bound}}\label{proof:data_aware_memory_lower_bound}
\begin{proof}
    Our proof connects the components of Goswami et al.'s
    argument~\cite{Goswami} in a slightly different way than their approach.
    Goswami et al. prove their lower bound by showing that a range filter~$F$
    with a data-aware FPR guarantee can be used to encode a set of keys from a
    universe of size~$u \gg N, R$. Specifically, they restrict their attention
    to encoding sets of keys chosen from a family of \emph{$R$-Well-Separated
    Sets}. Any set~$S$ in this family satisfies the following two conditions:
    \begin{itemize}
        \item For any $x, y \in S, x \neq y$, it must be that $|x-y| \geq 2R$.
        \item $\min(S) \geq 2R-1$ and $\max(S) \leq u-2R$.
    \end{itemize}
    Goswami et al. prove that this family of sets is large:
    \begin{lemma}[Lemma~2 of \cite{Goswami}]
        There are at least~$\frac{(u-4NR)^N}{N!} \geq
        \left(\frac{u-4NR}{N}\right)^N$ many $R$-well-separated sets of keys of
        size~$N$ in a universe of size~$u$.
        \label{lemma:well_separated_set_count}
    \end{lemma}
    They show that if the filter~$F$ is ``too accurate or space-efficient,''
    the encoding will surpass the theoretical entropy limit, which is
    impossible. Specifically, they show through an encoding argument that the
    range filter~$F$ allows for representing any set of keys from this family,
    thus yielding the lower bound.

    To generalize their bound to any well-behaved dataset, as defined in
    Definition~\ref{definition:well-behaved},
    we consider sets of~$N$ keys whose keys are chosen uniformly at random.
    Doing so generalizes the lower bound to semi-robust filters, since the
    uniform distribution is trivially well-behaved.
    Lemma~\ref{lemma:well_separated_set_count}
    implies that the probability of a uniformly random set of keys being well
    separate is at least
    \begin{align*}
        \frac{\frac{(u-4NR)^N}{N!}}{{u \choose N}} & = \frac{\frac{(u-4NR)^N}{N!}}{\frac{\prod_{i=0}^{N-1} (u-i)}{N!}} = \prod_{i=0}^{N-1} \frac{u-4NR}{u-i} \\
        & = \prod_{i=0}^{N-1} \left( 1 - \frac{4NR-i}{u-i} \right) = 1 - o(1).
    \end{align*}
    Here, the~$o(1)$ term vanishes as the universe size~$u$ increases. This
    implies that when~$u \gg N,R$, with overwhelming probability, a uniformly
    distributed dataset will also be~$R$-well-separated. When this is the case,
    by Goswami et al.'s lower bound, $F$ must consume at least~$\log_2
    \frac{R^{1-O(\epsilon)}}{\epsilon} - O(1)$ BPK of memory. Therefore, for
    uniformly distributed data, $F$ must use an average memory of at
    least~$(1-o(1)) \cdot \log_2 \frac{R^{1-O(\epsilon)}}{\epsilon} - O(1)$
    BPK, which can be made arbitrarily close to the original~$\log_2
    \frac{R^{1-O(\epsilon)}}{\epsilon} - O(1)$ bound.
\end{proof}

\section{Proof of Theorem~\ref{theorem:combined_range_filtering_lower_bound}}\label{proof:combined_range_filtering_lower_bound}
\begin{proof}
    Our proof follows an encoding argument similar to that of Goswami et
    al.~\cite{Goswami}. We use a range filter~$F$ with qualities described in
    the theorem to encode a set of uniformly random keys from a universe of
    size~$u \gg N, R$. If the filter~$F$ is ``too accurate or
    space-efficient'', our encoding will surpass the theoretical entropy limit,
    resulting in a contradiction and implying the theorem.

    We consider a set~$S$ of $N$ uniformly random keys from the universe. We
    encode~$S$ by querying~$F$ to determine whether ranges coming from the
    following classes contain keys:
    \begin{enumerate}[label=(\Alph*)]
        \item All ranges of the form~$[k \cdot R', k \cdot R' + R')$ for
            some~$R' \gg R$. Since~$F$ provides a data-oblivious FPR guarantee,
            it returns at most an~$\epsilon$ fraction of false positives for
            these ranges for a random~$S$.
        \item Ranges of the form~$[k \cdot 2R, k \cdot 2R + 2R)$ overlapping
            with the actual non-empty ranges of the previous class.
        \item For each~$x \in S$ and $i=1,\dots,\lceil \log_2 R \rceil + 1$, the
            left and right \emph{Level-$i$-Covering} intervals of $x$, i.e., 
            $$ I^l_i(x) = \left[ \left\lfloor \frac{x}{2^i} \right\rfloor \cdot 2^i, \left\lfloor \frac{x}{2^i} \right\rfloor \cdot 2^i + 2^{i-1} \right), $$
            $$ I^r_i(x) = \left[ \left\lfloor \frac{x}{2^i} \right\rfloor \cdot 2^i + 2^{i-1}, \left\lfloor \frac{x}{2^i} \right\rfloor \cdot 2^i + 2^i \right). $$
            We will supplement the filter with enough additional information
            that the decoding algorithm can identify and query the above ranges
            without knowing the~$x \in S$.
    \end{enumerate}
    We denote the number of positives returned by the filter~$F$ for
    class{\nobreakdash-}(A) ranges by~$A(F)$, class-(B) ranges by $B(F)$, and
    class-(C) ranges by $C(F)$. The data-aware and data-oblivious guarantees
    of~$F$ imply that
    \begin{align*}
        & \mathbb{E}_S[A(F)] \leq \epsilon \cdot \frac{u}{R'}, \\
        & \mathbb{E}_F[B(F)] \leq \epsilon \cdot N \cdot \frac{R'}{2R} = \epsilon \cdot \frac{NR'}{2R}, \\[3pt]
        & \mathbb{E}_F[C(F)] \leq \epsilon \cdot N \lceil \log_2 R \rceil.
    \end{align*}
    Note that, in line with the data-oblivious and data-aware FPR notions, the
    expectation in the first inequality is taken over the random choice of the
    dataset~$S$, while the expectations in the other two ineqalities are taken
    over the internal randomness used within the filter~$F$. The second line
    follows from there being at most~$N$ truly non-empty class-(A) ranges, and
    the third line follows from there being~$N$ keys within $S$.

    Similarly to the proof of Goswami et al.~\cite{Goswami}, we encode the set
    of keys~$S$ by applying the following three steps:
    \begin{enumerate}
        \item We write the internal representation of~$F$.

        \item Let~$\mathcal{A}$ denote the set of class-(A) intervals that were
            deemed non-empty by the filter. We represent the set of truly
            non-empty intervals, of which there are at most~$N$, and append it
            to the encoding. This set can be represented by first encoding the
            true number of non-empty intervals~$N'$ in $\log_2 N$ bits and
            using a code of length~$\log_2 {|\mathcal{A}| \choose N'} \leq
            \log_2 {|\mathcal{A}| \choose N}$ bits to represent the set itself.

        \item Let $\mathcal{B}$ denote the set of class-(B) intervals
            overlapping with the $N$ truly non-empty class-(A) intervals that
            are also deemed non-empty by the filter. Note that this set of
            exact intervals is known from the optimal encoding constructed in
            the previous step. As before, we encode the actual set of such
            non-empty intervals by first representing the true number of
            non-empty intervals~$N''$ in at most $\log_2 N$ bits and using a
            code of length~$\log_2 {|\mathcal{B}| \choose N''} \leq \log_2
            {|\mathcal{B}| \choose N}$ bits to represent the set itself.

        \item We iterate over the intervals of~$\mathcal{B}$ from left to right
            and for each, encode the number of keys from~$S$ it contains in
            unary. The total space cost of doing so is~$2N$ bits. Then, for
            each of these non-empty intervals such as~$I^*$, we employ the
            following recursive procedure. For all~$i$ from $\lceil \log_2 R
            \rceil+1$ down to 1, we check whether~$F$ returns a positive for
            both level-$i$-covering intervals contained within the non-empty
            level-$(i+1)$-covering intervals. If not, the filter itself
            indicates that the keys algorithm must encode/decode are in either
            the left or right interval. Otherwise, we must indicate the number
            of keys that appear in each half to disambiguate the process. To
            this end, we consider two cases based on the number of keys~$k$ within
            the current interval of the recursion. Note that~$k$ is known from
            either the unary encoding from the beginning of this step or the
            information recorded in the upper levels of the recursion.
            If~$k=1$, we store a single bit to indicate whether the key in
            question is in the left or the right half interval. If~$k>1$, we
            encoding in~$\log_2 (k+1)$ bits the number of keys within the left
            interval, which also determines the number of keys contained in the
            right interval. The algorithm then recurses on each non-empty
            interval.

            Notice how the number of false positives returned by~$F$ for
            type-(C) intervals is at most the number of left and right interval
            pairs that both result in a positive. This, combined with the two
            encoding cases of the recursive procedure above, implies that this
            step adds at most~$C(F)+O(N)$ bits to the encoding.
    \end{enumerate}
    The last two steps of the encoding above are where we crucially use the
    data-aware guarantee of~$F$, since the ranges we query are based on the
    knowledge of non-empty class-(A) and class-(B) ranges.

    A decoding algorithm similar to Goswami et. al.'s~\cite{Goswami} can
    recover the set~$S$ from the above encoding by simply reversing the
    encoding. We thus analyze the size of the encoding and compare it to the
    entropy of~$S$. Denoting the size of the filter~$F$ when it encodes $S$ by
    $s(S)$, the total size of the encoding is at most
    $$ s(S) + \log_2 {A(F) \choose N} + \log_2 {B(F) \choose N} + C(F) + O(\log N). $$
    On average, this size must be at least as large as the entropy
    of a uniformly random set. Since the logarithmic term is insignificant
    compared to the other terms, we can ignore it and deduce that
    $$ \mathbb{E}_S[s(S)] + \mathbb{E}_F\left[\log_2 {A(F) \choose N}\right] + \mathbb{E}_F\left[\log_2 {B(F) \choose N}\right] + \mathbb{E}_F[C(F)] \geq \log_2 {u \choose N}. $$
    By rearranging and using the well-known inequalities~${n \choose k} \geq
    (n/k)^k$ and ${n \choose k} \leq (en/k)^k$~\cite{MitzenmacherProbability},
    we get 
    \begin{align*}
        \mathbb{E}_S[s(S)] \geq \; & \log_2 {u \choose N} - \mathbb{E}_F\left[\log_2 {A(F) \choose N}\right] - \mathbb{E}_F\left[\log_2 {B(F) \choose N}\right] \\ 
        & - \mathbb{E}_F[C(F)] \\ 
        \geq \; & N \cdot \log_2 \frac{u}{N} - N \cdot \mathbb{E}_F\left[ \log_2 \frac{eA(F)}{N} \right] \\ 
        & - N \cdot \mathbb{E}_F\left[ \log_2 \frac{eB(F)}{N} \right] - \mathbb{E}_F[C(F)] \\
        = \; & N \cdot \left( \log_2 \frac{u}{N} - \mathbb{E}_F\left[ \log_2 \frac{A(F)}{N} \right] - \mathbb{E}_F\left[ \log_2 \frac{B(F)}{N} \right] \right) \\ 
        & - \mathbb{E}_F[C(F)] - O(N).
    \end{align*}
    Applying Jensen's inequality~\cite{MitzenmacherProbability} to the
    right-hand side then yields
    \begin{align*}
        \mathbb{E}_S[s(S)] \geq \; & N \cdot \left( \log_2 \frac{u}{N} - \log_2 \frac{\mathbb{E}_F[A(F)]}{N} - \log_2 \frac{\mathbb{E}_F[B(F)]}{N} \right) \\ 
        & - \mathbb{E}_F[C(F)] - O(N) \\
        \geq \; & N \cdot \left( \log_2 \frac{u}{N} - \log_2 \frac{\epsilon \cdot u/R'}{N} - \log_2 \frac{\epsilon \cdot NR'/2R}{N} \right) \\ 
        & - \epsilon \cdot N \lceil \log_2 R \rceil - O(N) \\
        = \; & N \cdot \left( (1-2\epsilon) \cdot \log_2 R + 2 \cdot \log_2 \frac{1}{\epsilon} \right) - O(N) \\
        = \; & N \cdot \log_2 \frac{R^{1-O(\epsilon)}}{\epsilon^2} - O(N)
    \end{align*}
    Thus, on average for a random uniformly distirbuted dataset~$S$, the
    filter~$F$ must consume at least~$\log_2 \frac{R^{1-O(\epsilon)}}{\epsilon^2} -
    O(1)$ BPK of memory, thereby proving the theorem.
\end{proof}

\fi 

\end{document}